\documentclass[11pt,a4paper]{article}

\usepackage{amsmath,amsfonts,amssymb,amsthm}

\usepackage{amsthm}
\usepackage{graphicx,color}
\usepackage{boxedminipage}

\newcommand{\spn}{\textrm{\rm span}}

\DeclareMathOperator{\operatorClassP}{P}
\newcommand{\classP}{\ensuremath{\operatorClassP}}
\DeclareMathOperator{\operatorClassNP}{NP}
\newcommand{\classNP}{\ensuremath{\operatorClassNP}}

\DeclareMathOperator{\operatorClassFPT}{FPT}
\newcommand{\classFPT}{\ensuremath{\operatorClassFPT}}
\DeclareMathOperator{\operatorClassW}{W}
\newcommand{\classW}[1]{\ensuremath{\operatorClassW[#1]}}
\DeclareMathOperator{\operatorClassParaNP}{Para-NP}
\newcommand{\classParaNP}{\ensuremath{\operatorClassParaNP}}

\newcommand{\SSp}{\textsc{Space Cover}\xspace}
\newcommand{\probSteiner}{\textsc{Steiner Tree}\xspace}
\newcommand{\probSteinerF}{\textsc{Steiner Forest}\xspace}
 \DeclareMathOperator{\spanS}{span}

\newcommand{\defproblemu}[3]{
  \vspace{1mm}
\noindent\fbox{
  \begin{minipage}{0.95\textwidth}
  #1 \\
  {\bf{Input:}} #2  \\
  {\bf{Question:}} #3
  \end{minipage}
  }
  \vspace{1mm}
}

\newcommand{\defparproblem}[4]{
  \vspace{1mm}
\noindent\fbox{
  \begin{minipage}{0.96\textwidth}
  \begin{tabular*}{\textwidth}{@{\extracolsep{\fill}}lr} #1  & {\bf{Parameter:}} #3 \\ \end{tabular*}
  {\bf{Input:}} #2  \\
  {\bf{Question:}} #4
  \end{minipage}
  }
  \vspace{1mm}
}

\usepackage{vmargin}
\setmarginsrb{1in}{1in}{1in}{1in}{0pt}{0pt}{0pt}{6mm}

\newcommand{\cO}{\mathcal{O}}
\newcommand{\Oh}{\mathcal{O}}
\newcommand{\yes}{{yes}}
\newcommand{\no}{{no}}
\newcommand{\yesinstance}{\yes-instance\xspace}
\newcommand{\noinstance}{\no-instance\xspace}

\newcommand{\bran}[1]{branchable\xspace}

\usepackage[linesnumbered, boxed, algosection,noline]{algorithm2e}

\newtheorem{theorem}{Theorem}
\newtheorem{lemma}{Lemma}[section]

\newtheorem{corollary}{Corollary}
\newtheorem{definition}{Definition}[section]
\newtheorem{observation}{Observation}[section]
\newtheorem{proposition}{Proposition}[section]
\newtheorem{redrule}{Reduction Rule}[section]

\newtheorem{reduction}{Reduction Rule}[section]
\newtheorem{branchrule}{Branching Rule}[section]
\newtheorem{define}{Definition}[section]

\usepackage{framed}

\usepackage{amsthm}
\usepackage{thmtools}

\usepackage{todonotes}
 \usepackage[ pdftex, plainpages = false, pdfpagelabels, 
                 bookmarks=false,
                 bookmarksopen = true,
                 bookmarksnumbered = true,
                 breaklinks = true,
                 linktocpage,
                 pagebackref,
                 colorlinks = true,  
                 linkcolor = blue,
                 urlcolor  = blue,
                 citecolor = red,
                 anchorcolor = green,
                 hyperindex = true,
                 hyperfigures
                 ]{hyperref}

\newcommand{\simple}{{basic}}
\newcommand{\wss}{{\sc Space Cover}}
\newcommand{\rwss}{{\sc Restricted Space Cover}}
\newcommand{\brwss}{{\sc (Restricted) Space Cover}}

\newcommand{\rsfs}{{\sc Restricted Subset Feedback Set}}
\newcommand{\sfs}{{\sc Subset Feedback Set}}
\newcommand{\rr}{{\sc Rank $h$-Reduction}}

\newcommand{\mat}{$M=(E,{\cal I})$}

\newcommand{\dpair}{{$(T,{\cal M})$}}

\makeatletter
\newenvironment{subxarray}{%
  \vcenter\bgroup
  \Let@ \restore@math@cr \default@tag
  \baselineskip\fontdimen10 \scriptfont\tw@
  \advance\baselineskip\fontdimen12 \scriptfont\tw@
  \lineskip\thr@@\fontdimen8 \scriptfont\thr@@
  \lineskiplimit\lineskip
  \ialign\bgroup\hfil
    $\m@th\scriptstyle##$&$\m@th\scriptstyle{}##$\hfil\crcr
}{%
  \crcr\egroup\egroup
}
\makeatother

\usepackage{thm-restate}

\usepackage{thmtools}

\begin{document}

\title{Covering vectors by spaces: Regular matroids\thanks{The preliminary version of this paper appeared as an extended abstract in the proceedings of ICALP 2017. The research leading to these results has received funding from the European Research Council under the European Union's Seventh Framework Programme (FP/2007-2013) / ERC Grant Agreement n. 267959 and the Research Council of Norway via the projects ``CLASSIS'' and ``MULTIVAL''.}}

\author{
Fedor V. Fomin\thanks{
Department of Informatics, University of Bergen, Norway.} \addtocounter{footnote}{-1}
\and
Petr A. Golovach\footnotemark{}\addtocounter{footnote}{-1}
\and 
Daniel Lokshtanov\footnotemark{}\addtocounter{footnote}{-1}
\and
Saket Saurabh\footnotemark{} \thanks{Institute of Mathematical Sciences, Chennai, India}}

\date{}

\maketitle

\begin{abstract}
Seymour's decomposition theorem for regular matroids is a fundamental   result with a number of combinatorial and algorithmic applications. In this work we demonstrate  how this theorem can be used  in the design of  parameterized algorithms on regular matroids. We consider the problem of covering a set of vectors  of a given finite dimensional linear space (vector space) by a subspace generated by a set of vectors of minimum size. Specifically, in the  \SSp problem,  we are given a matrix $M$  and a subset of its columns  $T$; the task is to find a minimum set   $F$ of columns of $M$ disjoint with $T$ such that   that the linear span of   $F$ contains all vectors of $T$. For graphic matroids this problem is essentially  \probSteinerF and  for cographic matroids this  is a generalization of \textsc{Multiway Cut}.  

 Our main result is the   algorithm with running time $2^{\Oh(k)}\cdot  ||M|| ^{\Oh(1)}$ solving \SSp  in  the case when $M$ is a totally unimodular matrix over rationals,  where $k$ is the size of $F$.   In other words, we show that on regular matroids the problem is fixed-parameter tractable parameterized by  the rank of the covering subspace. 
\end{abstract}

\section{Introduction}
We consider the  problem of covering a subspace  of a given finite dimensional linear space (vector space) by a set of vectors of minimum size.
The input of the problem is   a matrix $M$ given together with a function $w$ assigning a nonnegative weight to each column of $M$ and a set $T$  of terminal  column-vectors $T$  of $M$. The task is 
 to find a minimum set of column-vectors $F$ of $M$ (if such  a set exists) which is disjoint with $T$ and 
  generates a subspace containing the linear space generated by $T$. In other words, 
  $T \subseteq \spanS (F)$, where $ \spanS (F)$ is  the linear span of $F$. 
We refer to this problem as the \SSp {} problem.

The  \SSp problem encompasses various problems arising in different domains. 
 The \textsc{Minimum Distance} problem in coding theory asks for a minimum dependent set of columns in
a matrix  over GF(2). This problem can be reduced to \SSp {} by finding for each column $t$ in   matrix $M$  a minimum set of columns in the remaining part of the matrix that cover $T=\{t\}$.  The  complexity of this problem was asked by Berlekamp et al. 
 \cite{Berlekamp78} and  remained open for almost 30 years. It was resolved only in 1997,  
  when Vardy showed it to be NP-complete~\cite{Vardy97}.  The parameterized version of  the 
  \textsc{Minimum Distance} problem, namely \textsc{Even Set},  asks whether there is a dependent set $F\subseteq X$ of  size at most $k$. The parameterized complexity of  \textsc{Even Set} is a long-standing open question in the area, see e.g.  \cite{DowneyF13}. In the language of matroid theory,  the problem of finding a minimum dependent set is known as \textsc{Matroid Girth}, i.e. the problem of finding a circuit in matroid of minimum length   \cite{MR0278975}. 
In machine learning this problem is known as the \textsc{Subspace Recovery} problem \cite{HardtM13}.
   This problem also  generalizes the problem of computing the rank of a tensor.

 For our purposes, it is convenient to rephrase the definition  of the  \SSp problem in the language of matroids.
Given a matrix $N$, let \mat{} denote the matroid where the ground set $E$ corresponds to the 
columns of $N$ and  $\cal I$ denote the family of subsets of linearly independent columns. This matroid is called the vector matroid corresponding to matrix  $N$.  Then for  matroids,     finding a subspace covering $T$ corresponds to finding   $F\subseteq E\setminus T$, $F\in \cal I$,  such that $|F|\leq k$ and $T$ is spanned by $F$. Let us remind that in a matroid  set $F$ spans $T$, denoted by  $T\subseteq \spn(F)$, if $r(F)=r(T\cup F)$.
Here 
 $r\colon 2^E\rightarrow\mathbb{N}_0$  is the rank function of $M$. (We use  $\mathbb{N}_0$ to denote the set of nonnegative integers.)

Then  \SSp  is defined as follows. 

\medskip

\defparproblem{\SSp}%
{A binary matroid \mat{} given together with its matrix representation over GF(2), a weight function $w\colon E\rightarrow \mathbb{N}_0$,
a set of \emph{terminals} $T\subseteq E$, and a nonnegative integer $k$.}%
{$k$}
{Is there a set $F\subseteq E\setminus T$ with $w(F)\leq k$ such that $T\subseteq \spn(F)$?}

\medskip

\noindent\smallskip
Since a representation of a binary matroid is given as a part of the input, we always assume that
the \emph{size} of $M$ is $||M||=|E|$. 
For regular matroids, testing matroid regularity can be done efficiently, see e.g.~\cite{Truemper92}, and when the input binary matroid is regular,  the requirement that the matroid is given together with its representation can be omitted.

It is known (see, e.g.,~\cite{KhachiyanBEGM05}) that   \SSp  on special classes of binary matroids, namely graphic and cographic matroids,  
 generalizes    
two well-studied optimization problems on graphs, namely  \probSteiner and  \textsc{Multiway Cut}.  Both problems play fundamental roles in parameterized algorithms.

 Recall that in the \probSteinerF problem we are given a (multi) graph $G$,  a weight function $w\colon E\rightarrow \mathbb{N}$,
a collection of pairs of distinct vertices   $\{x_1,y_1\},\ldots,\{x_r,y_r\}$ of $G$, and a nonnegative integer $k$.
The task is to decide whether there is 
  a set $F\subseteq E(G)$ with $w(F)\leq k$ such that for each  $i\in\{1,\ldots,r\}$, graph $G[F]$ contains an $(x_i,y_i)$-path.
To see that \probSteinerF is a special case of  \SSp, for instance $(G,w,\{x_1,y_1\},\ldots,\{x_r,y_r\},k)$ of \probSteinerF, we construct the following graph. For each $i\in\{1,\ldots,r\}$, we add a new edge $x_iy_i$ to $G$
 and assign an arbitrary weight to it; notice that we can create multiple edges this way. Denote by $G'$ the obtained mulitigraph and let $T$ be the set of added edges and let $M(G')$ be the graphic matroid associated with $G'$. Then a set of edges $F\subseteq E(G)$  forms a graph containing all $(x_i,y_i)$-paths if an only if $T\subseteq \spn(F)$
 in $M(G')$.

The special case of \probSteinerF when  $x_1=x_2=\cdots=x_r$, i.e. when set $F$ should form a connected subgraph spanning  all demand vertices, is  the \probSteiner problem, the  fundamental problem in network optimization.
   By the classical result of Dreyfus and Wagner \cite{DreyfusW71}, \probSteiner is 
   fixed-parameter tractable (FPT) parameterized by the  number of terminals. 
 The study of parameterized algorithms for \probSteiner has led to the design of important techniques, such as Fast Subset Convolution~\cite{BjorklundHKK07} and the use of branching walks~\cite{Nederlof13}. 
Research on the parameterized complexity of \probSteiner is still on-going, with  recent significant advances
for the planar version of the problem~\cite{PilipczukPSL14}.
Algorithms for \probSteiner are frequently used as a subroutine in FPT algorithms for other problems; examples include vertex cover problems \cite{GuoNW05},  near-perfect phylogenetic tree reconstruction \cite{BlellochDHRSS06}, and connectivity augmentation problems~\cite{BasavarajuFGMRS14}.

The dual of \SSp{}, i.e., the variant of \SSp{} asking whether there is a set $F\subseteq E\setminus T$ with $w(F)\leq k$ such that $T\subseteq \spn(F)$ in the dual matroid $M^*$, is equivalent to the \rsfs{} problem. In this problem the task is  for a given   matroid $M$, a weight function $w\colon E\rightarrow\mathbb{N}_0$, a set $T\subseteq E$ and a nonnegative integer $k$, to decide whether there is a set $F\subseteq E\setminus T$ with $w(F)\leq k$ such that  matroid $M'$ obtained from $M$ by  deleting  the elements of $F$ has no circuit containing an element of $T$. Hence, \SSp{}
 for cographic matroids is equivalent to \rsfs{} for graphic matroids.   \rsfs{} for graphs was introduced  by Xiao and Nagamochi~\cite{XiaoN12}, who showed that this  problem is FPT parameterized by $|F|$. Let us note that in   
  order to obtain an algorithm for \SSp with a single-exponential dependence in $k$, we also need to design a new algorithm for \SSp on cographic matroids which 
 improves significantly the running time achieved by Xiao and Nagamochi~\cite{XiaoN12}.

\textsc{Multiway Cut}, another fundamental graph problem,  is   the special case of \rsfs{}, and therefore of \SSp. 
In the \textsc{Multiway Cut} problem we are given a (multi) graph $G$,  a weight function $w\colon E\rightarrow \mathbb{N}$,
a set $S\subseteq V(G)$, and a nonnegative integer $k$. 
The task is to decide whether there 
is  a set $F\subseteq E(G)$ with $w(F)\leq k$ such that the vertices of $S$ are in   distinct connected components of  the graph obtained from $G$ by deleting edges of $F$.
Indeed, let $(G,w,S,k)$ be an instance of \textsc{Multiway Cut}. We construct graph $G'$ by adding a new vertex $u$ and connecting it to the vertices of $S$.  Denote by $T$ the set of added edges and assign weights to them arbitrarily. 
Then $(G,w,S,k)$ is equivalent to the instance $(M(G'),w,T,k)$ of \rsfs.
If $|S| = 2$, \textsc{Multiway Cut} is exactly the classical min-cut  problem  which is solvable in  polynomial time. 
However,  as it was proved by Dahlhaus et al.~\cite{DahlhausJPSY94} already for three terminals the problem becomes NP-hard. Marx, in his celebrated work on important separators~\cite{Marx06},    has shown  that  \textsc{Multiway Cut} is FPT when parameterized by the size of the cut $|F|$.

While \probSteiner  is FPT parameterized by the number of terminal vertices, the hardness results for \textsc{Multiway Cut} with three terminals  yields that \SSp parameterized by the size of the terminal set $T$ is    \classParaNP-complete even if restricted to  cographic matroids. This explains why a meaningful parameterization of \SSp  is by the rank of the span and not the size of the terminal set.

  There is also a strong evidence that  \SSp is not tractable in its full generality on binary matroids for the following reason. 
It follows from the result of Downey et al. ~\cite{DowneyFVW99} on the
hardness of  the \textsc{Maximum-Likelihood Decoding}  problem,   that \SSp is 
  \classW{1}-hard for binary matroids when parameterized by $k$ even if restricted to the inputs with one terminal and unit-weight elements.
   However, it is still possible to establish the tractability of the problem on a large class of binary matroids. 
Sandwiched between  graphic and cographic (where the problem is FPT) and binary matroids (where the problem is intractable) is the class of regular matroids.

\medskip\noindent\textbf{Our results.}
Our main theorem establishes the tractability of \SSp on regular matroids.  
  
\begin{theorem}\label{thm:main}
\SSp on regular matroids is solvable in time $2^{\Oh(k)}\cdot ||M||^{\Oh(1)}$.
\end{theorem}
 
We believe that due to the generality of \SSp,  
 Theorem~\ref{thm:main} will be useful in the study of various optimization problems on  regular matroids. As an example, we consider  the \rr{} problem, see e.g. \cite{JoretV15}. Here we are given a binary matroid $M$ and positive integers $h$ and $k$, the task is to decide whether it is possible to decrease the rank of $M$ by at least $h$ by deleting  $k$ elements. For graphic matroids, this is 
 the \textsc{$h$-Way Cut} problem, which is for  a connected graph $G$ and positive integers $h$ and $k$,  to decide whether 
it is possible to  
  separate $G$ into at least $h$ connected components by deleting at most $k$ edges. By the celebrated result of Kawarabayashi  and Thorup~\cite{KawarabayashiT11}, 
   \textsc{$h$-Way Cut} 
is \classFPT{}   parameterized by $k$ even if $h$ is a part of the input. The result of Kawarabayashi  and Thorup cannot be extended to cographic matroids; we show that for cographic matroids the problem is \classW{1}-hard when parameterized by $h+k$. On the other hand, by making use of  Theorem~\ref{thm:main}, we solve 
  \rr{}   in time $2^{\Oh(k)}\cdot ||M||^{\Oh(h)}$ on regular matroids (Theorem~\ref{prop:rr-fpt}).

   Let us also remark that the running time of our algorithm is asymptotically optimal:  unless Exponential Time Hypothesis fails,  there is no algorithm of running time  $2^{o(k)}\cdot ||M||^{\Oh(1)}$
  solving \SSp on graphic (\probSteiner) or cographic (\textsc{Multiway Cut}) matroids, see e.g.  \cite{CyganFKLMPPS15}.

\smallskip
\noindent
{\bf Related work.} The main building block of our algorithm is the fundamental theorem of  Seymour~\cite{Seymour80a} on a decomposition of regular matroids. 
Roughly speaking (we  define it properly in Section~\ref{sec:decomp}), 
the Seymour's decomposition provides a way to decompose a regular matroid into  much simpler \emph{base} matroids that are graphic, cographic  or have a constant size in such way that all ``communication'' between base matroids is limited to ``cuts'' of small rank
(we
refer to the monograph of Truemper~\cite{Truemper92} 
and the survey of Seymour~\cite{Seymour95}
for the 
introduction to matroid decompositions). 
This theorem has a number of important combinatorial and algorithmic applications.   
Among the classic algorithmic applications of Seymour's decomposition are the  
  polynomial time algorithms of  Truemper~\cite{Truemper87} (see also~\cite{Truemper92}) for finding  maximum flows and shortest routes and the polynomial algorithm of  Golynski and Horton~\cite{GolynskiH02} for constructing a minimum cycle basis. 
More recent applications of  Seymour's decomposition can be found in approximation, on-line and parameterized algorithms. 
Godberg and Jerrum \cite{GoldbergJ13} used Seymour's decomposition theorem for obtaining a  fully polynomial randomized approximation scheme (FPRAS) for the partition function
    of the ferromagnetic Ising model on  regular matroids. Dinitz and Kortsarz in~\cite{DinitzK14} applied  the decomposition theorem for  
    the \textsc{Matroid Secretary} problem. 
In~\cite{GavenciakKO12}, Gavenciak,  Kr{\'{a}}l and Oum initiated the study of  the \textsc{Minimum Spanning Circuit} problem for matroids that generalizes the classical \textsc{Cycle Through Elements} problem for graphs. The problem asks for a matroid $M$, a set $T\subseteq E$ and a nonnegative integer $\ell$, whether there is a circuit $C$ of $M$ with $T\subseteq C$ of size at most $\ell$. Gavenciak,  Kr{\'{a}}l and Oum~\cite{GavenciakKO12} proved that the problem is  \classFPT{} when parameterized by $\ell$ if $|T|\leq 2$. Very recently, in~\cite{FominGLS17}, we extended this result by showing that \textsc{Minimum Spanning Circuit} is \classFPT{} when parameterized by $k=\ell-|T|$.

On a very superficial level, all the algorithmic approaches based on the Seymour's decomposition theorem
utilize the same idea: solve the problem  on base matroids and then ``glue''  solutions into a global solution. However, such a view is a strong  oversimplification. 
First of all, the original decomposition of Seymour in~\cite{Seymour80a} was not meant for algorithmic purposes and almost every time to use it algorithmically one has to apply  nontrivial adjustments to the original decomposition.   For example, in order to solve \textsc{Matroid Secretary} on regular matroids, 
Dinitz and Kortsarz in~\cite{DinitzK14} had to give a refined decomposition theorem suitable for their algorithmic needs. Similarly, in order to use the decomposition theorem for approximation algorithms, Goldberg and Jerrum in \cite{GoldbergJ13}  had to add several new ingredients to the original Seymour's construction. We face exactly the same nature of difficulties in using Seymour's decomposition theorem. Our starting point is the variant of 
 Seymour's decomposition theorem proved by  Dinitz and Kortsarz in~\cite{DinitzK14}.  However, even the decomposition of Dinitz and Korsatz cannot be used as a black box for our purposes. 
  Our algorithm, while recursively constructing a solution has to ``dynamically'' transform the decomposition. This  occurs when the algorithm processes  cographic matroids ``glued'' with other matroids and for that part of the algorithm the transformation of the decomposition is essential.  

\section{Organization of the paper and outline of  the algorithm}\label{sec:descr-alg}
In this section we explain the structure of our paper and
give a high level overview of our algorithm.  In Section~\ref{sec:defs} we give basic definition and prove some simple auxiliary results.  
 One of the crucial components of our algorithm is  the classical theorem of   Seymour~\cite{Seymour80a} on a decomposition of regular matroids and in Section~\ref{sec:decomp} we briefly introduce these structural results.
Roughly speaking, the theorem of Seymour says that every regular matroid can be decomposed via ``small sums" into basic matroids which are graphic, cographic and very special matroid of constant size called $R_{10}$. 
The   very general description of our  approach is:  First solve \SSp on basic matroids, second  move through matroid decomposition and combine solutions from basic matroids. However when it comes to the implementation of this approach, many difficulties arise.  In what follows we give an overview of our algorithm.

To describe the decomposition of matroids, we need the notion of ``$\ell$-sums''  of matroids; we refer to~\cite{Oxley11,Truemper92} for a formal introduction to matroid sums. 
However, for our purpose, it is sufficient that we restrict ourselves to binary matroids and up to $3$-sums~\cite{Seymour80a}. 
For two binary matroids $M_1$ and $M_2$, the \emph{sum} of $M_1$ and $M_2$, denoted by $M_1\bigtriangleup M_2$, is the matroid $M$ with the ground set $E(M_1)\bigtriangleup E(M_2)$ whose cycles are all subsets  $C\subseteq E(M_1)\bigtriangleup E(M_2)$ of the form $C=C_1\bigtriangleup C_2$, where $C_1$ is a cycle of $M_1$ and  $C_2$ is a cycle of $M_2$.
\begin{itemize}
\item[i)] If $E(M_1)\cap E(M_2)=\emptyset$ and  $E(M_1), E(M_2)\neq\emptyset$, then  $M$ is the \emph{$1$-sum} of $M_1$ and $M_2$ and we write $M=M_1\oplus_1 M_2 $.
\item[ii)] If $|E(M_1)\cap E(M_2)|=1$, the unique $e\in E(M_1)\cap E(M_2)$ is not a loop or coloop of $M_1$ or $M_2$, and $|E(M_1)|,|E(M_2)|\geq 3$, then $M$ is the \emph{$2$-sum} of $M_1$ and $M_2$ and we write $M=M_1\oplus_2 M_2$.
\item[iii)] If $|E(M_1)\cap E(M_2)|=3$, the 3-element set $Z=E(M_1)\cap E(M_2)$ is a circuit of $M_1$ and $M_2$, $Z$ does not contain a cocircuit of $M_1$ or $M_2$, 
and $|E(M_1)|,|E(M_2)|\geq 7$, 
then $M$ is the \emph{$3$-sum} of $M_1$ and $M_2$ and we write $M=M_1\oplus_3 M_2$.
\end{itemize}
An \emph{$\{1,2,3\}$-decomposition} of a matroid $M$ is a collection of matroids $\mathcal{M}$, called the \emph{basic matroids} and a rooted binary tree $T$ in which $M$ is the root and the elements of $\mathcal{M}$ are the leaves such that any internal node is 1, 2 or 3-sum of its children. 

By the celebrated result of Seymour~\cite{Seymour80a},
every regular matroid $M$ has an  $\{1, 2, 3\}$-decom\-position in which every basic matroid is either graphic, cographic, or isomorphic to $R_{10}$. Moreover, such a decomposition (together with the graphs whose cycle and bond matroids are isomorphic to the corresponding basic graphic and cographic matroids) can be found in time polynomial in $|E(M)|$.
The matroid $R_{10}$  is a binary matroid represented over   ${\rm GF}(2)$ by the $5\times 10$-matrix with distinct column formed by the vectors with exactly three elements that equal 1.  

In this paper we use a variant of Seymour's decomposition suggested by Dinitz and Kortsarz in~\cite{DinitzK14}. With a regular  matroid one can associate a \emph{conflict graph}, which is an intersection graph of the basic matroids. In other words,  the nodes of the conflict  graph are the basic matroids and two nodes are adjacent if and only if the intersection of the corresponding matroids is nonempty.  It was shown by Dinitz and Kortsarz in~\cite{DinitzK14} that every regular matroid $M$ can be decomposed into basic matroids  such that the corresponding  conflict graph is a forest. Thus every node of this forest  is one of the basic matroids that are either graphic, or cographic, or isomorphic to $R_{10}$ (we can relax this condition and allow  variations of $R_{10}$ obtained by adding parallel elements to participate in a decomposition). Two nodes are adjacent if the corresponding matroids have some elements in common, the edge connecting these nodes corresponds to $2$-, or $3$-sum. We complement this forest into a \emph{conflict tree} $\mathcal{T}$ by edges which correspond to $1$-sums. As it was shown by Dinitz and Kortsarz, then  regular matroid $M$ can be obtained from  $\mathcal{T}$ by consecutive performing the sums between adjacent matroids in any order. 

In matroid language, it is much more convenient to speak in terms of minimal dependent sets, i.e. circuits. In this language,  
  a set $F\subseteq E(M)\setminus T$ spans $T\subseteq E(M)$ in   matroid $M$ if and only if for every $t\in T$, there is a circuit $C$ of $M$ such that $t\in C\subseteq F\cup\{t\}$. In what follows, we often be using an equivalent 
reformulation of \SSp, namely the problem of finding a minimum-sized set $F$, such that for every terminal element 
$t$, the set $F\cup\{t\}$ contains a circuit with $t$.

\subsection{Outline of the algorithm}
We start our algorithm with solving \SSp on basic matroids in Section~\ref{sec:basic}.  
The problem  is trivial for $R_{10}$.
 If $M$ is a graphic matroid, then there is a graph $G$ such that $M$ is isomorphic to the cycle matroid $M(G)$ of $G$.  That is,  the circuits of $M(G)$ are exactly the cycles of $G$. 
Hence, $F\subseteq E(G)$ spans $t=uv\in E(G)$ if and only if $F$ contains an $(u,v)$-path. By this observation, we can reduce an instance of \SSp to an instance of \probSteinerF. The solution to \probSteinerF is very similar to the classical algorithm for \probSteiner  \cite{DreyfusW71}. 

Recall that  \SSp{} on  cographic matroids is equivalent to \textsc{Restricted Edge-Subset Feedback Edge Set}.  
 Xiao and Nagamochi proved in~\cite{XiaoN12} that this problem can be solved in time $(12k)^{6k}2^k\cdot n^{\Oh(1)}$ on $n$-vertex graphs. 
To get a single-exponential in $k$ algorithm for regular matroids, we improve this result and construct a single-exponential algorithm for \SSp{} on  cographic matroids.
We consider a graph $G$ such that  $M$ is isomorphic to the bond matroid $M^*(G)$ of $G$.
 The set of circuits of $M$ is the set of inclusion-minimal edge cut-sets of $G$, and  we can restate \SSp as a cut problem in  $G$:
 for a given set  
$T\subseteq E(G)$, we need to find a set $F\subseteq E(G)\setminus T$ such that the  edges of $T$ are bridges of $G-F$. 
 For cut problems of  such type there is a powerful technique proposed by Marx in~\cite{Marx06}  which is based on 
 enumerating  \emph{important separators} or \emph{cuts}. However, for our purposes this technique cannot be applied directly and 
we introduce special important edge-cuts tailored for \wss{} that we call \emph{semi-important} and obtain structural results for these cuts. Then   a branching algorithm based on the enumeration of the semi-important cuts solves the problem in time 
$2^{\cO(k)}\cdot n^{\cO(1)}$ 
for $n$-vertex graphs.

The algorithm for the general case is described in Section~\ref{sec:alg}.
Suppose that we have an instance of \SSp for a regular matroid $M$. First, we apply some reduction rules described in Section~\ref{sec:reduction}
to simplify the instance. In particular, for technical reasons we allow zero weights of elements, but a nonterminal element of zero weight can be always assumed to be included in a solution. Hence, we can contract such elements. Also, if the set of terminals $T$ contains a circuit $C$, then the deletion form $M$ of any $e\in C$ leads to an equivalent instance of the problem. This way, we can bound the number of terminals in the parameter $k$.

By the next step, we construct a conflict tree $\mathcal{T}$. If $\mathcal{T}$ has one node, then $M$ is graphic, cographic or a copy of $R_{10}$ and we solve the problem directly.  
Otherwise, we select arbitrarily a root node $r$ of $\mathcal{T}$, and its selection defines the parent-child relation on $\mathcal{T}$. We say that $u$ is a \emph{sub-leaf} if its children are leaves of $\mathcal{T}$. Clearly, such a node exists and can be found in polynomial time.  Let a basic matroid $M_s$ be a sub-leaf of $\mathcal{T}$. We say that a child of $M_s$ is a $1$, $2$ or $3$\emph{-leaf} respectively if the edge between $M_s$ and the leaf corresponds to 1, 3 or 3-sum respectively.  We either reduce a leaf $M_\ell$ that is a child of $M_s$ by the deletion of $M_\ell$ from the decomposition and the modification of $M_s$ or we branch on $M_\ell$ or $M_s$. For each branch, we delete $M_\ell$ or/and modify $M_s$ in such a way that the parameter $k$ decreases.  

The case when there is an 1-leaf $M_\ell$ is trivial, because we can solve the problem for $M_\ell$ independently. For the cases of 2 and 3-leaves, we recall that a solution $F$ together with $T$ is a union of circuits and analyze the possible structure of these circuits. 

If $M_\ell$ is a 2-leaf, we have two cases: $M_\ell$ contains no terminal and $E(M_s)\cap T\neq\emptyset$. If $M_s$ has no terminal, we are able to delete $M_\ell$ from the decomposition and assign to the unique element $e\in E(M_s)\cap E(M_\ell)$ the weight that is the minimum weight of $F_\ell\subseteq E(M_\ell)\setminus\{e\}$ that spans $e$ in $M_\ell$.  If $T_\ell=E(M_\ell)\cap T\neq\emptyset$, then we have three possible  cases for $F_\ell=E(M_\ell)\cap F$, where $F$ is a (potential) solution: i) $F_\ell$ spans $T_\ell$ and $e$ in $M_\ell$, i.e., we can use the elements of $F_\ell$ that together with $e$ form a circuit of $M_\ell$ to span $t\in T\setminus T_\ell$, ii) the symmetric case, where $F_\ell\cup \{e\}$ spans $T_\ell$ and we need the elements of $F\setminus F_\ell$ that together with $e$ form a circuit to span the elements of $T_\ell$, and 
iii) $F_\ell$ spans $T_\ell$ in $M_\ell$ and no element of $F_\ell$ is needed to span the remaining terminals. Respectively, we branch according to theses cases. It can be noticed that in ii), we have a degenerate possibility that $e$ spans $T_\ell$. Then the branching does not decrease the parameter. To avoid this situation, we observe that if there is $t\in T_\ell$ that is parallel to $e$ in $M_\ell$, then we can modify the decomposition by deleting $t$ from $M_\ell$ and adding a new element $t$ to $M_\ell$ that is parallel to $e$. Now we have that for each branch, we reduce the parameter, because we have no zero-weight elements after the preprocessing.   

The analysis of the cases when we have only 3-leaves is done in similar way, but it becomes a great deal more complicated. If we have a 3-leaf $M_\ell$ that contains terminals, then we branch. Here we have 6 types of branches, and the total number of branches is 15. Moreover, for some of branches, we have to solve a special variant of the problem called \rwss{} for the leaf to break the symmetry. If there is no a 3-leaf with terminals, then our strategy depends on the type of $M_s$ that can be graphic or cographic.

 If $M_s$ is a graphic matroid, then we consider a graph $G$ such that the cycle matroid $M(G)$ is isomorphic to $M_s$ and assume that $M(G)=M_s$. If $M_\ell$ is a 3-leaf, then the elements of $E(M_s)\cap E(M_\ell)$ form a cycle $Z$ of size 3 in $G$. We delete $M_\ell$ from the decomposition and modify $G$ as follows: construct a new vertex $u$ and join $u$ with the vertices of $Z$ be edges. Then we assign the weights to the edges of $Z$ and the edges incident to $u$ to emulate all possible selections of elements of $M_\ell$ for a solution.

As with the basic matroids, the case of cographic matroids proved to be most difficult.   
If $M_s$ is cographic, then there is a graph $G$ such that the bond matroid $M^*(G)$ is isomorphic to $M_s$. Recall that the circuits of $M^*(G)$ are exactly the minimal edge cut-set of $G$. In particular, the intersections of the sets of elements of the 3-leafs with $E(M_s)$ are mapped by an isomorphism of $M_s$ and $M^*(G)$ to minimal cut-sets of $G$. We analyze the structure of these cuts. It is well-known that \emph{minimum} cut-sets of odd size form a tree-like structure (see~\cite{DinicKL76}). In our case, we can assume that $G$ has no bridges, but still $G$ is not necessarily 3 connected. 
We show that we always can find an isomorphism $\alpha$ of $M_s$ to $M^*(G)$ and a 3-leaf $M_\ell$ such that a minimal cut-set $Z=\alpha(E(M_s)\cap E(M_\ell))$ separates $G$ into two components in such a way that there is a component $H$ such that $H$ has no bridges and no element of a basic matroid $M\rq{}\neq M_s$ is mapped by $\alpha$ to an edge of $H$.
In the case of a graphic sub-leaf, we are able to get rid of a leaf by making a simple local adjustment of the corresponding graph. For the cographic case, this approach does not work as we are working with cuts. Still, if  $H$ contains no terminal, then we make a replacement but we are replacing the leaf $M_\ell$ and $H$ in $G$ simultaneously by a triangle in $G$ as follows:
 if $\{x_1y_1,x_2y_2,x_3y_3\}=Z$ and $x_1,x_2,x_3\notin V(H)$ (notice that some vertices could be the same), then delete $V(H)$ from $G$, add 3 pairwise adjacent vertices $z_1,z_2,z_3$ and construct edges $x_1z_1,x_2z_2,x_1z_3$.   Then we assign the weights to the new edges to emulate all possible selections of elements of $M_\ell$ and $E(H)$ for a solution. If $H$ has terminals, the replacement does not work and we branch on $H$ or reduce $H$. To do it, we decompose further $M^*(G)$ into a sum of two cographic matroids and obtain a new leaf of the considered sub-leaf from $H$. Then we either reduce the new leaf if it is an 1-leaf or branch on it if it is 2 or 3-leaf.  

\medskip
In Section~\ref{sec:rank-red} we discuss the application of our main result to the \rr{}  problem.

\section{Preliminaries}\label{sec:defs}
\noindent
{\bf Parameterized Complexity.}
Parameterized complexity is a two dimensional framework
for studying the computational complexity of a problem. One dimension is the input size
$n$ and another one is a parameter $k$. It is said that a problem is \emph{fixed parameter tractable} (or \classFPT), if it can be solved in time $f(k)\cdot n^{O(1)}$ for some function $f$.  We refer to the recent books of Cygan et al.~\cite{CyganFKLMPPS15} and  Downey and Fellows~\cite{DowneyF13} for  the introduction  to parameterized complexity.

It is standard for a parameterized algorithm to use \emph{(data) reduction rules}, i.e., polynomial 
or \classFPT{}
algorithms that either solve an instance or reduce it to another one that typically has a lesser input size and/or a lesser value of the parameter. A reduction rule is \emph{safe} if it either correctly solves the problem or outputs an equivalent instance.

Our algorithm for \wss{} uses the bounded search tree technique or \emph{branching}. It means that the algorithm includes steps, called \emph{branching rules}, on which we either solve the problem directly or recursively call the algorithm on several instances (\emph{branches}) for lesser values of the parameter. We say that a branching rule is \emph{exhaustive} if it either correctly solves the problem or the considered instance is a yes-instance if and only if there is a branch with a yes-instance.

\medskip
\noindent
{\bf Graphs.} 
We consider finite undirected (multi) graphs that can have  loops or multiple edges. 
We use $n$ and $m$ to denote the number of vertices and edges  of the considered graphs respectively if it does don create confusion.
For a graph $G$ and a subset $U\subseteq V(G)$ of vertices, we write $G[U]$ to denote the subgraph of $G$ induced by $U$. We write $G-U$ to denote the subgraph of $G$ induced by $V(G)\setminus U$, and $G-u$ if $U=\{u\}$.
Respectively, for $S\subseteq E(G)$, $G[S]$ denotes the graph induced by $S$, i.e., the graph with the edges $S$ whose vertices are the vertices of $G$ incident to the edges of $S$.
We denote by $G-S$ the graph obtained from $G$ by the deletion of the edges of $G$; for a single element set, we write $G-e$ instead of $G-\{e\}$.
For $e\in E(G)$, we denote by $G/e$ the graph obtained by the contraction of $e$. Since we consider multigraphs, it is assumed that if $e=uv$, then to construct $G/e$, we delete $u$ and $v$, construct a new vertex $w$, and then for each $ux\in E(G)$ and each $vx\in E(G)$, where $x\in V(G)\setminus \{u,v\}$, we construct new edge $wx$ (and possibly obtain multiple edges), and for each $e'=uv\neq e$, we add a new loop $ww$.
A set $S\subseteq E(G)$ is an \emph{(edge) cut-set} if the deletion of $S$ increases the number of components. A cut-set $S$ is \emph{(inclusion) minimal} if any proper subset of $S$ is not a cut-set. A \emph{bridge} is a cut-set of size one. 

\medskip
\noindent
{\bf Matroids.}
We refer to the book of Oxley~\cite{Oxley11} for the detailed introduction to the matroid theory.
 Recall that a matroid $M$ is a pair $(E,\mathcal{I})$, where $E$ is a finite \emph{ground} set of $M$ and $\mathcal{I}\subseteq 2^E$ is a collection of \emph{independent} sets that satisfy the following three axioms:
\begin{itemize}
\item[I1.] $\emptyset\in \mathcal{I}$,
\item[I2.] if $X\in \mathcal{I}$ and $Y\subseteq X$, then $Y\in\mathcal{I}$,
\item[I3.] if $X,Y\in \mathcal{I}$ and $|X|<|Y|$, then there is $e\in Y\setminus X$ such that $X\cup\{e\}\in \mathcal{I}$.
\end{itemize}
We denote the ground set of $M$ by $E(M)$
and 
the set of independent set by $\mathcal{I}(M)$ or simply by $E$ and  $\mathcal{I}$ if it does not create confusion. If a set $X\subseteq E$ is not independent, then $X$ is \emph{dependent}. 
Inclusion maximal independent sets are called \emph{bases} of $M$. We denote the set of bases by $\mathcal{B}(M)$ (or simply by $\mathcal{B}$).
The matroid $M^*$ with the ground set $E(M)$ such that $\mathcal{B}(M^*)=\mathcal{B}^*(M)=\{E\setminus B\mid B \in\mathcal{B}(M) \}$ is \emph{dual} to $M$. The bases of $M^*$ are \emph{cobases} of $M$.
 
A function $r\colon 2^E\rightarrow \mathbb{Z}_0$ such that for any $Y\subseteq E$, 
$r(Y)=\max\{|X|\mid X\subseteq Y\text{ and }X\in\mathcal{I}\}$ is called the \emph{rank} function of $M$. Clearly, $X\subseteq E$ is independent if and only if $r(X)=|X|$. The \emph{rank} of $M$ is $r(M)=r(E)$. Repectively, the \emph{corank} $r^*(M)=r(M^*)$. 
 
Recall that  a set $X\subseteq E$ \emph{spans} $e\in E$ if $r(X\cup \{e\})=r(X)$, and $\spn(X)=\{e\in E\mid X\text{ spans }e\}$. Respectively, $X$ \emph{spans a set} $T\subseteq E$ if 
$T\subseteq \spn(X)$.
 Let $T\subseteq E$.  Notice that if $F\subseteq T$ spans every element of $T$, then an independent set of maximum size $F'\subseteq F$ spans $T$ as well by the definition. Hence, we can observe the following. 

\begin{observation}\label{obs:spn}
Let $T\subseteq E$ for a matroid $M$, and let $F\subseteq E\setminus T$ be an inclusion minimal set such that $F$ spans $T$. Then $F$ is independent.
\end{observation}

 An (inclusion) minimal dependent set is called a \emph{circuit} of $M$. We denote the set of all circuits of $M$ by $\mathcal{C}(M)$ or simply $\mathcal{C}$ if it does not create a confusion.  
The circuits satisfy the following conditions (\emph{circuit axioms}):
\begin{itemize}
\item[C1.] $\emptyset\notin \mathcal{C}$,
\item[C2.] if $C_1,C_2\in \mathcal{C}$ and $C_1\subseteq C_2$, then $C_1=C_2$,
\item[C3.] if $C_1,C_2\in \mathcal{C}$, $C_1\neq C_2$, and $e\in C_1\cap C_2$, then there is $C_3\in \mathcal{C}$ such that $C_3\subseteq (C_1\cup C_2)\setminus\{e\}$.
\end{itemize}  
An one-element circuit is called \emph{loop}, and if $\{e_1,e_2\}$ is a two-element circuit, then it is said that $e_1$ and $e_2$ are \emph{parallel}. An element $e$ is \emph{coloop} if 
$e$ is a loop of $M^*$ or, equivalently,
$e\in B$ for every $B\in\mathcal{B}$.  A \emph{circuit} of $M^*$ is called \emph{cocircuit} of $M$. 
 A set $X\subseteq E$ is a \emph{cycle} of $M$ if $X$ either empty or $X$ is a disjoint union of circuits. By $\mathcal{S}(M)$ (or $\mathcal{S}$) we denote the set of all cycles of $M$. 
We often use the property that the sets of circuits and cycles completely define matroid. Indeed, 
a set is independent if and only if it does not contain a circuits, and the circuits are exactly inclusion minimal nonempty cycles.
 
We can observe the following.

\begin{observation}\label{obs:par}
Let $\{e_1,e_2\},C\in \mathcal{C}$ for a matroid $M$. If $e_1\in C$ and $e_2\notin C$, then $C'=(C\setminus\{e_1\})\cup\{e_2\}$ is a circuit.
\end{observation}
 
\begin{proof} 
By the axiom C3, $(\{e_1,e_2\}\cup C)\setminus\{e_1\}=(C\setminus\{e_1\})\cup\{e_2\}=C'$ contains a circuit $C''$. Suppose that $C''\neq C'$. Notice that because $C\setminus \{e_1\}$ contains no circuit, 
$e_2\in C''$. As $e_1\notin C''$, we obtain that $(\{e_1,e_2\}\cup C'')\setminus\{e_2\}$ contains a circuit, but $(\{e_1,e_2\}\cup C'')\setminus\{e_2\}$ is a proper subset of $C$; a contradiction. Hence, $C''=C'$, i.e., $C'$ is a circuit. 
\end{proof} 

It is convenient for us to express the property that a set $X$ spans an element $e$ in terms of circuits or, equivalently, cycles.

\begin{observation}\label{obs:eq}
Let $e\in E$ and $X\subseteq E\setminus\{e\}$ for a matroid $M$. Then $e\in\spn(X)$ if and only if there is a circuit (cycle) $C$ such that $e\in C\subseteq X\cup\{e\}$. 
\end{observation}

\begin{proof}
Denote by $r$ the rank function of $M$.
Let $e\in\spn(X)$. Then $r(X\cup \{e\})=r(X)$. Denote by $Y$ an independent set such that $Y\subseteq X$ and $r(X)=r(Y)$. 
We have that $r(Y\cup\{e\})\leq r(X\cup\{e\})=r(X)=r(Y)$. Hence, $Y\cup \{e\}$ is not independent and, therefore, there is a circuit (cycle) $C$ such that $C\subseteq Y\cup \{e\}\subseteq X\cup\{e\}$. Because $Y$ is independent $C\not\subseteq Y$ and $e\in C$. We obtain that $e\in C\subseteq X\cup \{e\}$.

Suppose that there is a circuit $C$ such that $e\in C\subseteq X\cup\{e\}$. Let $Y=C\cap X$. As $e\in C$ and $e\notin X$, $Y$ is a proper subset of $C$, i.e., $Y$ is independent. 
Denote by $Z$ an (inclusion) maximal independents set such that $Y\subseteq Z\subseteq X$ and let $Z'$ be a maximal independent set such that $Z'\subseteq X\cup \{e\}$. If $|Z'|>|Z|$, then by the axiom I3, there is $e'\in Z'\setminus Z$ such that $Z\cup\{e'\}$ is independent. Because $Z$ is a maximal independent set such that  $Y\subseteq Z\subseteq X$, $e'\notin X$. Hence, $e'=e$, but then $C=Y\cup\{e\}\subseteq Z\cup\{e\}$ contradicting the independence of $Z\cup\{e\}$. It means that $|Z|=|Z'|$ and, therefore,
$r(X)\leq r(X\cup\{e\})=|Z'|=|Z|\leq r(X)$. Hence, $e\in\spn(X)$.

Finally, if there is a cycle $C$ such that $e\in C\subseteq X\cup\{e\}$, then there is a circuit $C\rq{}\subseteq C$ such that $e\in C\rq{}\subseteq X\cup\{e\}$ and, therefore, $e\in\spn(X)$ by the previous case.
\end{proof}

By Observation~\ref{obs:eq}, the question of \SSp is equivalent to the following one: is there a set $F\subseteq E\setminus T$ with $w(F)\leq k$ such that for any $e\in T$, there is a circuit (or cycle) $C$ such that $e\in C\subseteq F\cup\{e\}$?

We use this equivalent definition in the majority of the proofs without referring to Observation~\ref{obs:eq}.

Let $M$ be a matroid and $e\in E(M)$ is not a loop. We say that $M'$ is obtained from $M$ by \emph{adding a parallel to $e$  element} if $E(M')=E(M)\cup\{e'\}$, where $e'$ is a new element, and 
$\mathcal{I}(M')=\mathcal{I}(M)\cup \{(X\setminus\{e\})\cup\{e'\}\mid X\in \mathcal{I}(M)\text{ and }e\in X\}$. It is straightforward to verify that $\mathcal{I}(M')$ satisfies the axioms I.1-3, i.e., $M'$ is a matroid with the ground set $E(M)\cup\{e'\}$. It is also easy to see that $\{e,e'\}$ is a circuit, that is, $e$ and $e'$ are parallel elements of $M'$. 

\medskip
\noindent
{\bf Deletions and contractions.} 
Let $M$ be a matroid, $e\in E(M)$. 
The matroid $M'=M-e$ is obtained by \emph{deleting} $e$ if $E(M')=E(M)\setminus\{e\}$ and $I(M')=\{X\in \mathcal{I}(M)\mid e\notin X\}$. 
It is said that $M'=M/e$ is obtained by \emph{contracting} $e$ if $M'=(M-e)^*$. In particular, if $e$ is not a  loop,
then $I(M')=\{X\setminus\{e\}\mid e\in X\in \mathcal{I}(M)\}$. 
Notice that deleting an element in $M$ is equivalent to contracting it in $M^*$ and vice versa.
Let $X\subseteq E(G)$. Then $M-X$ denotes the matroid obtained from $M$ by the deletion of the elements of $X$ and 
 $M/X$ is the matroid obtained by the consecutive contractions of the elements of $X$. 
The \emph{restriction} of $M$ to $X$, denoted by $M|X$, is the matroid obtained by the deletion of the elements of $E(G)\setminus X$.

\medskip
\noindent
{\bf Matroids associated with graphs.} Let $G$ be a graph. The \emph{cycle} matroid $M(G)$ has the ground set $E(G)$ and a set $X\subseteq E(G)$ is independent if $X=\emptyset$ or $G[X]$ has no cycles.  Notice that $C$ is a circuit of $M(G)$ if and only if $C$ induces a cycle of $G$.  The \emph{bond} matroid $M^*(G)$  with the ground set $E(G)$ is dual to $M(G)$, and  $X$ is a circuit of $M^*(G)$ if and only if $X$ is a minimal cut-set of $G$.  It is said that $M$ is a \emph{graphic} matroid if $M$ is isomorphic to $M(G)$ for some graph $G$. Respectively, $M$ is \emph{cographic} if 
there is graph $G$ such that  $M$ is isomorphic to $M^*(G)$.
Notice that $e\in E$ is a loop of a cycle matroid $M(G)$ if and only if $e$ is a loop of $G$, and 
$e$ is a loop of $M^*(G)$ if and only if $e$ is a bridge of $G$.
Notice also that by the addition of an element parallel to $e\in E$ for $M(G)$ we obtain $M(G')$ for the graph $G'$ obtained by adding a new edge with the same end vertices as $e$. Respectively,  by
adding of an element parallel to $e\in E$ for $M^*(G)$ we obtain $M^*(G')$ for the graph $G'$ obtained by subdividing $e$.

\medskip
\noindent
{\bf Matroid representations.}  Let $M$ be a matroid and let $F$ be a field.  An $n\times m$-matrix $A$ over $F$ is a \emph{representation of $M$ over $F$} if there is one-to-one correspondence $f$ between $E$ and the set of columns of $A$ such that for any $X\subseteq E$, $X\in \mathcal{I}$ if and only if the columns $f(X)$ are linearly independent (as vectors of $F^n$); if $M$ has such a representation, then it is said that $M$ has a \emph{representation over $F$}. In other words, $A$ is a representation of $M$ if $M$ is isomorphic to the \emph{column matroid} of $A$, i.e., the matroid whose ground set is the set of columns of $A$ and a set of columns is independent if and only if these columns are linearly independent. 
A matroid is \emph{binary} if it can be represented over ${\rm GF}(2)$. A matroid is \emph{regular} if it can be represented over any field. In particular, graphic and cographic matroids are regular.
Notice that any matroid obtained from a regular matroid by deleting and contracting its elements is regular. Observe also that the matroid obtained from a regular matroid by adding a parallel element is regular as well.

We stated in the introduction that we assume that we are given a representation  over  ${\rm GF}(2)$ of the input matroid of an instance of \SSp. Then it can be checked in polynomial time whether a subset of the ground set is independent by checking the linear independence of the corresponding columns.  

We use the following observation about cycles of binary matroids.

\begin{observation}[see~\cite{Oxley11}]\label{obs:symm}
Let $C_1$ and $C_2$ be circuits (cycles) of a binary matroid $M$. Then $C_1\bigtriangleup C_2$ is a cycle of $M$.
\end{observation}

\medskip
\noindent
{\bf The dual of \SSp.}  We recall the  definition of  \rsfs.

\defproblemu{\rsfs}{A binary matroid $M$, a weight function $w\colon E\rightarrow\mathbb{N}_0$, $T\subseteq E$, and a nonnegative integer $k$.}{Is there  a set $F\subseteq E\setminus T$ with $w(F)\leq k$ such that   matroid $M'=M-F$ has no circuit containing an element of $T$. }

This problem is   dual to  \SSp.

\begin{proposition}\label{prop:duality}
\rsfs{} on matroid $M$ is equivalent to \SSp{} on the dual of $M$.
\end{proposition}

\begin{proof}
Let $M$ be a binary matroid and $T\subseteq E$. By Observation~\ref{obs:eq}, it is sufficient to show that for every $F\subseteq E\setminus T$, $M-F$ has no circuit containing an element of $T$ if and only if
for each $t\in T$ there is a cocircuit $C$ of $M$ such that $t\in C\subseteq F\cup\{t\}$.

Suppose that for each $t\in T$, there is a cocircuit $C$ of $M$ such that $t\in C\subseteq F\cup\{t\}$. We show that $M-F$ has no circuit containing an element of $T$. To obtain a contradiction, assume that there is $t\in T$ and a circuit $C'$ of $M$ such that $t\in C'$ and $C'\cap F=\emptyset$.  
Let $C$ be a cocircuit of $M$ such that $t\in C\subseteq F\cup\{t\}$. Then $C\cap C'=\{t\}$, but it contradicts the well-known property (see~\cite{Oxley11}) that for every circuit and every cocircuit of a matroid, theirs intersection is either empty of contains at least two elements. 

Suppose now that   $M-F$ has no circuit containing an element of $T$. In particular, it means that $T$ is independent in $M$,  and hence in $M-F$. Then there is a basis $B$ of $M-F$ such that $T\subseteq B$. Clearly, $B$ is an independent set of $M$. Hence, there is a basis $B'$ of $M$ such that $B\subseteq B'$. Consider cobasis$B^*=E\setminus B'$. Let $t\in T$.
The set $B^*\cup\{t\}$ contains a unique cocircuit $C$ and $t\in C$. We claim that $C\subseteq F\cup\{t\}$.
To obtain a contradiction, assume that there is $e\in C\setminus (F\cup\{t\})\neq\emptyset$.  Since $C\cap B'=\{t\}$, $e\notin B$ and, therefore, $e\notin B'$. The set $B\cup \{e\}$ contains a unique circuit $C'$ of $M-F$ such that $e\in C'$. Notice that $C'$ is a circuit of $M$ as well. Observe that $e\in C\cap C'\subseteq\{e,t\}$. Since $C\cap C'\neq\emptyset$, $|C\cap C'|\geq 2$. Hence, $t\in C'$. We obtain that $C'$ is a circuit of $M$ containing $t$ but $C'\cap F=\emptyset$; a contradiction.
\end{proof}

The variant of \rsfs{} for graphs, i.e., \rsfs{} for graphic matroids, was introduced by  Xiao and Nagamochi in~\cite{XiaoN12}. They proved that this problem 
can be solved in time $2^{\Oh(k\log k)}\cdot n^{\Oh(1)}$ for $n$-vertex graphs. In fact, they considered the problem without weights, but their result can be generalized for weighted graphs. They also considered the unweighted variant of the problem without the restriction $F\subseteq E\setminus S$. They proved that this problem can be solved in polynomial time.  We observe that this results holds for binary matroids. More formally,  we consider the following problem.

\defproblemu{\sfs}{A binary matroid $M$, $T\subseteq E$ and a nonnegative integer $k$.}{Is there  a set $F\subseteq E\setminus T$ with $|F|\leq k$ such that the matroid $M'$ obtained from $M$ by the deletion of the elements of $F$ has no circuits containing elements of $T$. }

\begin{proposition}\label{prop:sfs}
\sfs{} is solvable in polynomial time.
\end{proposition}
\begin{proof}
To see that \sfs{} is solvable in polynomial time, it is sufficient to notice that it is dual to the similar variant of \SSp{} without weights and without the condition $F\subseteq E\setminus T$. The proof of this claim is  almost the same as the proof of Proposition~\ref{prop:duality}; the only difference is that   $F\subseteq E$ spans $T$ in $M$ if and only if for every $t\in T\setminus F$, there is a circuit $C$ such that $t\in C\subseteq F\cup\{t\}$. This variant of \SSp{} is solvable in polynomial time because the set of minimum size that spans $T$ can be chosen to be a maximal independent subset of $T$. 
\end{proof}

Notice also that if we allow weights but do not restrict $F\subseteq E\setminus T$, then this variant of \SSp{} is at least as hard as the original variant of the problem, because by assigning the weight $k+1$ to the elements of $T$ we can forbid their usage in the solution.

\medskip
\noindent
{\bf Restricted Space Cover problem.} For technical reasons, in the  algorithm we have to solve the following restricted variant of \SSp
on graphic and cographic matroids.

\defparproblem{\rwss}%
{Matroid $M$ with a ground set $E$, a weight function $w\colon E\rightarrow \mathbb{N}_0$,
a set of terminals $T\subseteq E$, a nonnegative integer $k$, and $e^*\in E$ with $w(e^*)=0$ and $t^*\in T$.}%
{$k$}
{Is there a set $F\subseteq E\setminus T$ with $w(F)\leq k$ such that $T\subseteq \spn(F)$ and $t^*\in\spn(F\setminus\{e^*\})$?}

\medskip
We conclude the section by some hardness observations.

\begin{proposition}\label{prop:W-hard}
\wss {} is \classW{1}-hard for binary matroids when parameterized by $k$ even if restricted to the inputs with one terminal and unit-weight elements.
\end{proposition}

\begin{proof}
Downey et al. proved in~\cite{DowneyFVW99} that the following parameterized problem is \classW{1}-hard:

\defparproblem{\textsc{Maximum-Likelihood Decoding}}%
{A binary $n\times m$ matrix $A$, a target binary $n$-element vector $s$, and a positive integer $k$.}%
{$k$}
{Is there a set of at most $k$ columns of $A$ that sum to $s$?}

\noindent
The \classW{1}-hardness is proved in~\cite{DowneyFVW99}  for nonzero $s$; in particular, it holds if $s$ is the vector of ones.

Let $(A,s,k)$ be an instance of \textsc{Maximum-Likelihood Decoding} for nonzero $s$. 
We define the matrix $A'$ by adding the column $s$ to $A$. Let $M$ be the column matroid of $A'$ and $T=\{s\}$. For every $e\in E(M)$, we set $w(e)=1$. 

Suppose that there are at most $k$ columns of $A$ that sum to $s$. Then there are at most $k$ linearly independent columns that sum to $s$. Clearly, these columns span $s$ in $M$. If there is a set $F\subseteq E(M)\setminus\{e\}$ of size at most $k$ that spans $e$, then  there is a circuit $C$ of $M$ such that $e\in C\subseteq F\cup\{e\}$. It immediately implies that the sum of columns of $C$ is the zero vector and, therefore, the columns of $C\setminus \{e\}$ sum to $s$.
\end{proof}

We noticed that \text{Steiner Tree} is a special case of \SSp  for the cycle matroid of an input graph.  This reduction together with the result of Dom, Lokshtanov and Saurabh~\cite{DomLS14} that \text{Steiner Tree} has no polynomial kernel (we refer to \cite{CyganFKLMPPS15} for the formal definitions of kernels) unless $\classP\subseteq {\rm coNP}/{\rm poly}$ immediately implies the following statement. 

\begin{proposition}\label{prop:ker}
\wss has no polynomial kernel unless $\classP\subseteq {\rm coNP}/{\rm poly}$ even if restricted to  graphic matroids and the inputs with unit-weight elements.
\end{proposition}

Finally,   it was proved by Dahlhaus et al.~\cite{DahlhausJPSY94} that \textsc{Multiway Cut} is \classNP-complete even if $|S|=3$. It implies as the following proposition.

\begin{proposition}\label{prop:par-term}
The version of \wss, where the parameter is $|T|$, is \classParaNP-complete even if restricted to  cographic matroids and the inputs with unit-weight elements.
\end{proposition}

\section{Regular matroid decompositions}\label{sec:decomp}
In this section we describe matroid decomposition theorems that are pivotal for algorithm for \wss.  In particular we start by giving the structural decomposition for regular matroids given by  
Seymour~\cite{Seymour80a}. Recall that, for two sets $X$ and $Y$, $X\bigtriangleup Y=(X\setminus Y)\cup (Y\setminus X)$ denotes the \emph{symmetric difference} of $X$ and $Y$. 
Recall, that  to describe the decomposition of matroids we need the notion of ``$\ell$-sums''  of matroids for $\ell=1,2,3$. 
We already defined these sums in Section~\ref{sec:descr-alg} (see also~\cite{Oxley11,Truemper92}) but, for convenience, we restate the definitions.   
Let $M_1$ and $M_2$ be binary matroids. The \emph{sum} of $M_1$ and $M_2$, denoted by $M_1\bigtriangleup M_2$, is the matroid $M$ with the ground set $E(M_1)\bigtriangleup E(M_2)$ such that the cycles of $M$ are all subsets of the form $C\subseteq E(M_1)\bigtriangleup E(M_2)$, where $C_1$ is a cycle of $M_1$ and  $C_2$ is a cycle of $M_2$.
For our purpose the following special cases of matroid sums are sufficient. 
\begin{enumerate}
\setlength{\itemsep}{-2pt}
\item  If $E(M_1)\cap E(M_2)=\emptyset$ and  $E(M_1), E(M_2)\neq\emptyset$, then  $M$ is the \emph{$1$-sum} of $M_1$ and $M_2$ and we write $M=M_1\oplus_1 M_2 $.
\item  If $|E(M_1)\cap E(M_2)|=1$, the unique $e\in E(M_1)\cap E(M_2)$ is not a loop or coloop of $M_1$ or $M_2$, and $|E(M_1)|,|E(M_2)|\geq 3$, then $M$ is the \emph{$2$-sum} of $M_1$ and $M_2$ and we write $M=M_1\oplus_2 M_2$.
\item  If $|E(M_1)\cap E(M_2)|=3$, the 3-element set $Z=E(M_1)\cap E(M_2)$ is a circuit of $M_1$ and $M_2$, $Z$ does not contain a cocircuit of $M_1$ or $M_2$, 
and $|E(M_1)|,|E(M_2)|\geq 7$, 
then $M$ is the \emph{$3$-sum} of $M_1$ and $M_2$ and we write $M=M_1\oplus_3 M_2$.
\end{enumerate}
If $M=M_1\oplus_kM_2$ for some $k\in\{1,2,3\}$, then we write $M=M_1\oplus M_2$.

\begin{definition}
A \emph{$\{1,2,3\}$-decomposition} of a matroid $M$ is a collection of matroids $\mathcal{M}$, called the \emph{basic matroids} and a rooted binary tree $T$ in which $M$ is the root and the elements of $\mathcal{M}$ are the leaves such that any internal node is either $1$-, $2$- or $3$-sum of its children. 
\end{definition}

We also need the special binary matroid $R_{10}$ to be able to define the decomposition theorem for regular matroids. It is represented over   ${\rm GF}(2)$ by the 
$5\times 10$-matrix whose columns are formed by vectors that have exactly three non-zero entries (or rather three ones) and no two columns are identical.  
Now we are ready to give the  
 decomposition theorem for regular matroids due to Seymour~\cite{Seymour80a}.

\begin{theorem}[\cite{Seymour80a}]\label{thm:decomp}
Every regular matroid $M$ has an  $\{1, 2, 3\}$-decomposition in which every basic matroid is either graphic, cographic, or isomorphic to $R_{10}$. Moreover,
such a decomposition (together with the graphs whose cycle and bond matroids are isomorphic to the corresponding basic graphic and cographic matroids) can be found in time polynomial in $|E(M)|$.
\end{theorem}

\subsection{Modified Decomposition}
For our algorithmic purposes we will not use the Theorem~\ref{thm:decomp} but rather a modification 
proved by Dinitz and Kortsarz in~\cite{DinitzK14}. Dinitz and Kortsarz in~\cite{DinitzK14} first observed that  some restrictions in the definitions of $2$- and $3$-sums are not important for the algorithmic purposes. In particular, in the definition of the $2$-sum, the unique $e\in E(M_1)\cap E(M_2)$ is not a loop or coloop of $M_1$ or $M_2$, and $|E(M_1)|,|E(M_2)|\geq 3$ could be dropped. Similarly, in the definition of $3$-sum the conditions  that $Z=E(M_1)\cap E(M_2)$ does not contain a cocircuit of $M_1$ or $M_2$, and $|E(M_1)|,|E(M_2)|\geq 7$ could be dropped. We define \emph{extended} $1$-, $2 $- and $3$-sums by omitting these restrictions.  Clearly, Theorem~\ref{thm:decomp} holds if we replace sums by extended sums in the definition of the $\{1,2,3\}$-decomposition. To simplify notations, we use $\oplus_1,\oplus_2,\oplus_3$ and $\oplus$ to denote these extended sums. Finally, we also need the notion of a conflict graph associated with a \emph{$\{1,2,3\}$-decomposition} of a 
matroid $M$ given by  Dinitz and Kortsarz in~\cite{DinitzK14}.

\begin{definition}[\cite{DinitzK14}]
Let \dpair\ be a $\{1, 2, 3\}$-decomposition of  a matroid $M$. The \emph{intersection} (or \emph{conflict}) graph of \dpair\  is the graph $G_T$ with the vertex set $\mathcal{M}$ such that distinct $M_1,M_2\in \mathcal{M}$ are adjacent in $G_T$ if and only if $E(M_1)\cap E(M_2)\neq\emptyset$. 
\end{definition}

Dinitz and Kortsarz in~\cite{DinitzK14} showed how to modify a given decomposition in order to make the conflict graph a forest. In fact they proved a slightly stronger condition that for any $3$-sum (which by definition is summed along a circuit of size $3$), the circuit in the intersection is contained entirely in two of the lowest-level matroids. In other words, while the process of summing matroids might create new circuits that contain elements that started out in different matroids, any circuit that is used as the intersection of a sum existed from the very beginning.

Let \dpair\ be a $\{1, 2, 3\}$-decomposition of  a matroid $M$. A node of $V(T)$ with degree at least $2$ is called an {\em internal} node of $T$. 
 Note that with each internal node $t$ of $T$  one can associate a matroid $M_t$, but also the set of elements that is the intersection of the ground sets of its children (and is thus not in the ground set of $M_t$). This set is either the empty set (if $M_t$ is the $1$-sum of its children), a single element (if it is the $2$-sum), or three elements that form a circuit in both of its children (if it is the $3$-sum). For an internal node $t$, let $Z_{M_t}$ denote this set.  Essentially,  corresponding to an internal node of $t\in V(T)$ with children $t_1$ and $t_2$, denote by $Z_{M_t}=E(M_{t_1})\cap E(M_{t_2})$ its \emph{sum-set}.

Let $t$ be an internal node of $T$ and $t_1$ and $t_2$ be its children. 
Using the terminology of Dinitz and Kortsarz in~\cite{DinitzK14}, we say that  $Z_{M_t}$ is \emph{good} if all the elements of $Z_{M_t}$ belong to the same basic matroid that is a descendant of $M_{t_1}$ in $T$ and they belong to the same basic matroid that is a descendant of $M_{t_2}$ in $T$. We say that 
a $\{1, 2, 3\}$-decomposition of $M$ is \emph{good} if all the sum-sets are good. 
We state the result of \cite{DinitzK14} in the following form that is convenient for us.

\begin{theorem}[\cite{DinitzK14}]\label{thm:decomp-good}
Every regular matroid $M$ has a good  $\{1, 2, 3\}$-decomposition in which every basic matroid is either graphic, cographic, or isomorphic to a matroid obtained from $R_{10}$ by  (possibly)  adding  parallel elements. Moreover,
such a decomposition (together with the graphs whose cycle and bond matroids are isomorphic to the corresponding basic graphic and cographic matroids) can be found in time polynomial in $||M||$.
\end{theorem}

Using this theorem, for a given regular matroid, we can obtain in polynomial time a good $\{1,2,3\}$-decomposition with a collection $\mathcal{M}$ of basic matroids,  where every basic matroid is either graphic, or  cographic, or isomorphic to a matroid obtained from $R_{10}$ by  deleting some elements and adding parallel elements and deleting. Then  we obtain a conflict forest $G_T$, whose nodes are basic matroids and the edges correspond to extended  $2$- or $3$-sums such that their sum-sets are  the elements of the basic matroids that are the end vertices of the corresponding edge. 
By adding bridges between components of $G_T$ corresponding to 1-sums, we obtain a $\mathcal{T}$ is a \emph{conflict tree} for a good $\{1,2,3\}$-decomposition, whose edges correspond to extended 1, 2 or 3-sums between adjacent matroids. Then we get the following corollary.

\begin{corollary}\label{cor:decomp-good}
For a given regular matroid $M$, there is a (conflict) tree $\mathcal{T}$, whose set of nodes is a set of matroids $\mathcal{M}$, where each element of $\mathcal{M}$ is a graphic or cographic matroid, or a matroid obtained from $R_{10}$ by adding (possibly) parallel elements,  that has the following properties:
\begin{itemize}
\item[i)]  if two distinct matroids $M_1,M_2\in \mathcal{M}$ have nonempty intersection, then $M_1$ and $M_2$ are adjacent in $\mathcal{T}$,
\item[ii)] for any distinct $M_1,M_2\in \mathcal{M}$, $|E(M_1)\cap E(M_2)|=0,~1$ or $3$,
\item[iii)] $M$ is obtained by the consecutive performing extended 1, 2 or 3-sums for adjacent matroids in any order. 
\end{itemize}
Moreover, $\mathcal{T}$ can be constructed in a polynomial time.
\end{corollary}

If $\mathcal{T}$ is a conflict tree for a matroid $M$, we say that $M$ is defined by $\mathcal{T}$.

\section{Elementary reductions for Space Cover}\label{sec:reduction}
In this section we give some elementary  reduction rules that we apply on the instances of 
\wss\ and \rwss\ to make it more structured and thus easier to design an \classFPT\ algorithm. 
Throughout this section we will assume that the input matroid \mat{} is regular. 
\subsection{Reduction rules for \wss}
Let $(M,w,T,k)$ be an instance of \wss, where $M$ is a regular matroid.  For technical reasons, we permit the weight function $w$  to assign $0$ to the elements of $E$. 
However, it observe that if $M$ has a nonterminal element $e$ with $w(e)=0$, then we can always  
include it in a (potential) solution. This  simple observation is formulated in the following reduction rule.  
We  do not give a proof of safeness as it follows easily.

\begin{redrule} [{\bf Zero-element reduction rule}] 
\label{rule:zero-rule}
If there is an element $e\in E\setminus T$ with $w(e)=0$, then contract $e$. 
\end{redrule}

\noindent 
The next rule is used to remove irrelevant terminals.

\begin{redrule} [{\bf Terminal circuit reduction rule}] 
\label{rule:terminalrule}
If there is a circuit $C\subseteq T$, then delete an arbitrary element $e\in C$ from $M$.
\end{redrule}

\begin{lemma}
Reduction Rule~\ref{rule:terminalrule} is safe.
\end{lemma}
\begin{proof}
We first prove the forward direction.  
Suppose  there is a circuit $C\subseteq T$ and $e\in C$.
Clearly, if $F\subseteq E\setminus T$ spans $T$, then $F$ spans $T\setminus\{e\}$ as well. For the reverse direction assume that  $F\subseteq E\setminus T$ spans $T\setminus \{e\}$. 
Let $C=\{e,e_1,\ldots,e_r\}$. Since $F\subseteq E\setminus T$ spans $T\setminus \{e\}$, there are circuits $C_1,\ldots,C_r$ such that $e_i\in C_i\subseteq F\cup\{e_i\}$. By Observation~\ref{obs:symm}, 
$\tilde{C}=(C_1\bigtriangleup\ldots\bigtriangleup C_r)\bigtriangleup C$ is a cycle. However, observe that $\tilde{C}$ only contains elements from $F\cup \{e\}$. In fact, 
since $e\notin C_i$ for $i\in\{1,\ldots,r\}$, $e\in \tilde{C}$ and thus there is a circuit $C'$ such that 
$e\in C'\subseteq  \tilde{C}$. This implies that $e\in C'\subseteq F\cup\{e\}$ and thus $F$ spans $e$. This completes the proof. 
\end{proof}

\noindent 
Now we remove irrelevant nonterminals. It is clearly safe to delete loops as there always exists a solution 
 $F$  such that $F\in {\cal I}$. 
\begin{redrule} [{\bf Loop reduction rule}] 
\label{rule:looprule}
If  $e\in E\setminus T$  is a loop, then delete $e$.
\end{redrule}
\noindent 
We remark that it is  safe to apply Reduction Rule~\ref{rule:looprule} even for \rwss.  Our next rule removes parallel elements. 
\begin{redrule} [{\bf Parallel reduction rule}] 
\label{rule:parallelrule}
If there are two  elements $e_1,e_2\in E\setminus T$ such that $e_1$ and $e_2$ are parallel and $w(e_1)\leq w(e_2)$, then delete $e_2$.
\end{redrule}

\begin{lemma}
Reduction Rule~\ref{rule:parallelrule} is safe.
\end{lemma}
\begin{proof}
Let $e_1,e_2\in E\setminus T$ be parallel elements such that $w(e_1)\leq w(e_2)$. Then, by Observations~\ref{obs:par}, for any $F\subseteq E\setminus T$ that spans $T$ such that $e_2\in F$, $F'=(F\setminus\{e_2\})\cup\{e_1\}$ also spans $T$. Hence, it is safe to delete $e_2$.
\end{proof}

\noindent
To sort out  the trivial \yesinstance or \noinstance after the exhaustive applications of the above rules, we apply the next rule. 

\begin{redrule} [{\bf Stopping rule}]
\label{rule:stoprule}
If  $T=\emptyset$, then return \yes\ and stop. Else, if $E\setminus T=\emptyset$ or $|T|>k$, then return 
\no\  and stop.
\end{redrule}

\begin{lemma}
Reduction Rule~\ref{rule:stoprule} is safe.
\end{lemma}
\begin{proof}
Indeed if $T=\emptyset$, then we have a \yesinstance\ of the problem, and if $T\neq\emptyset$ and $E\setminus T=\emptyset$, then the considered instance is a  \noinstance. 
If we cannot apply Reduction Rule~\ref{rule:terminalrule} ({\bf Terminal circuit reduction rule}), then $T$ is an independent set of $M$. Hence, if $F\subseteq E\setminus T$ spans $T$, $|F|\geq |T|$. Since we have no element with zero weight after the exhaustive application of Reduction Rule~\ref{rule:zero-rule} ({\bf Zero-element reduction rule}), if $k<|T|$, then the input instance is a  \noinstance. 
\end{proof}

\subsection{Reduction rules for \rwss}
For \rwss, we use the following modifications of Reduction Rules~\ref{rule:zero-rule}-\ref{rule:stoprule}, where applicable. Proofs of these rules are analogous to its counter-part for \wss\ and thus omitted. 

\begin{redrule} [{\bf Zero-element reduction rule$^*$}] 
\label{rule:zero-rule*}
 If there is an element $e\in E\setminus (T\cup\{e^*\})$ with $w(e)=0$, then contract $e$.
\end{redrule}

\begin{redrule} [{\bf Terminal circuit reduction rule$^*$}] 
\label{rule:terminalrule*}
If there is a circuit $C\subseteq T$, then delete an arbitrary  element $e\in C$ such that $e\neq t^*$ from $M$. If $t^*$ is a loop, then delete $t^*$.
\end{redrule}

\begin{redrule} [{\bf Parallel reduction rule$^*$}] 
\label{rule:parallelrule*}
If there are two  elements   $e_1,e_2\in E\setminus T$ such that $e_1$ and $e_2$ are parallel, $e_1\neq e^*$ and $w(e_1)\leq w(e_2)$, then delete $e_2$.
\end{redrule}

\noindent 
It could happen that $w(e^*)=0$, thus we obtain the following variant of Reduction Rule~\ref{rule:stoprule}.

\begin{redrule} [{\bf Stopping rule$^*$}] 
\label{rule:stoprule*}
If  $T=\emptyset$, then return \yes\ and stop. Else, if $E\setminus T=\emptyset$ or $|T|>k+1$, then return 
\no\  and stop.
\end{redrule}

\subsection{Final lemma}

If we have an independence oracle for \mat{} or if  we can check in polynomial time using a given representation of $M$ whether a given subset of $E$ belongs to $\cal I$, then we get the following lemma.

\begin{lemma}\label{lem:prepr}
Reduction Rules~\ref{rule:zero-rule}-\ref{rule:stoprule*}
can be applied in time polynomial in 
$||M||$.
\end{lemma}

\section{Solving Space Cover for basic matroids}\label{sec:basic}
In this section we solve \brwss\  on basic matroids that are building blocks of regular matroid. In particualr we solve \brwss\ for $R_{10}$, graphic and cographic  matroids. We first give an algorithm on $R_{10}$, followed by algorithms on graphic matroids based on  algorithms for {\sc Steiner Tree} and its generalization. Finally, we give algorithms on cographic  matroids based on ideas inspired by important separators.

\subsection{\brwss\ on $R_{10}$}
In this section we give an algorithm for \brwss\   about matroids that could be obtained from $R_{10}$ 
by either adding parallel elements, or  by deleting  elements or by contracting elements. Observe that an instance of \brwss\ for such a matroid is reduced to an instance with a matroid that has at most $20$ elements by the exhaustive application of {\bf Terminal circuit reduction rule} and 
{\bf Parallel reduction rule}. Indeed, in the worst case we obtain the matroid from $R_{10}$ by adding exactly one parallel element for each element of $R_{10}$. Since the matroid, \mat{}, of the reduced instance has at most $20$ elements we can solve \brwss\ by examining all subsets of  $E$ of size at most 
$k$. This brings us to the following. 

\begin{lemma}\label{obs:r10}
\wss\ can be solved in polynomial time for matroids that can be obtained from $R_{10}$ by adding parallel elements, element deletions and contractions. 
\end{lemma}

\subsection{\wss{} for graphic matroids}\label{sec:graphic}
Recall that,  \textsc{Steiner Forest} that we restate here can be seen as a special case of  \wss\ on graphic matroids by a simple reduction.

\smallskip
\defparproblem{\probSteinerF}%
{A (multi) graph $G$,  a weight function $w\colon E\rightarrow \mathbb{N}$,
a collection of pairs of distinct vertices (\emph{demands}) $\{x_1,y_1\},\ldots,\{x_r,y_r\}$ of $G$, and a nonnegative integer $k$} {$k$}
{Is there a set $F\subseteq E(G)$ with $w(F)\leq k$ such that for any $i\in\{1,\ldots,r\}$, $G[F]$ contains an $(x_i,y_i)$-path?}
\\
\noindent
In this section, we ``reverse this reduction'' in a sense and use this reversed reduction to solve \brwss. In particular we utilize an algorithm for {\sc Steiner Forest} to give an \classFPT{} algorithm for \brwss\ on graphic matroids. 
It seems a folklore knowledge that  \textsc{Steiner Forest} is \classFPT{} when parameterized by the number of edges in a solution. However, we did not find any paper that contains an algorithm for 
\textsc{Steiner Forest}.  Furthermore, we needed to know concrete running time for \textsc{Steiner Forest}  to obtain the desired form or running time for \wss.  Thus, we decided to give an algorithm here.

\subsubsection{A single-exponential algorithm for {\sc Steiner Forest}}
Our algorithm  is based on the \classFPT{} algorithm for the following well-known 
parameterization of  \probSteiner. 
Let us remind that in  \probSteiner, we are given a (multi) graph $G$,  a weight function $w\colon E\rightarrow \mathbb{N}$,
a set of vertices $S\subseteq V(G)$ called \emph{terminals}, and a nonnegative integer $k$.
The task is to decide whether  there is a set $F\subseteq E(G)$ with $w(F)\leq k$ such that the subgraph of $G$ induced by $F$ is a tree that contains the vertices of $S$.

 It was already shown by  Dreyfus and Wagner~\cite{DreyfusW71}, in 1971, that \textsc{Steiner Tree} can be solved in time $3^p\cdot n^{\cO(1)}$, where $p$ is the number of terminals.  The current best \classFPT-algorithms for \textsc{Steiner Tree}  are given by Bj{\"{o}}rklund et al.~\cite{BjorklundHKK07} and Nederlof~\cite{Nederlof13} (the first algorithm demands exponential in $p$ space and the latter uses polynomial space) and  runs in time $2^p\cdot n^{\cO(1)}$.  Finally, we are ready to describe the algorithm for \probSteinerF.

\begin{lemma}\label{lem:forest}
\probSteinerF is solvable in time $4^k\cdot n^{\cO(1)}$. 
\end{lemma}

\begin{proof}
Let $(G,w,\{x_1,y_1\},\ldots,\{x_r,y_r\},k)$ be an instance of \textsc{Steiner Forest}. Consider the auxiliary graph $H$ with $V(H)=\cup_{i=1}^r\{x_i,y_i\}$ and $E(H)=\{x_1,y_1\},\ldots,\{x_r,y_r\}$. Let 
$S_1,\ldots,S_t$ denote the sets of vertices of the connected components of $H$. Recall, that 
a set $F\subseteq E(G)$ with $w(F)\leq k$ is said to be a {\em solution-forest} for {\sc Steiner Forest} is  
for any $i\in\{1,\ldots,r\}$, $G[F]$ contains a $(x_i,y_i)$-path. Now notice that 
 $F\subseteq E(G)$, of weight at most $k$, is a solution-forest to an instance  
 $(G,w,\{x_1,y_1\},\ldots,\{x_r,y_r\},k)$  of   \textsc{Steiner Forest} if and only if the vertices of $S_i$ are in the same component of $G[F]$ for every $i\in\{1,\ldots,t\}$. We will use this correspondence to obtain an algorithm for  \probSteinerF.  

Now we bound the number of vertices in $V(H)$. Let $F$ be a minimal forest-solution. First of all observe that  since the weights on edges are positive we have that $|F|\leq k$. The vertices of $S_i$  must be in the same component of $G[F]$, thus we have that  $|F|\geq \sum_{i=1}^t(|S_i|-1)$. Hence, $\sum_{i=1}^t|S_i|\leq |F|+t$.  If $\sum_{i=1}^t|S_i|> |F|+t$ we return that  $(G,w,\{x_1,y_1\},\ldots,\{x_r,y_r\},k)$ is a \noinstance.  So from now onwards assume that $\sum_{i=1}^t|S_i|\leq |F|+t$.   Furthermore,  since $F$ is  a minimal forest-solution each connected component of $G[F]$ has size at least $2$ and thus $t\leq k$.  
Thus, we have an instance with $|V(H)|\leq 2k$ and $t\leq k$.

For $I\subseteq \{1,\ldots,t\}$, let  $W(I)$ denotes the minimum weight of a Steiner tree for the set of terminals $\cup_{i\in I}S_i$.  We assume that $W(I)=+\infty$ if such a tree does not exist.  Furthermore, if the minimum weight of a Steiner tree  is at least $k+1$ then also we assign $W(I)=+\infty$. 
All the $2^t$ values of $W(I)$ corresponding to $I\subseteq \{1,\ldots,t\}$ 
can be computed in time $2^{|V(H)|)}\cdot n^{\cO(1)}=4^{k}\cdot n^{\cO(1)}$ using the results of \cite{BjorklundHKK07} or~\cite{Nederlof13}.  

For  $J\subseteq \{1,\ldots,t\}$, let  $W'(J)$ denote the minimum weight of a solution-forest 
for the sets $S_j$, where $j\in J$.  That is, $W'(J)$ is assigned the minimum weight of a set $F\subseteq E(G)$ such that the vertices of $S_j$ for $j\in J$ are in the same component of $G[F]$. Furthermore, if 
 such a set $F$ does not exist or  the weight is at least $k+1$ then  $W'(J)$ is assigned $+\infty$. 
 Clearly, $W'(\emptyset)=0$. Notice that $(G,w,\{x_1,y_1\},\ldots,\{x_r,y_r\},k)$ is a \yesinstance\ for {\sc Steiner Forest}  if and only if $W'(\{1,\ldots,t\})\leq k$.  Next, we give the  recurrence relation for the dynamic programming algorithm to compute the values of $W'(J)$.

\begin{equation}\label{eq:W}
W'(J)=\min_{\begin{subxarray} I\subseteq J  \\ I \neq \emptyset  \end{subxarray}}\Big\{W'(J\setminus I)+W(I) \Big\}.
\end{equation}
We claim that the above recurrence holds for  every $J\subseteq \{1,\ldots,t\}$. 
To prove the  forward direction of the claim,  assume that $F\subseteq E(G)$ is a set of edges of minimum weight such that the vertices in $S_j$,  $j\in J$, are in the same component of $G[F]$. Let $X$ be a set of vertices of an arbitrary component of $G[F]$ and  $L$ denote the set of edges of this component.  Let 
$I=\{i\in J\mid S_i\subseteq X\}$. Notice that by the minimality of $F$, $I\neq\emptyset$. Since $W(I)\leq w(L)$ and $W'(J\setminus I)\leq w(F\setminus L)$,  we have that 
\begin{align*}
W'(J)= w(F)=w(F\setminus L)+w(L)\geq  W'(J\setminus I)+W(I)\geq 
 \min_{\begin{subxarray} I\subseteq J  \\ I \neq \emptyset  \end{subxarray}}\Big\{W'(J\setminus I)+W(I) \Big\}.
\end{align*}
To show the opposite inequality, consider a nonempty set  $I\subseteq J$, and 
let $L$ be the set of edges of a Steiner tree of minimum weight for the set of terminals $\cup_{i\in I}S_i$ and let 
$F$ be the set of edges of a Steiner forest of minimum weight for the sets of terminals $S_j$ for $j\in J\setminus I$. Then we have that for $F'=L\cup F$, the vertices of $S_i$ are in the same component of $G[F']$ for each $i\in J$. Hence,
\begin{equation}\label{eq:leq}
W'(J)\leq w(L)+w(F)=W'(J\setminus I)+W(I).
\end{equation}
Because (\ref{eq:leq}) holds for any nonempty $I\subseteq J$,
$$W'(J)\leq \min_{\begin{subxarray} I\subseteq J  \\ I \neq \emptyset  \end{subxarray}}\Big\{W'(J\setminus I)+W(I) \Big\}.
 $$

We compute the values for $W'(J)$ in the increasing order of the size of $J\subseteq \{1,\ldots,t\}$. Towards this we use  Equation~\ref{eq:W} and the fact that $W'(\emptyset)=0$.  Each entry of  
$W'(J)$ can be computed by taking a minimum over $2^{|J|}$ pre-computed entries in $W'$ and $W$. Thus, the total time to compute $W'$ takes $(\sum_{i=0}^t {t \choose i} 2^i)\cdot n^{\cO(1)} =3^t\cdot n^{\cO(1)}=3^k\cdot n^{\cO(1)}$. 
Having computed  $W'$ we  return yes or no based on whether $W'(\{1,\ldots,t\})\leq k$. This completes the proof. 
\end{proof}

\subsubsection{An algorithm for \wss\ on graphic matroids}
Now using the algorithm for {\sc Steiner Forest} mentioned in Lemma~\ref{lem:forest}, we design an algorithm for 
\wss\ on graphic matroids. 
\begin{lemma}\label{lem:graphic}
\wss\ can be solved in time 
$4^k\cdot ||M||^{\cO(1)}$ on graphic matroids. 
\end{lemma}
\begin{proof}
Let  $(M,w,T,k)$ be an instance of \wss\, where $M$ is a graphic matroid. First, we exhaustively apply  
Reduction Rules~\ref{rule:zero-rule}-\ref{rule:stoprule}. Thus, by Lemma~\ref{lem:prepr}, in polynomial  time we either solve the problem or obtain an equivalent instance, where $M$ has no loops and the weights of nonterminal elements are positive.  To simplify notations, we also denote the reduced instance by  $(M,w,T,k)$. Observe that $M$ remains to be graphic. It is well-known that given a graphic matroid, in polynomial time one can find a graph $G$ such that $M$ is isomorphic to the cycle matroid 
$M(G)$~\cite{Seymour81}. Assume that $T=\{x_1y_1,\ldots, x_ry_r \}$ is the set of edges of $G$ corresponding to the terminals of the instance of \wss. We define the graph $G'=G-T$. 
Recall that $F\subseteq E(G)\setminus T$ spans $T$ if and only if for each $e\in T$, there is a cycle 
$C$ of $G$ such that $e\in C\subseteq F\cup \{e\}$.  
Clearly, the second condition can be rewritten as follows: for any $i\in\{1,\ldots,r\}$, $G[F]$ contains an $(x_i,y_i)$-path. It means that the instance $(G',w,\{x_1,y_1\},\ldots,\{x_r,y_r\},k)$ of \probSteinerF is equivalent to the instance  $(M,w,T,k)$ of \wss.
Now we apply Lemma~\ref{lem:forest} on the instance $(G',w,\{x_1,y_1\},\ldots,\{x_r,y_r\},k)$ of \probSteinerF to solve the problem.
\end{proof}

\subsubsection{An Algorithm for  \rwss\ on graphic matroids}
Besides solving \wss{}, we need to solve \rwss{} on graphic matroids. In fact, \rwss{} can be reduced to \probSteinerF{}. From the other side, we can solve this problem modifying the algorithm for  
\probSteinerF{} from Lemma~\ref{lem:forest} and this approach given in this section allows to obtain a better running time.

\begin{lemma}\label{lem:graphic-restr}
\rwss\ can be solved in time $6^k\cdot ||M||^{\cO(1)}$ on graphic matroids. 
\end{lemma}

\begin{proof}
Let  $(M,w,T,k,e^*,t^*)$ be an instance of \rwss, where $M$ is a graphic matroid. First, we exhaustively apply Reduction Rules~\ref{rule:looprule} and~\ref{rule:zero-rule*}-\ref{rule:stoprule*}. 
Thus, by Lemma~\ref{lem:prepr}, in polynomial  time we either solve the problem or obtain an equivalent instance.  
Notice that it can happen that $e^*$ is deleted by  Reduction Rules~\ref{rule:looprule} and~\ref{rule:zero-rule*}-\ref{rule:stoprule*}. For example, if $e^*$ is a loop then it can be deleted by Reduction Rule~\ref{rule:looprule}.  In this case we obtain an instance of \wss\ and can solve it using 
Lemma~\ref{lem:graphic}.  From now onwards we assume that $e^*$ is not delated by our reduction rules.

To simplify notations, we use  $(M,w,T,k,e^*,t^*)$ to denote the reduced  instance. If we started with graphic matroid then it remains so even after applying Reduction Rules~\ref{rule:looprule} and~\ref{rule:zero-rule*}-\ref{rule:stoprule*}. Furthermore, given $M$, in polynomial time we can find  a graph $G$ such that $M$ is isomorphic to the cycle matroid $M(G)$~\cite{Seymour81}. 
Let $T=\{x_1y_1,\ldots, x_ry_r \}$ denote the set of edges of $G$ corresponding to the terminals of the instance of \rwss. Without loss of generality assume  that $t^*=x_1y_1$.
Let  $G'$ and  $G_e^*$ denote the graphs $G-T$ and $G-\{e^*\}$, respectively.  
Recall that, $F\subseteq E(G)\setminus T$ spans $T$ if and only if for each $e\in T$, there is a cycle $C$ of $G$ that contains $e$ and all the edges in $C$ are contained in $F\cup \{e\}$. 
Clearly, the second condition can be rewritten as follows: for every $i\in\{1,\ldots,r\}$, $G[F]$ contains a  
$(x_i,y_i)$-path.  For \rwss, we additionally have the condition that $F\setminus \{e^*\}$ spans $t^*$. That is,  $G[F]$  contains a $(x_1,y_1)$-path that does not contain $e^*$. 
In terms of graphs, we obtain a special variant of  \probSteinerF. We solve the problem by  slightly modifying the algorithm of  Dreyfus and Wagner~\cite{DreyfusW71} and Lemma~\ref{lem:forest}.

As in the proof of Lemma~\ref{lem:forest}, we 
consider the auxiliary graph $H$ with $V(H)=\cup_{i=1}^r\{x_i,y_i\}$ and $E(H)=\{x_1,y_1\},\ldots,\{x_r,y_r\}$. Let 
$S_1,\ldots,S_t$ denote the sets of vertices of the connected components of $H$. Without loss of generality we  assume that  $x_1,y_1\in S_1$.  Let $F$ be a minimal solution. It is clear that 
$G[F]$ is a forest. 
Notice that $F\subseteq E(G)-T$, of weight at most $k$, is a minimal solution to an instance  
 $(G,w,\{x_1,y_1\},\ldots,\{x_r,y_r\},e^*,t^*,k)$  of   \rwss\ if and only if the vertices of $S_i$ are in the same component of $G[F]$ for every $i\in\{1,\ldots,t\}$ and the unique path between 
$x_1$ and $y_1$ in the component containing $S_1$ does not contain $e^*$. We will use this correspondence to obtain an algorithm for  the special variant of  \probSteinerF and hence \rwss. 

Now we bound the number of vertices in $V(H)$. Let $F$ be a minimal solution. First of all observe that  since the weights on edges are positive, with an exception of $e^*$,  we have that $|F|\leq k+1$. The vertices of $S_i$  must be in the same component of $G[F]$, thus we have that  $|F|\geq \sum_{i=1}^t(|S_i|-1)$. Hence, $\sum_{i=1}^t|S_i|\leq |F|+t$.  If $\sum_{i=1}^t|S_i|> |F|+t$ we return that  $(G,w,\{x_1,y_1\},\ldots,\{x_r,y_r\},e^*,t^*,k)$ is a \noinstance.  So from now onwards assume that $\sum_{i=1}^t|S_i|\leq |F|+t$.   Furthermore,  since $F$ is  a minimal solution each connected component of $G[F]$ has size at least $2$ and thus $t\leq k+1$.  
Thus, we have an instance with $|V(H)|\leq 2k+1$ and $t\leq k+1$.  

Given $I\subseteq \{1,\ldots,t\}$, by $Z_I$, we denote $\cup_{i\in I}S_i $. 
For $I\subseteq \{1,\ldots,t\}$, let  $W(I)$ denote the minimum weight of a  tree $R$ in $G'$ 
such that $Z_I \subseteq V(R)$ and if $x_1,y_1\in Z_I$, then the $(x_1,y_1)$-path in $R$ does not contain $e^*$. We assume that $W(I)=+\infty$ if such a tree does not exist.  Furthermore, if the minimum weight of such a tree $R$  is at least $k+1$ then also we assign $W(I)=+\infty$.  
Notice that if $|Z_I|>k+2$, then $W(I)\geq k+1$ as any tree that contains $Z_I$ has at least $|Z_I|-1>k+1$ edges and only $e^*$ can have weight $0$.  
In this case we can safely set $W(I)=+\infty$, because we are interested in trees of weight at most $k$.
Thus from now onwards we can assume that  $|Z_I|\leq k+2$.
We compute the values of $I\subseteq \{1,\ldots,t\}$ such that $1\in I$ by modifying  the algorithm of  Dreyfus and Wagner~\cite{DreyfusW71}. Next  we present this modified algorithm to compute the values of $W$. 

For a vertex $v\in V(G)$ and $X\subseteq Z_I$, let $c(v,X,\ell)$ be the minimum weight of a subtree $R'$ of $G'$ with at most $\ell$ edges such that 
\begin{itemize}
\setlength{\itemsep}{-2pt}
\item[i)] $X\subseteq V(R')$, 
\item[ii)] $v\in V(R)$,
\item[iii)] if $x_1,y_1\in X$, then the $(x_1,y_1)$-path in $R'$ does not contain $e^*$, 
\item[iv)] if $x_1\in X$ and $y_1\notin X$, then the $(x_1,v)$-path in $R'$ does not contain $e^*$, and
\item[v)] if $y_1\in X$ and $x_1\notin X$, then the $(y_1,v)$-path in $R'$ does not contain $e^*$.
\end{itemize}
We assume that $c(v,X,\ell)=+\infty$ if such a tree $R'$ does not exist.

Initially we set \[c(v,X,0)=
\begin{cases}
0& \mbox{if }\{v\}=X,\\
+\infty& \mbox{if }\{v\}\neq X.
\end{cases}
\]
We compute $c(v,X,\ell)$ using the following auxiliary recurrences. 
For an ordered pair of vertices $(u,v)$ such that $uv\in E(G')$,
$$
c'(u,v,X,\ell)=
\begin{cases}
+\infty&\mbox{if }uv=e^*\text{ and }|X\cap\{x_1,y_1\}|=1,\\
c(v,X,\ell-1)+w(uv)& \mbox{otherwise}.
\end{cases}
$$
For an ordered pair of vertices $(u,v)$ such that $uv\in E(G')$, a nonempty $Y\subseteq X$, and two nonnegative integers $\ell_1$ and $\ell_2$ such that $\ell=\ell_1+\ell_2+1$,
$$
c''(u,v,X,Y,\ell_1,\ell_2)=
\begin{cases}
+\infty &\mbox{if }uv=e^*\text{ and }\\
 &|Y\cap\{x_1,y_1\}|=1,\\
c(u,X\setminus Y,\ell_1)\\
+c(v,Y,\ell_2)+w(uv)&\mbox{otherwise}.\\
\end{cases}
$$
Finally,
\begin{align*}
c(u,X,\ell)&=\min \bigg\{c(u,X,\ell-1),\min\limits_{v\in N_{G'}(u)}c'(u,v,X),\\
&\min\limits_{v\in N_{G'}(u)}\Big\{c''(u,v,X\setminus Y,Y,\ell_1,\ell_2)\mid \emptyset\neq Y\subseteq X,\ell_1+\ell_2=\ell-1 \Big\} \bigg\}.
\end{align*}

For all $v\in V(G)$, we fill the table $c(v,\cdot,\cdot)$ as follows. We iteratively consider the values of $\ell$ starting from $1$ and ending at $k$ and for each value of $\ell$ we  consider the subsets of $Z_I$ in the increasing order of their size. 
If there is a vertex $v\in V(G)$ with $c(v,Z_I,k+1)\leq k$ then we set  $W(I)=c(v,Z_I,k+1)$, else, we set 
$W(I)=+\infty$.

The correctness of the computation of $W(I)$ can be proved by standard dynamic programming arguments. In fact, it essentially follows along the lines of 
 the proof of Dreyfus and Wagner~\cite{DreyfusW71}. The  only difference is that we have to take into account the conditions $iii)$ to $v)$ that are used to ensure that the $(x_1,y_1)$-path in the obtained tree avoids $e^*$.  Since $|Z|\leq k+2$, the computation of $W(I)$ for a given $I$ can be done in  time $3^k\cdot n^{\cO(1)}$. 
 Thus, all the $2^t$ values of $W(I)$ corresponding to $I\subseteq \{1,\ldots,t\}$  such that $1\in I$ can be computed in time  $6^k\cdot n^{\cO(1)}$.

Next, we show how we can compute $W(I)$ for $I\subseteq \{2,\ldots,t\}$. Recall that $x_1,y_1\in S_1$  and thus  for $I\subseteq \{2,\ldots,t\}$,   $W(I)$ just denotes the minimum weight of a Steiner tree for the set of terminals $Z_I$ in the graph $G'$. Hence, for $I\subseteq \{2,\ldots,t\}$, we can compute $W(I)$ by 
using  the algorithm of Dreyfus and Wagner~\cite{DreyfusW71} without modification.  We could also  compute $W(I)$ using the results of \cite{BjorklundHKK07} or~\cite{Nederlof13}. Thus, we can compute 
all the $2^t$ values of $W(I)$ corresponding to $I\subseteq \{1,\ldots,t\}$ in $6^k\cdot n^{\cO(1)}$ time. 

Now we use the table  $W$ to solve the instance $(M,w,T,k,e^*,t^*)$ of \rwss. As in the proof of Lemma~\ref{lem:forest}, for 
each $J\subseteq \{1,\ldots,t\}$,  denote by $W'(J)$ the minimum weight of a set $F\subseteq E(G')$ such that the vertices of $Z_J$  are in the same component of $G'[F]$  and  if $1\in J$  then the $(x_1,y_1)$-path in $G'[F]$ avoids $e^*$. In the . Furthermore,  if such a set $F$ does not exist or has weight at least $k+1$ then we set $W'(J)=+\infty$. 

Clearly, $W'(\emptyset)=0$. Notice that $(M,w,T,k,e^*,t^*)$ is a \yesinstance\ for \rwss\  if and only if $W'(\{1,\ldots,t\})\leq k$.  Next we give the  recurrence relation for the dynamic programming algorithm to compute the values of $W'(J)$.

\begin{equation}
\label{eq:Wrestrict}
W'(J)=\min_{\begin{subxarray} I\subseteq J  \\ I \neq \emptyset  \end{subxarray}}\Big\{W'(J\setminus I)+W(I) \Big\}.
\end{equation}
The proof of the correctness of the recurrence given in Equation~\ref{eq:Wrestrict} is verbatim same to the proof of recurrence given in Equation~\ref{eq:W} in the proof of Lemma~\ref{lem:forest}.

We compute the values for $W'(J)$ in the in the increasing order of size of $J\subseteq \{1,\ldots,t\}$. Towards this we use  Equation~\ref{eq:Wrestrict} and the fact that $W'(\emptyset)=0$.  Each entry of  
$W'(J)$ can be computed by taking a minimum over $2^{|J|}$ pre-computed entries in $W'$ and $W$. Thus, the total time to compute $W'$ takes $(\sum_{i=0}^t {t \choose i} 2^i)\cdot n^{\cO(n)} =3^t\cdot n^{\cO(1)}=3^k\cdot n^{\cO(1)}$. 
Having computed  $W'$ we  return yes or no based on whether $W'(\{1,\ldots,t\})\leq k$. This completes the proof. 
\end{proof}

\subsection{\brwss{} for cographic matroids}
In this section we design  algorithms for \brwss\ 
on co-graphic matroids. By the results of Xiao and Nagamochi~\cite{XiaoN12}, \wss{} can be solved in time $2^{\Oh(k\log k)}\cdot ||M||^{\Oh(1)}$, but to obtain a single-exponential in $k$ algorithm we use a different approach based
on the enumeration of \emph{important separators}  
proposed by Marx in~\cite{Marx06}. However, for our purpose we use the similar notion of 
 \emph{important cuts} and we follow the terminology given in~\cite{CyganFKLMPPS15} to define these objects. 

To introduce this technique, we need some additional definitions.  Let $G$ be a graph and let $X,Y\subseteq V(G)$ be disjoint. A set of edges $S$ is an \emph{$X-Y$ separator} if $S$ separates $X$ and $Y$ in $G$, i.e., every path that connects a vertex of $X$ with a vertex of $Y$ contains an edge of $S$. If $X$ 
is a single element set $\{u\}$, we simply write that $S$ is a $u-Y$ 
 separator. An $X-Y$-separator is \emph{minimal} if it is an inclusion minimal $X-Y$ separator.  It will be convenient to look at minimal $(X,Y)$-cuts from a different
perspective, viewing them as edges on the boundary of a certain set of
vertices. If $G$ is an undirected graph and $R\subseteq V(G)$ is a set
of vertices, then we denote by $\Delta_G(R)$ the set of edges with
exactly one endpoint in $R$, and we denote $d_G(R) = |\Delta_G(R)|$
(we omit the subscript $G$ if it is clear
from the context). We say that a vertex $y$ is \emph{reachable} from a vertex $x$ in a graph $G$ 
if $G$ has an $(x,y)$-path. For a set $X$, a vertex $y$ is reachable from $X$ if it is reachable from a vertex of $X$. Let $S$ be a minimal $(X,Y)$-cut in $G$ and let $R_G(X)$
be the set of vertices reachable from $X$ in $G\setminus S$; clearly,
we have $X\subseteq R_{G}(X)\subseteq V(G)\setminus Y$. 
Then it is easy to see
that $S$ is precisely $\Delta(R_G(X))$. 
  Indeed, every such edge has to be in $S$ (otherwise a
vertex of $V(G)\setminus R$ would be reachable from $X$) and $S$ cannot have an edge with both endpoints in $R_G(X)$ or both endpoints in
$V(G)\setminus R_G(X)$, as omitting any such edge would not change the fact that the set is an
$(X,Y)$-cut, contradicting minimality.  When the context is clear we omit the subscript and the set $X$ while defining $R$. 
\begin{proposition}[\cite{CyganFKLMPPS15}]
\label{prop:mincutboundary}
  If $S$ is a minimal $(X,Y)$-cut in $G$, then $S=\Delta_G(R)$, where $R$ is
  the set of vertices reachable from $X$ in $G\setminus S$.
\end{proposition}
Therefore, we may always characterize a minimal $(X,Y)$-cut $S$ as
$\Delta(R)$ for some set $X\subseteq R \subseteq V(G)\setminus Y$.

\begin{definition} {\rm \cite[Definition 8.6]{CyganFKLMPPS15}  [Important cut]}
\label{def:impsep}
  Let $G$ be an undirected graph and let $X,Y\subseteq V(G)$ be two
  disjoint sets of vertices. Let $S\subseteq E(G)$ be an $(X,Y)$-cut
  and let $R$ be the set of vertices reachable from $X$ in $G\setminus
  S$.  We say that $S$ is an {\em important
    $(X,Y)$-cut} if it is inclusion-wise
  minimal and there is no $(X,Y)$-cut $S'$ with $|S'|\le |S|$ such
  that $R\subset R'$, where $R'$ is the set of vertices reachable from
  $X$ in $G\setminus S'$.
\end{definition}

\begin{theorem}{\rm \cite{LokshtanovM13,MarxR14}, \cite[Theorems 8.11 and 8.13]{CyganFKLMPPS15} }
\label{thm:imp}
  Let $X,Y\subseteq V(G)$ be two disjoint sets of vertices in graph
  $G$ and let $k\ge 0$ be an integer. There are at most $4^k$ important
  $(X,Y)$-cuts of size at most $k$. Furthermore, the set of all important
  $(X,Y)$-cuts of size at most $k$ can be enumerated in time $\cO(4^k\cdot
  k\cdot (n+m))$.

\end{theorem}

For a partition $(X,Y)$ of the vertex set of a graph $G$, we denote by $E(X,Y)$ the set of edges with one end vertex in $X$ and the other in $Y$. 
 For a set of bridges $B$ of a graph $G$ and a bridge $uv\in B$, we say that $u$ is a \emph{leaf with respect to $B$}, if the component of $G-B$ that contains $u$ has no end vertex of any edge of $B\setminus\{uv\}$.
Clearly, for any set of bridges, there is a leaf with respect to it. Also we can make the following observation.

\begin{observation}\label{obs:bridges}
For the bond matroid $M^*(G)$ of a graph $G$ and $T\subseteq E(G)$, a set $F\subseteq E(G)\setminus T$ spans $T$ if and only if the edges of $T$ are bridges of $G-F$.
\end{observation}

\subsubsection{An algorithm for \wss\ on cographic matroids}
For our purpose we need a slight modification to the definition of important cuts. We start by defining the object we need and proving a combinatorial upper bound on it. 

\begin{define}
Let $G$ be a graph $s\in V(G)$ be a vertex and $T\subseteq V(G)\setminus \{s\}$ be a subset of terminals. We say that a set $W\subseteq V(G)$ is interesting if (a) $G[W]$ is connected, (b) $s\in W$ and 
$|T\cap W|\leq 1$. 
\end{define}

Next we define a partial order on all interesting sets of a graph. 
\begin{define}
Let $G$ be a graph $s\in V(G)$  be a vertex and $T\subseteq V(G)\setminus \{s\}$ be a subset of terminals. Given two interesting sets $W_1$ and $W_2$ we say that $W_1$ is better than $W_2$ and denote by 
$W_2\preceq W_1 $ if (a) $W_2 \subseteq W_1$, $|\Delta(W_1)|\leq |\Delta(W_2)|$ and $T\cap W_1 \subseteq T\cap W_2$.  
\end{define}

\begin{define}
Let $G$ be a graph $s\in V(G)$  be a vertex, $T\subseteq V(G)\setminus \{s\}$ be a subset of terminals and $k$ be a nonnegative integer. We say that an interesting set $W$ is a {\em $(s,T,k)$-semi-important set} if 
$|\Delta(W)|\leq k$ and there is no set $W'$ such that $W\preceq W' $. That is, $W$ is a maximal set under the relation $\preceq$. Furthermore,  $\Delta(W)$ corresponding to a $(s,T,k)$-semi-important set is called   
 a {\em $(s,T,k)$-semi-important cut}.
\end{define}

Now we have all the necessary definitions to state our lemma that upper bounds the number of semi-important sets 
and semi-important cuts.  
\begin{lemma}
\label{lem:numSIC}
For every graph $G$, a vertex  $s\in V(G)$, a subset $T\subseteq V(G)\setminus \{s\}$ and a nonnegative integer $k$, there are at most  $4^k(1+4^{k+1})$     $(s,T,k)$-semi-important cuts with $|\Delta(W)|=k$. 
Moreover, all such sets can be listed in time $16^k n^{\cO(1)}$.
\end{lemma}
\begin{proof}
Observe that  $(s,T,k)$-semi-important cuts and $(s,T,k)$-semi-important sets are in bijective correspondence and thus bounding one implies a bound on the other. In what follows we upper bound the number of   $(s,T,k)$-semi-important sets. Let $\cal F$ denote the set of all $(s,T,k)$-semi-important sets.  There are two kinds of  $(s,T,k)$-semi-important sets, those that do not contain any vertex of $T$ and those that contain exactly one vertex of $T$.  We denote the set of $(s,T,k)$-semi-important sets of first kind by  ${\cal F}_0$ and the second kind by  ${\cal F}_1$. We first bound the size of  ${\cal F}_0$. We claim that  for every set $W\in {\cal F}_0$, $\Delta(W)$ is an important $(s,T)$-cut of size $k$ in $G$. For a contradiction assume that there is a  set $W\in {\cal F}_0$ such that 
$\Delta(W)$ is not  an important $(s,T)$-cut of size $k$ in $G$. Then there exists a set  $W'$ such that 
$W\subsetneq W'$, $s\in W'$, $W'\cap T=\emptyset$ and  $|\Delta(W')|\leq |\Delta(W)|$. However, this implies that 
$W\preceq W' $ --  a contradiction. Thus,  for every set $W\in {\cal F}_0$,  $\Delta(W)$ is an important 
$(s,T)$-cut of size $k$ in $G$ and thus, by Theorem~\ref{thm:imp} we have that $|{\cal F}_0|\leq 4^k$. 

Now we bound the size of ${\cal F}_1$. Towards this we first modify the given graph $G$  and  obtain a new graph $G'$.  We first add a vertex $t\notin V(G)$ as a sink terminal. Then for every vertex $v_i\in T$ we add $k+1$ new vertices $Z_i=\{v_i^{1},\ldots,v_i^{k+1}\}$ and add an edge $v_{i}z$, for all $z\in Z_i$. Now for every vertex $v_i^{j} \in Z_i$ we make $2k+3$ new vertices $Z_i^j=\{v_i^{j_1},\ldots,v_i^{j_{2k+3}}\}$ and add an edge  $tz$, 
for all $z\in Z_i^j$.  Now we claim that  for every set $W\in {\cal F}_1$, $\Delta(W)$ is an important $(s,t)$-cut of size 
$2k+1$ in $G'$.  For a contradiction assume that there is a  set $W\in {\cal F}_1$ such that 
$\Delta(W)$ is not  an important $(s,t)$-cut of size $2k+1$ in $G'$. Then there exists a set   
$W'$ such that $W\subsetneq W'$, $s\in W'$, $W'\cap \{t\}=\emptyset$ and  $|\Delta(W')|\leq |\Delta(W)|$. That is, 
$\Delta(W')$ is an important cut dominating $\Delta(W)$.  Since $W\in {\cal F}_1$, there exists a vertex (exactly one) say $w\in T$ such that $w\in W$. Observe that $W'$ can not contain 
(a) any vertex but  $w$ from $T$ and (b) any vertex from the set $Z_i$, $v_i\in T$. If it does then $|\Delta(W')|$ will become strictly more than $2k+1$. This together with the fact that  $G[W']$ is connected we have that it does not 
contain any newly added vertex. That is, $W'\subseteq V(G)$ and contains only $w$ from  $T$. However, this implies that $W\preceq W' $ --  a contradiction. Thus,  for every set $W\in {\cal F}_1$, $\Delta(W)$ is an important 
$(s,t)$-cut of size $2k+1$ in $G'$ and thus, by Theorem~\ref{thm:imp} we have that $|{\cal F}_1|\leq 4^{2k+1}$. Thus, $|{\cal F}_0|+|{\cal F}_1|\leq 4^k +  4^{2k+1}$. This concludes the proof. 
\end{proof}

\begin{lemma}\label{lem:imp-new}
Let $M^*(G)$ be the bond matroid of $G$, $T\subseteq E(G)$, and suppose that  $F\subseteq E(G)\setminus T$ spans $T$. Let also $x$ be an end vertex of an edge $xy$ of $T$ such that $x$ is either in a leaf block or in a degree two block in $G-F$, $Y$ is the set of end vertices of the edges of $T$ distinct from $x$,   
$G'=G-T$ and let $W=R_{G'-F}(x)$. 
Then there is a  $(x,Y,k)$-semi-important set 
$W'$  such that $|\Delta_{G'}(W')|\leq |\Delta_{G'}(W)|$ and $F'=(F\setminus \Delta_{G'}(W))\cup \Delta_{G'}(W')$ spans $T$ in $M^*(G)$.
\end{lemma}

\begin{proof}
It is clear that $W$ is an interesting set.
 If $W$ is a semi-important set and $ \Delta_{G'}(W)$ is a $(x,Y,k)$-semi-important cut  of $G'$, then the claim holds for $W'=W$. Assume that $\Delta_{G'}(W)$ is not a $(x,Y,k)$-semi-important cut. Then there is a $(x,Y,k)$-semi-important set $W'$ of $G'$ such that $W \preceq W'$. Recall that this implies that (a) $G'[W']$ is connected, 
 (b) $W\subsetneq W'$, (c) $s\in W'$, (d) $|Y\cap W'|\leq 1$ and $|\Delta_{G'}(W')|\leq |\Delta_{G'}(W)|$. Since $G'$ does not have any edge of $T$ we have that  $\Delta_{G'}(W')\cap T=\emptyset$. Hence, $F'=(F\setminus \Delta_{G'}(W))\cup \Delta_{G'}(W')$ is disjoint from $T$. That is, $F'\subseteq E(G)\setminus T$.

To prove that $F'$ spans $T$, it is sufficient to show that for every $uv\in T$, there is a minimal cut-set 
$C_{uv}^*$ of $G$ such that $uv\in C_{uv}^*\subseteq F'\cup\{uv\}$.  
Let $uv\in T\setminus\{xy\}$. To obtain a contradiction, suppose  there is no {\em minimal} cut-set $\hat{C}_{uv}$ in $G$ such that $uv\in \hat{C}_{uv}\subseteq F'\cup \{uv\}$. Then, there is a $(u,v)$-path $P$ in $G$ such that $P$ has no edge of $F'\cup\{uv\}$. On the other hand $G$ has a cut-set $C_{uv}$ such that $uv\in C_{uv}\subseteq F\cup \{uv\}$. This implies that  every path between $u$ and $v$ that exists in $G-(F'\cup\{uv\})$, including $P$, 
must contain an edge of $C_{uv}$ such that it is present in $\Delta_{G'}(W)$ (these are the only edges we have removed from $F$). By our assumption we know that $P$ does not contain any edge from $\Delta_{G'}(W)$ (else we will be done). Now we know that $W$ can contain at most one vertex from $Y$.  Since $W$ does not contain both end-points of an edge in $T$ we have that at most one of $u$ or $v$ belongs to $W$. First let us assume that $W\cap \{u,v\}=\emptyset$. Thus by the definition of 
semi-important set, $W'\cap Y \subseteq W\cap Y$, we have that $u,v$ is outside of  $W'$.  However, we know that 
$\Delta_{G'}(W)$ contains an edge of $P$ and thus contains a vertex $z\in W$ that is on $P$. Since 
$W\subsetneq W'$ we have that $\Delta_{G}(W')$ contains at least two edges of $P$. However, none  of these edges are present in 
$\Delta_{G'}(W')$. The only edges $G'$ misses are those in $T$ and thus the edges present in 
$\Delta_{G}(W')\cap E(P)$  must belong to $T$. Let $Z$ denote the set of end-points of edges in $\Delta_{G}(W')\cap E(P)$. Observe that, $Z\cap S'=Z\cap S$. Let $z_1$ denote the first vertex on $P$ belonging to $W'$ (or $W$) and $z_2$ denote the last vertex on $P$ belonging to $W'$ (or $W$), respectively,  when we walk along the path $P$ starting from $u$. Since $z_1$ and $z_2$ belongs to $W$ and $G[W]$ is connected we have that there is a path $Q_{z_1z_2}$ in $G[W]$. Let $P_{uz_1}$ denote the subpath of $P$ between  $u$ and $z_1$ and let 
$P_{z_2v}$ denote the subpath of $P$ between  $z_2$ and $v$. This implies that the path $P'$ between $u$ and $v$ obtained by concatenating $P_{uz_1}Q_{z_1z_2}P_{z_2v}$ does not intersect $\Delta_{G'}(W)$. Observe that $P'$ does not contain any edge of $\Delta_{G'}(W)$ and $F'\cup\{uv\}$.  This is a contradiction to our assumption that every path between $u$ and $v$ that exists in $G-(F'\cup\{uv\})$ must contain an edge of $C_{uv}$ such that it is present in $\Delta_{G'}(W)$. 

Now we consider the case when $|W\cap \{u,v\}|=1$ and say $W\cap \{u,v\}$ is $u$. We know that 
$\Delta_{G'}(W)$ contains an edge of $P$. Since $W\subsetneq W'$ we have that $\Delta_{G}(W')$ also 
contains at least one edge of $P$. However, none  of these edges are present in 
$\Delta_{G'}(W')$. The only edges $G'$ misses are those in $T$ and thus the edges present in 
$\Delta_{G}(W')\cap E(P)$  must belong to $T$. Let $Z$ denote the set of end-points of edges in $\Delta_{G}(W')\cap E(P)$. Observe that, $Z\cap S'=Z\cap S$. Let $z_1$ denote the first vertex on $P$ belonging to $W'$ (or $W$)   when we walk along the path $P$ starting from $v$. Since $z_1$ and $u$ belongs to $W$ and $G[W]$ is connected we have that there is a path $Q_{uz_1}$ in $G[W]$. Let $P_{w_1v}$ denote the subpath of $P$ between  $w_2$ and $v$. This implies that the path $P'$ between $u$ and $v$ obtained by concatenating $P_{uz_1}P_{z_1v}$ does not intersect $\Delta_{G'}(W)$. Observe that $P'$ does not contain any edge of $\Delta_{G'}(W)$ and $F'\cup\{uv\}$.  This is a contradiction to our assumption that every path between $u$ and $v$ that exists in $G-(F'\cup\{uv\})$ must contain an edge of $C_{uv}$ such that it is present in $\Delta_{G'}(W)$. This completes the proof. 
\end{proof}

\begin{lemma}\label{lem:cographicnew}
\wss\ can be solved in time 
$2^{\cO(k)}\cdot ||M||^{\cO(1)}$ 
on cographic matroids.
\end{lemma}

\begin{proof}
Let  $(M,w,T,k)$ be an instance of \wss, where $M$ is a cographic matroid. 

First, we exhaustively apply  
Reduction Rules~\ref{rule:zero-rule}-\ref{rule:stoprule}. Thus, by Lemma~\ref{lem:prepr}, in polynomial  time we either solve the problem or obtain an equivalent instance, where $M$ has no loops, the weights of nonterminal elements are positive and $|T|\leq k$.  To simplify notations, we also denote the reduced instance by  $(M,w,T,k)$. Observe that $M$ remains to be cographic. It is well-known that given a cographic matroid, in polynomial time one can find a graph $G$ such that $M$ is isomorphic to the bond matroid $M^*(G)$~\cite{Seymour81}.

Next, we replace the weighted graph $G$ by the unweighted graph $G'$ as follows. For any nonterminal edge $uv$, we replace $uv$ by $w(uv)$ parallel edges with the same end vertices $u$ and $v$ if $w(uv)\leq k$, and we replace $uv$ by $k+1$ parallel edges if $w(uv)>k$. There is $F\subseteq E(G)\setminus T$ of weight at most $k$ such that $F$ spans $T$ in $G$ if and only if there is $F'\subseteq E(G')\setminus T$ of size at most $k$ such that $F'$ spans $T$ in $G'$. In other words, we have that 
$I=(M^*(G'),w',T,k)$, where $w'(e)=1$ for $e\in E(G')$, is an equivalent instance of the problem.
Notice that Reduction Rule~\ref{rule:terminalrule*} ({\bf Terminal circuit reduction rule}) for $M^*(G')$ can be restated as follows: if there is a minimal cut-set $R\subseteq T$, then contract any edge $e\in R$ in the graph $G'$.

It is well known that if $H$ is a forest on $n$ vertices then there are at least $\frac{n}{2}$ vertices of degree at most two. 
Suppose that $I$ is a yes-instance, and $F\subseteq E(G\rq{})\setminus T$ of size at most $k$ spans $T$.
 We know that in $G'-F$ every edge of $T$ is a bridge and we let the degree of a connected component $C$ 
of $G'-F-T$, denoted by $d^*(C,G'-F-T)$, be equal to the number of edges of $T$ it is incident to. Notice that if we shrink each connected component to a single vertex then we get a forest on at most $|T|+1\leq k+1$ vertices and thus there are at least $|T|/2$ components  such that  $d^*(C,G'-F-T)$ is at most two.  Let 
$I=(M^*(G'),w',T,k)$ denote our instance. Let $Q$ denote the set of end vertices of edges in $T$ and 
$Z\subseteq Q$.  
We assume by guessing  all possibilities in Step~3 that $Z$ has the following property: If $I$ is a yes-instance with a solution $F\subseteq E(G\rq{})\setminus T$, then 
$Z$ is the set of end vertices of terminals that are in the connected 
components $C$ of 
$G-F-T$ such that  $d^*(C,G'-F-T)\leq 2$. Initially $Z=\emptyset$.

Our algorithm {\sf ALG-CGM} takes as instance
$(I,Q,Z)$ and executes the following steps.  
\begin{enumerate}
\item While there is a minimal cut-set $R\subseteq T$ of $G$ do the following.
Denote by $Z_1\subseteq Z$ the set of $z\in Z$ such that $z$ is incident to exactly one $t\in T$, and let $Z_2\subseteq Z$ be the set of $z\in Z$ such that $z$ is incident to two edges of $T$. Clearly, $Z_1$ and $Z_2$ form a partition of $Z$.
Find a minimal cut-set $R\subseteq T$ and select $xy\in R$. Contract $xy$ and denote the contracted vertex by $z$. 
Set $T=T\setminus \{xy\}$ and recompute $Q$.  
If $x,y\in Z_1$ or if $x\notin Z$ or $y\notin Z$, then set $Z=Z\setminus\{x,y\}$. Otherwise, if $x,y\in Z$ and $\{x,y\}\cap Z_2\neq\emptyset$, set $Z=(Z\setminus \{x,y\})\cup\{z\}$. 
\item  If $Z$ is empty go to next step. Else, pick a vertex $s\in Z$ and finds all the $(s,Y,k)$ semi-important set $W$ in $G'-T$ such that 
$\Delta(W)\leq k$,  where $Y=W\setminus \{s\}$, using Lemma~\ref{lem:numSIC}. 
 For each such semi-important set $W$, we  call the algorithm {\sf ALG-CGM} on 
 $(M^*(G'-\Delta(W)),w',T,k-|\Delta(W)|)$, $W$ and 
$Z$.
By Lemma~\ref{lem:imp-new}, $I$ is a yes-instance if and only if one of the obtained instances is a yes-instance of \wss. 
\item Guess a subset $Z\subseteq Q$ with the property that  
if $I$ is a yes-instance with a solution $F\subseteq E(G\rq{})\setminus T$, then 
$Z$ is the set of end vertices of terminals that are in the connected 
components $C$ of 
$G-F-T$ such that  $d^*(C,G'-F-T)\leq 2$.
In particular, we do not include in $Z$ the vertices that are incident to at least 3 edges of $T$.
Now call {\sf ALG-CGM} on 
$(I,Q, Z)$. 
By the properties of the forest we know that the size of $|Z|\geq \frac{|T|}{2}$. 
\end{enumerate}

Notice that because 
on Step~2  there are no minimal cut-sets $R\subseteq T$, 
for each considered semi-important set $W$, $\Delta(W)$ is not empty. It means that the parameter decreases in each recursive call.
Moreover, by considering semi-important cuts of size $i$ for $i=\{1,\ldots,k\}$, we decrease the parameter by at least $i$. Let 
$\ell=|Q|-|Z|$.
Because there are at most $4^i(1+4^{i+1})$ semi-important sets of size $i$, we have the following recurrences for the algorithm: 
\begin{eqnarray}
T(\ell,k) & \leq & 2^{\ell } T\left(\ell-\frac{\ell}{4},k\right) \\
T(\ell,k) & \leq & \sum_{i=1}^{k} (4^i(1+4^{i+1})) T\left(\ell,k-i\right) 
\end{eqnarray}
By induction hypothesis we can show that the above recurrences solve to $16^{\ell}84^k$. Since $\ell \leq 2k$ we get that the above algorithm runs in time $2^{\cO(k)}\cdot n^{\cO(1)}$. This completes the proof. 
\end{proof}

\subsubsection{An algorithm for \rwss\ }
For \rwss\ we need the following variant of Lemma~\ref{lem:imp-new}. 

\begin{lemma}\label{lem:imp-restr-new}
Let $M^*(G)$ be the bond matroid of $G$, $T\subseteq E(G)$, $t^*\in T$, $e^*=uv\in E(G)$. Suppose
that $F\subseteq E(G)\setminus T$ spans $T$ and $F\setminus\{e^*\}$ spans $t^*$. 
Let also $x$ be an end vertex of an edge $xy$ of $T$ such that $x$ is either in a leaf block or in a degree two block in $G-F$, $Y$ is the set of end vertices of the edges of $T$ distinct from $x$,   
$G'=G-T$ and let $W=R_{G'-F}(x)$. 
If $u,v\notin R_{G'-F}(x)$, then  there is a  $(x,Y\cup\{u,v\},k)$-semi-important set 
$W'$  such that $|\Delta_{G'}(W')|\leq |\Delta_{G'}(W)|$ and 
for $F'=(F\setminus \Delta_{G'}(W))\cup \Delta_{G'}(W')$, it holds that $u,v\notin R_{G'-F\rq{}}(x)$,  $F\rq{}$
spans $T$ in $M^*(G)$ and $F\rq{}\setminus\{e^*\}$ spans $t$.
\end{lemma}

The proof of Lemma~\ref{lem:imp-restr-new} uses exactly the same arguments as the proof of  Lemma~\ref{lem:imp-new}.  The only difference is that we have to find a $(x,Y\cup\{u,v\},k)$-semi-important set
 $W\rq{}$ that separates $x$ and $\{u,v\}$. To guarantee it, we can replace $e^*$ by $k+1$ parallel edges for $k=|\Delta_{G'}(W')|$ with the end vertices  being $u$ and $v$ and use a $(x,Y\cup\{u,v\},k)$-semi-important set in the obtained graph. Modulo, this modification the proof is analogous to Lemma~\ref{lem:imp-new} and hence omitted.  Next we give the algorithm for \rwss\ on cographic matroids.

\begin{lemma}\label{lem:cographic-restr-new}
\textsc{Restricted} \SSp  can be solved in time 
$2^{\cO(k)}\cdot ||M||^{\cO(1)}$ 
on cographic matroids.
\end{lemma}

\begin{proof} The proof uses the same arguments as the proof of Lemma~\ref{lem:cographicnew}. Hence, we only sketch the algorithm here.

Let  $(M,w,T,k,e^*,t^*)$ be an instance of \rwss, where $M$ is a cographic matroid. First, we exhaustively apply Reduction Rules~\ref{rule:looprule} and~\ref{rule:zero-rule*}-\ref{rule:stoprule*}. 
Thus, by Lemma~\ref{lem:prepr}, in polynomial  time we either solve the problem or obtain an equivalent instance,  where $M$ has no loops, the weights of nonterminal elements are positive and $|T|\leq k+1$.
Notice that it can happen that $e^*$ is deleted by  Reduction Rules~\ref{rule:looprule} and~\ref{rule:zero-rule*}-\ref{rule:stoprule*}. For example, if $e^*$ is a loop then it can be deleted by Reduction Rule~\ref{rule:looprule}.  In this case we obtain an instance of \wss\ and can solve it using 
Lemma~\ref{lem:cographicnew}.  From now onwards we assume that $e^*$ is not delated by our reduction rules.

To simplify notations, we use  $(M,w,T,k,e^*,t^*)$ to denote the reduced  instance. If we started with cographic matroid then it remains so even after applying Reduction Rules~\ref{rule:looprule} and~\ref{rule:zero-rule*}-\ref{rule:stoprule*}. Furthermore, given $M$, in polynomial time we can find  a graph $G$ such that $M$ is isomorphic to  the bond matroid $M^*(G)$~\cite{Seymour81}.  Let $e^*=pq$.

Next, we replace the weighted graph $G$ by the unweighted graph $G'$ as follows. For any nonterminal edge $uv\neq e^*$,  if $w(uv)\leq k$ then we replace $uv$ by $w(uv)$ parallel edges with the same end vertices $u$ and $v$. On the other hand if  $w(uv)>k$  then we replace $uv$ by $k+1$ parallel edges. Recall that $w(e^*)=0$. Nevertheless, we replace $e^*$ by $k+1$ parallel edges with the end vertices $p$ and $q$ to forbid including $pq$ to a set that  spans $t^*$. 

Suppose that $(M,w,T,k,e^*,t^*)$ is a yes-instance and let $F\subseteq E(G)\setminus T$ is a solution. 
Recall that  in $G-F$ every edge of $T$ is a bridge and the degree of a connected component $C$ 
of $G'-F-T$, denoted by $d^*(C,G-F-T)$, is equal to the number of edges of $T$ it is incident to. Notice that if we shrink each connected component to a single vertex then we get a forest on at most $|T|+1\leq k+1$ vertices and thus there are at least $|T|/2$ components  such that  $d^*(C,G-F-T)$ is at most two. Only two components can contain $p$ or $q$. Hence, there are at least $|T|/2-2$ such components that do not include $p,q$. Moreover, there is at least one such component, because $F\setminus \{e\}$ spans $t^*$.
Let  $Q$ denote the set of end vertices of edges in $T$ and 
$Z\subseteq Q$. Initially $Z=\emptyset$, but we assume that  
$Z$ is the set of end vertices of terminals that are in the connected 
components $C$ of degree one of the graph obtained from $G\rq{}$ by deleting the edges of a solution and the terminals and, moreover, $p,q\notin C$.

Our algorithm {\sf ALG-CGM-restricted} takes as instance
$(G\rq{},T,k,Q,Z)$ and proceeds as follows.  
\begin{enumerate}
\item While there is a minimal cut-set $R\subseteq T$ of $G$ do the following.
Denote by $Z_1\subseteq Z$ the set of $z\in Z$ such that $z$ is incident to exactly one $t\in T$, and let $Z_2\subseteq Z$ be the set of $z\in Z$ such that $z$ is incident to two edges of $T$. Clearly, $Z_1$ and $Z_2$ form a partition of $Z$.
Find a minimal cut-set $R\subseteq T$ and select $xy\in R$ such that $xy\neq t^*$ if $R\neq\{t^*\}$ and let $xy=t^*$ otherwise.
Contract $xy$ and denote the obtained vertex $z$. 
Set $T=T\setminus \{xy\}$ and recompute $W$.  
If $x,y\in Z_1$ or if $x\notin Z$ or $y\notin Z$, then set $Z=Z\setminus\{x,y\}$. Otherwise, if $x,y\in Z$ and $\{x,y\}\cap Z_2\neq\emptyset$, set $Z=(Z\setminus \{x,y\})\cup\{z\}$. 

\item If $t^*\notin T$, then delete the edges $pq$. Notice that $t^*\notin T$ only if we already constructed a set that spans $t^*$. Hence, it is safe to get rid of $e^*$ of weight 0. 

\item  If $Z$ is empty go to the next step. Else, pick a vertex $s\in Z$ and find all the $(s,Y,k)$ semi-important sets $W$ in $G'-T$ such that 
$\Delta(W)\leq k$,  where 
\[Y=
\begin{cases}
(Q\setminus \{s\})\cup\{p,q\},&\mbox{if }t^*\in T, \\
Q\setminus \{s\},&\mbox{if }t^*\notin T,
\end{cases}\]
using Lemma~\ref{lem:numSIC}.
Notice that if $t^*\in T$, then there are $k+1$ copies of $pq$. Hence, $W$ separates $s$ from $p$ and $q$. For each such semi-important set $W$, we  call the algorithm {\sf ALG-CGM-restricted} on 
$(G'-\Delta(W),T,k-|\Delta(W)|,Q,Z)$.  We use Lemma~\ref{lem:imp-restr-new} to argue that the branching step is safe.

\item Guess a subset $Z\subseteq Q$ with the property that  $Z$ is the set of end vertices of terminals that are in the connected 
components $C$ of degree at most two of the graph obtained from $G\rq{}$ by the deletion of edges of a solution and the terminals and, moreover, $p,q\notin C$.
In particular, we do not include in $Z$ the vertices that are incident to at least 3 edges of $T$.
Now call {\sf ALG-CGM-restricted} on $(G\rq{},T,k,W,Z)$. Notice, that by the properties of the forest we know that $Z\neq\emptyset$ and the size of $|Z|\geq \frac{|T|}{2}-2$. 
\end{enumerate}

Notice that because of Step~3  there are no minimal cut-sets $R\subseteq T$ and thus for each considered semi-important set $W$, $\Delta(W)$ is not empty. It means that the parameter decreases in each recursive call.
Moreover, by considering semi-important cuts of size $i$ for $i=\{1,\ldots,k\}$, we decrease the parameter by at least $i$. Let 
$\ell=|Q|-|Z|$. 
Because there are at most $4^i(1+4^{i+1})$ semi-important sets of size $i$, we have the following recurrences for the algorithm: 
\begin{eqnarray}
T(\ell,k) & \leq & 2^{\ell } T\left(\ell-\frac{\ell}{4}+2,k\right) \\
T(\ell,k) & \leq & \sum_{i=1}^{k} (4^i(1+4^{i+1})) T\left(\ell,k-i\right) 
\end{eqnarray}
As in the proof of Lemma~\ref{lem:cographicnew} using induction hypothesis we can show that the above recurrences solve to $16^{\ell}84^k$. Since $\ell \leq 2k+1$ we get that the above algorithm runs in time $2^{\cO(k)}\cdot n^{\cO(1)}$. This completes the proof. 
\end{proof}

\section{Solving Space Cover for regular matroids}\label{sec:alg}
In this section we conjure all that have developed so far and design an algorithm for \wss\  on regular matroids, running in time 
$2^{\cO(k)}\cdot ||M||^{\cO(1)}$. 
To give a clean presentation of our algorithm we have divided the section into three parts. We first give some generic steps, followed by steps when matroid in consideration is either  graphic and cographic and ending with a result that ties them all. 

Let $(M,w,T,k)$ be the given instance of \wss.  First, we exhaustively apply  
Reduction Rules~\ref{rule:zero-rule}-\ref{rule:stoprule}. Thus, by Lemma~\ref{lem:prepr}, in polynomial  time we either solve the problem or obtain an equivalent instance, where $M$ has no loops and the weights of nonterminal elements are positive.  To simplify notations, we also denote the reduced 
instance by  $(M,w,T,k)$. We say that a matroid $M$  is {\em \simple} if it can be obtained from $R_{10}$ by adding parallel elements or $M$ is graphic or cographic. If $M$ is a \simple\ matroid  then we can
solve  \wss\ using  Lemmas~\ref{obs:r10}, or~\ref{lem:graphic} or~\ref{lem:cographicnew} respectively in time 
$2^{\cO(k)}\cdot ||M||^{\cO(1)}$.  
This results in the following lemma.

\begin{lemma}
\label{lem:simpleAlgo}
Let  $(M,w,T,k)$ be an instance of \wss. 
If $M$ is a \simple\ matroid then \wss\ can be solved in time  
$2^{\cO(k)}\cdot ||M||^{\cO(1)}$. 
\end{lemma}

From now  onwards we assume that the matroid  $M$ in the instance $(M,w,T,k)$ is not \simple. 
Now using Corollary~\ref{thm:decomp-good},  we find a conflict tree $\mathcal{T}$. 
Recall that the set of nodes of $\mathcal{T}$ is the collection of basic matroids $\mathcal{M}$ and the edges correspond to $1$-, $2-$ and 3-sums.  The key observation is that $M$ can be constructed from $\mathcal{M}$ by performing the sums corresponding to the edges of $\mathcal{T}$ in an arbitrary order. Our algorithm is based on performing {\em bottom-up} traversal of the tree $\mathcal{T}$.  We select an arbitrarily {\em node $r$ as the root} of $\mathcal{T}$. Selection of $r$,  as the root,  defines the natural parent-child, descendant and ancestor relationship on the nodes of $\mathcal{T}$. We say that $u$ is a \emph{sub-leaf} if its children are leaves of $\mathcal{T}$. Observe that there always exists a sub-leaf in a tree on at least two nodes. Just take a node which is not a leaf and is  farthest from the root.  Clearly, this node can be found in polynomial time.  
\begin{quote}{\it
Throughout, this section 
we fix a sub-leaf of $\mathcal{T}$ --  a basic matroid  $M_s$. 
We say that a child of $M_s$ is a $1$-, $2$- or $3$\emph{-leaf}, respectively, if the edge between $M_s$ and the leaf corresponds to $1$-, $2$- or 3-sum, respectively.  }
\end{quote}

\noindent 
We first  modify the decomposition by an exhaustive application of the following rule. 
\begin{reduction}[{\bf Terminal flipping rule.}] 
\label{rule:term-flip-rule}
 If there is a child $M_\ell$ of a sub-leaf $M_s$ such that there is $e\in E(M_s)\cap E(M_\ell)$ that is parallel to a terminal $t\in E(M_\ell)\cap T$ in $M_\ell$, then delete $t$ from $M_\ell$ and add $t$ to $M_s$ as an element parallel to $e$. 
\end{reduction}

\noindent 
The safeness of Reduction Rule~\ref{rule:term-flip-rule} follows from the following observation.

\begin{observation}[\cite{DinitzK14}]
\label{obs:flip}
Let $M=M_1\oplus M_2$. Suppose that  there is $e'\in E(M_2)\setminus E(M_1)$ such that $e'$ is parallel to $e\in E(M_1)\cap E(M_2)$. Then $M=M_1'\oplus M_2'$, where $M_1'$ is obtained from $M_1$ by adding a new element $e'$ parallel $e$ and $M_2'$ is obtained from $M_2'$ by the deletion of $e'$. 
\end{observation}

Proof of Observation~\ref{obs:flip} is implicit in~\cite{DinitzK14}. Furthermore Reduction Rule~\ref{rule:term-flip-rule} can be applied in polynomial time. Notice also allowed to a matroid obtained from $R_{10}$ by adding parallel elements to be a basic maroid of a decomposition.
Thus, we get the following lemma. 

\begin{lemma}
\label{lem:tfrulesafe}
Reduction Rule~\ref{rule:term-flip-rule}  is safe and can be applied in polynomial time. 
\end{lemma} 

From now we assume that there is no child $M_\ell$ of $M_s$ such that there exists an element 
$e\in E(M_s)\cap E(M_\ell)$ that is parallel to a terminal $t\in E(M_\ell)\cap T$ in $M_\ell$. 
In what follows we do a bottom-up traversal of $\cal T$ and at each step we delete one of the child of $M_s$. A child of $M_s$ is  deleted either because of  an application of a reduction rule or 
because of  recursively solving the problem on a smaller sized tree. It is possible that, while recursively solving the problem, we could possibly modify (or replace) $M_s$  to encode some auxiliary information 
that we have already computed while solving the problem.  We start by giving some generic steps that do not depend on the types of either $M_s$ or its child. 
{\em Throughout the section, given the conflict tree  $\cal T$,  we 
denote  by $M_{\cal T}$ the matroid   defined by $\cal T$. }
\subsection{A few generic steps}\label{sec:init}
We start by giving a reduction rule that is useful when we have $1$-leaf. The reduction rule is as follows.

\begin{reduction}[{\bf $1$-Leaf reduction rule}]
\label{rule:one-leaf}
 If there is a child $M_{\ell}$ of $M_s$ that is a 1-leaf, then do the following.
\begin{itemize} 
\setlength{\itemsep}{-2pt}
\item[(i)] If $E(M_\ell)\cap T=\emptyset$, then delete $M_\ell$ from $\mathcal{T}$.

\item[(ii)] If $E(M_\ell)\cap T\neq \emptyset $, then find the minimum $k'\leq k$ such that $(M_\ell,w_\ell,T\cap E(M_\ell),k')$ is a \yesinstance\ of \wss\ using  Lemmas~\ref{obs:r10}, or~\ref{lem:graphic} or~\ref{lem:cographicnew},  respectively, depending on which primary matroid $M_\ell$ is. Here, $w_\ell$ is the restriction of $w$ on $E(M_\ell)$. 
If  $(M_\ell,w_\ell,T\cap E(M_\ell),k')$ is a \noinstance\ for every $k'\leq k$ then we return \no.  
Let $\mathcal{T}'$ be obtained from $\mathcal{T}$ by deleting the node  
$M_\ell$. Furthermore, for simplicity, let $M_{{\cal T}'}$ be denoted by $M'$, 
restriction of $w$ to $E(M_{{\cal T}'})$ by $w'$ and  $T\cap E(M_{{\cal T}'})$ be denoted by $T'$. 
Our new instance is $(M',w', T',k-k')$. 
\end{itemize}
\end{reduction}
Safeness of the reduction rule follows by the definition of  $1$-sum, and it can be applied in time
$2^{\cO(k)}\cdot ||M||^{\cO(1)}$. 
Thus we get the following result. 

\begin{lemma}\label{lem:red-1}
Reduction Rule~\ref{rule:one-leaf} is safe and can be applied in 
$2^{\cO(k)}\cdot ||M||^{\cO(1)}$
time. 
\end{lemma}

\subsubsection{Handling $2$-leaves}
For $2$-leaves, we either reduce a leaf or apply a recursive procedure based on whether the leaf contains a terminal or not.

\begin{reduction}[{\bf $2$-Leaf reduction rule}]
\label{rule:two-leaf}
If there is a child $M_{\ell}$ of $M_s$ that is a 2-leaf with $E(M_s)\cap E(M_\ell)=\{e\}$ and $T\cap E(M_\ell)=\emptyset$, then find the minimum $k'\leq k$ such that $(M_\ell,w_\ell,\{e\},k')$ is a yes-instance of \wss\  using  Lemmas~\ref{obs:r10}, or~\ref{lem:graphic} or~\ref{lem:cographicnew}, respectively,  depending on which primary matroid $M_\ell$ is. Here, $w_\ell(e')=w(e')$ for $e'\in E(M_\ell)\setminus\{e\}$ and $w_\ell(e)=0$. 
If  $(M_\ell,w_\ell,\{e\},k')$ is a \noinstance\ for every $k'\leq k$ then we set $k'=k+1$.  
Let $\mathcal{T}'$ be obtained from $\mathcal{T}$ by deleting the node  
$M_\ell$. Furthermore, for simplicity, let $M_{{\cal T}'}$ be denoted by $M'$.
We define $w'$ on  $E(M')$ as follows: 
for every $e^*\in E(M_{{\cal T}'})$, $e^*\neq e$, set $w'(e^*)=w(e^*)$ and let $w'(e)=k'$.  Our new instance is $(M',w', T,k)$. 
\end{reduction}

\begin{lemma}\label{lem:red-2}
Reduction Rule~\ref{rule:two-leaf} is safe and can be applied 
$2^{\cO(k)}\cdot ||M||^{\cO(1)}$ time. 
\end{lemma}

\begin{proof}
To show that the rule is safe, denote by $M'$ the matroid defined by $\mathcal{T}'=\mathcal{T}-M_\ell$ and let $w'(e')=w(e')$ for $e'\in E(M')\setminus\{e\}$ and $w'(e)=k'$. By {\bf 2-Leaf reduction rule}, there is a cycle $C$ of $M_\ell$ such that $e\in C$ and the weight $w(C\setminus\{e\})=k'$ is minimum among all cycles that include $e$.

Suppose that $(M,w,T,k)$ is a yes-instance of \wss. Let $F\subseteq E(M)\setminus T$ be a set of weight at most $k$ that spans $T$. 
If $F\cap E(M_\ell)=\emptyset$, then $F$ spans $T$ in $M'$ and because $e\notin F$, the weight of $F$ is the same as before. Hence, $(M',w',T,k)$ is a yes-instance. Assume that $F\cap E(M_\ell)\neq\emptyset$. Let $F'=(F\cap E(M')\cup\{e\}$. 
For each $t\in T$, there is a circuit $C_t$ of $M$ such that $t\in C_t\subseteq F\cup\{t\}$. Because $F\cap E(M_\ell)\neq\emptyset$, there is $t\in T$ such that 
$C_t\cap E(M_\ell)\neq\emptyset$. By the definition of 2-sums,  there are cycles $C_t'$ of $M'$ and $C_t''$ of $M_\ell$ such that $C_t=C_t'\bigtriangleup C_t''$ and we have that  $e\in C_t'\cap C_t''$, because $C_t$ is a circuit, i.e., an inclusion-minimal nonempty cycle. 
Since $w(C_t''\setminus\{e\})\geq w(C\setminus\{e\})$, we have that $w(F')\leq k$. To show that $F'$ spans $T$, consider $t\in T$ and a cycle $C_t$ of $M$ such that $t\in C_t\subseteq F\cup\{t\}$. If $C_t\subseteq E(M')$, then $C_t\subseteq F'\cup\{t\}$ and $F'$ spans $t$ in $M'$. If $C_t\cap E(M_\ell)\neq\emptyset$, then there are cycles $C_t'$ of $M'$ and $C_t''$ of $M_\ell$ such that $e\in C_t'\cap C_t''$ and $C_t=C_t'\bigtriangleup C_t''$. Because $C_t'\subseteq  F'\cup\{t\}$, we have that $F'$ spans $t$. 

Assume now that $(M',w',T,k)$ is a yes instance. Let $F'\subseteq E(M')\setminus T$ be a set of weight at most $k$ that spans $T$ in $M'$. If $e\notin F'$, then $F'$ spans $T$ in $M$ and   $(M,w,T,k)$ is a yes-instance. Suppose that $e\in F'$. Let $F=F'\bigtriangleup C$. Clearly, $w(F)=w(F')\leq k$. We have to show that $F$ spans $T$. Let $t\in T$.
There is a cycle $C_t'$ in $M'$ such that $t\in C_t'\subseteq F'\cup\{t\}$. If $e\notin C'$, then $C_t'\subseteq F\cup\{t\}$ and $F$ spans $t$. If $e\in C_t'$, then for $C_t=C_t'\bigtriangleup C$, we have that $t\in C_t\subseteq F\cup\{t\}$ and it implies that $F$ spans $t$.  

The rule can be applied in time $2^{\cO(k)}\cdot ||M||^{\cO(1)}$ by Lemma~\ref{lem:simpleAlgo}. In fact, it can be done in polynomial time, because we are solving \wss{} for the sets of terminal of size one. It is easy to see that if $M_\ell$ is graphic, then the problem can be reduced to finding a shortest path, and if $M_\ell$ is cographic, then we can reduce it to the minimum cut problem. 
\end{proof}

Reduction Rule~\ref{rule:two-leaf} takes care of the case when $M_\ell$ has no terminal. If it has a terminal then we recursively solve the problem as described below in Branching Rule~\ref{brule:2lb} and if any of these returns yes then we return that the given instance is a \yesinstance.

\begin{branchrule}[{\bf 2-Leaf branching}]
\label{brule:2lb}
 If there is a child $M_{\ell}$ of $M_s$ that is a $2$-leaf with $E(M_s)\cap E(M_\ell)=\{e\}$ and $T\cap E(M_\ell)=T_\ell\neq\emptyset$, then do the following.
Let $M'$ the matroid defined by $\mathcal{T}'=\mathcal{T}-M_\ell$ and let $T'=T\setminus T_\ell$. Consider the following three branches.
\begin{itemize}
\setlength{\itemsep}{-2pt}
\item[(i)] Let $w'(e')=w(e')$ for $e'\in E(M')\setminus\{e\}$ and $w'(e)=0$. Define $w_\ell(e')=w(e')$ for $e'\in E(M_\ell)\setminus\{e\}$ and $w_\ell(e)=0$. 
Find the minimum $k_1\leq k$ such that $(M_\ell,w_\ell,T_\ell\cup\{e\},k_1)$ is a yes-instance of \wss\  using  Lemmas~\ref{obs:r10}, or~\ref{lem:graphic} or~\ref{lem:cographicnew},  respectively, depending on the type of $M_\ell$. 
 If  $(M_\ell,w_\ell,T_\ell\cup\{e\},k_1)$ is a \noinstance\ for every $k_1\leq k$, then we return \no\ 
 and stop. 
 Otherwise, solve the problem on the instance $(M',w',T',k-k_1)$.
\item[(ii)] Let $w'(e')=w(e')$ for $e'\in E(M')\setminus\{e\}$ and $w'(e)=0$. Define $w_\ell(e')=w(e')$ for $e'\in E(M_\ell)\setminus\{e\}$ and $w_\ell(e)=0$.  Find the minimum $k_2\leq k$ such that $(M_\ell,w_\ell,T_\ell,k_2)$ is a yes-instance of \wss\  using  Lemmas~\ref{obs:r10}, or~\ref{lem:graphic} or~\ref{lem:cographicnew},  respectively, depending on the type of $M_\ell$.  
If  $(M_\ell,w_\ell,T_\ell,k_2)$ is a \noinstance\ for every $k_2\leq k$, then we return  \no\ and stop. 
Otherwise, solve the problem on the instance  $(M',w',T'\cup\{e\},k-k_2)$.
\item[(iii)] Let $w'(e')=w(e')$ for $e'\in E(M')\setminus\{e\}$ and $w'(e)=k+1$. Define $w_\ell(e')=w(e')$ for $e'\in E(M_\ell)\setminus\{e\}$ and $w_\ell(e)=k+1$.  Find the minimum $k_3\leq k$ such that 
$(M_\ell,w_\ell,T_\ell,k_3)$ is a yes-instance of \wss\  using  Lemmas~\ref{obs:r10}, or~\ref{lem:graphic} or~\ref{lem:cographicnew},  respectively, depending on the type of $M_\ell$. 
 If  $(M_\ell,w_\ell,,k_3)$ is a \noinstance\ for every $k_3\leq k$,  then we return  \no\ and stop. 
Otherwise, solve the problem on the instance  $(M',w',T',k-k_3)$.
\end{itemize}
\end{branchrule}

\begin{lemma}\label{lem:branch-2}
Branching Rule~\ref{brule:2lb} is exhaustive and in each recursive call the parameter strictly reduces.  
Each call of the rule takes $2^{\cO(k)}\cdot ||M||^{\cO(1)}$ time.
\end{lemma}

\begin{proof}
To show correctness, assume first that $(M,w,T,k)$ is a yes-instance of \wss. Let $F\subseteq E(M)\setminus T$ be a set of weight at most $k$ that spans $T$. 
Without loss of generality we assume that $F$ is inclusion-minimal and, therefore, $F$ is independent by Observation~\ref{obs:spn}. 
For each $t\in T$, there is a circuit $C_t$ of $M$ such that $t\subseteq C_t\subseteq F\cup\{t\}$. We have the following three cases.

\medskip
\noindent 
{\bf Case~1.} There is $C_t$ such that $t\in T'$ and $C_t\cap E(M_\ell)\neq\emptyset$.  Let $F_\ell=F\cap E(M_\ell)$ and $F'=(F\cap E(M'))\cup\{e\}$. 
We claim that $F_\ell$ spans $T_\ell\cup\{e\}$ in $M_\ell$ and $F'$ spans $T'$ in $M'$.

First, we show that  $F_\ell$ spans $T_\ell\cup\{e\}$ in $M_\ell$. Since there is a circuit $C_t$ such that $t\in T'$ and $C_t\cap E(M_\ell)\neq\emptyset$, there are cycles
$C_t'$ of $M'$ and $C_t''$ of $M_\ell$ such that $C_t=C_t'\bigtriangleup C_t''$ and $e\in C_t'\cap C_t''$. Because $e\in C_t''$ and $C_t''\setminus\{e\}\subseteq F_\ell$, we have that $F_\ell$ spans $e$ in $M_\ell$. Let $t'\in T_\ell$. Since $F$ spans $t'$ in $M$, there is a cycle $C_{t'}$ of $M$ such that $t'\in C_{t'}\subseteq F\cup\{t'\}$. If $C_{t'}\setminus {t'}\subseteq E(M_{\ell})$, then $F_\ell$ spans $t'$, because $C_{t'}\setminus\{t'\}\subseteq F_\ell$. Suppose that $C_{t'}\cap E(M')\neq \emptyset$. Then by the definition of 2-sum, there are  cycles $C_{t'}'$ of $M'$ and $C_{t'}''$ of $M_\ell$ such that $e\in C_{t'}'\cap C_{t'}''$ and $C_{t'}=C_{t'}'\bigtriangleup C_{t'}''$.  Consider $C=C_{t}''\bigtriangleup C_{t'}''$. By Observation~\ref{obs:symm}, $C$ is a cycle.  As $C\setminus\{e\}\subseteq F_\ell$,
$e\in C_{t'}''\cap C_t''$
and $t'\notin C_t''$, we obtain that $C$ is a cycle of $M_\ell$ and $t'\in C\subseteq F_\ell\cup\{t'\}$. Therefore, $F_\ell$ spans $t'$.

To prove that $F'$ spans $T'$ in $M'$, consider $t'\in T'$. Since $F$ spans $t'$ in $M$, there is a circuit $C_{t'}$ of $M$ such that $t'\in C_{t'}\subseteq F\cup\{t'\}$. If $C_{t'}\setminus {t'}\subseteq E(M')$, then $F'$ spans $t'$, because $C_{t'}\setminus\{t'\}\subseteq F'$. 
Suppose that $C_{t'}\cap E(M_\ell)\neq \emptyset$. Then by the definition of 2-sum, there are  cycles $C_{t'}'$ of $M'$ and $C_{t'}''$ of $M_\ell$ such that $e\in C_{t'}'\cap C_{t'}''$ and $C_{t'}=C_{t'}'\bigtriangleup C_{t'}''$. Observe that $C_{t'}'\setminus\{t'\}\subseteq F'$ and, therefore, $F'$ spans $t'$ in $M'$.

Since $F_\ell$ spans $T_\ell\cup\{e\}$ in $M_\ell$, $w(F_\ell)\geq k_1$. Because $w(F')+w(F_\ell)=w(F)\leq k$ if the weight of $e$ in $M'$ is 0, 
 $w(F')\leq k-k_1$ in this case. Hence, $(M',w',T',k-k_1)$ is a yes-instance for the first branch. 

\medskip
\noindent 
{\bf Case~2.} There is $C_t$ such that $t\in T_\ell$ and $C_t\cap E(M')\neq\emptyset$. This case is symmetric to Case~1, and by the same arguments, we show that $(M',w',T'\cup\{e\},k-k_2)$ for the second branch.

\medskip
Otherwise, we have the remaining case.

\medskip
\noindent 
{\bf Case~3.} For any $t\in T'$, $C_t\subseteq E(M')\setminus\{e\}$, and for any $t\in T_\ell$, $C_t\subseteq E(M_\ell)\setminus\{e\}$.
 Let $F_\ell=F\cap E(M_\ell)$ and $F'=(F\cap E(M'))$. Observe that $F_\ell$ spans $T_\ell$ in $M_\ell$ and $F'$ spans $T'$ in $M'$.  In particular, $w(F_\ell)\geq k_3$. Since $w(F')+w(F_\ell)=w(F)\leq k$, $(M',w',T',k-k_3)$ is a yes-instance for the third branch.

\medskip
Suppose now that  we have a yes-answer for one of the branches. We consider 3 cases depending on the branch.

\medskip
\noindent 
{\bf Case~1.} $(M',w',T',k-k_1)$ is a yes-instance for the first branch.
Let $F_\ell\subseteq E(M_\ell)$ be a set of weight at most $k_1$ that spans $T_\ell\cup\{e\}$ in $M_\ell$ and let $F'$
be a set of weight at most $k-k_1$ that spans $T'$ in $M'$. Consider $F=F'\bigtriangleup F_\ell$. Clearly, $w(F)\leq k$. We claim that $F$ spans $T$. Let $t\in T$. Suppose that $t\in T_\ell$. Notice that $e\notin F_\ell$, as $e$ is a terminal in the instance $(M_\ell,w_\ell,T_\ell\cup\{e\},k_1)$. It implies that $F_\ell$ spans $t$ in $M$. Assume now that $t\in T'$. Since $F'$ spans $t$, there is a cycle $C_t$ of $M'$ such that $t\in C_t\subseteq F'\cup\{t\}$ . If $e\notin C_t$, then $C_t\setminus\{t\}$ and, therefore, $F$ spans $t$ in $M$. Suppose that $e\in C_e$. 
The set $F_\ell$ spans $e$ in $M_\ell$. Hence, there is a cycle $C_e$ of $M_\ell$ such that $e\in C_e\subseteq F_\ell\cup\{e\}$. Let $C_{t}'=C_t\bigtriangleup C_e$. By definition, $C_{t}'$ is a cycle of $M$. 
Because $t\in C_{t}'$ and $e\notin  C_{t'}$, we have that $C_{t}'\setminus\{t\}$ spans $t$. As $C_t'\subseteq F$, $F$ spans $t$. Because $F$ is a set of weight at most $k$ that spans $T$, $(M,w,T,k)$  is a yes-instance.

\medskip
\noindent
{\bf Case~2.} $(M',w',T'\cup\{e\},k-k_2)$ is a yes-instance for the second branch. This case is symmetric to Case~1, and  we use the same arguments to show that  $(M,w,T,k)$  is a yes-instance.

\medskip
\noindent
{\bf Case~3.} $(M',w',T',k-k_3)$ is a yes-instance for the third branch. 
Let $F_\ell\subseteq E(M_\ell)$ be a set of weight at most $k_1$ that spans $T_\ell$ in $M_\ell$ and let $F'$
be a set of weight at most $k-k_1$ that spans $T'$ in $M'$. Notice that $e\notin F_\ell$ and $e\notin F'$, because the weight of $e$ is $k+1$ in $M_\ell$ and $M'$.  
Let  $F=F'\cup F_\ell$. Clearly, $w(F)\leq k$.  Let $t\in T$. Then $F_\ell$ spans $t$ in $M$.  If $t\in T'$, then $F'$ spans $t$ in $M$. Hence, $F$ spans $T$. Therefore,
$(M,w,T,k)$  is a yes-instance.

\medskip
Notice that $M_\ell$ has no nonterminal elements of zero weight for the first and third branches and the elements of $T_\ell$ are not loops, because of the application of the reduction rules. Hence,
$k_1,k_3\geq 1$.  For the second branch, $e$ has the zero weight, but $F_\ell$ has no terminals parallel to $e$, because of  {\bf Terminal flipping rule}, hence, $k_2\geq 1$ as well. We conclude that all recursive calls are done for the parameters that are strictly lesser that $k$.

The claim that each call of the rule (without recursive steps) takes $2^{\cO(k)}\cdot ||M||^{\cO(1)}$ time follows from Lemma~\ref{lem:simpleAlgo}.
\end{proof}

\subsubsection{Handling $3$-leaves}
In this section we assume that all the children of $M_s$ are $3$-leaves. The analysis of this cases is done along the same lines as for the case of $2$-leaves. However, this case is  significantly more complicated. 

\begin{observation}\label{obs:even}
Let $M$ be a matroid obtained from $R_{10}$ by adding some parallel elements. Then any circuit of $M$ has even size.
\end{observation}

It immediately implies that $M_s$ and its children are  graphic or cographic matroids. 

For $3$-sums, it is convenient to make the following observation.

\begin{observation}\label{obs:through-sum} 
Let $M=M_1\oplus_3 M_2$. If $C$ is a cycle of $M$, then there are cycles $C_1$ and $C_2$ of $M_1$ and $M_2$ respectively such that $C=C_1\bigtriangleup C_2$ and either $C_1\cap C_2=\emptyset$ or $|C_1\cap C_2|=1$. Moreover, if $C$ is a circuit of $M$, then either $C$ is a circuit of $M_1$ or $M_2$, or 
there are circuits $C_1$ and $C_2$ of $M_1$ and $M_2$ respectively such that 
$C=C_1\bigtriangleup C_2$ and  $|C_1\cap C_2|=1$. 
\end{observation}

\begin{proof}
Let $Z=C_1\cap C_2$. Recall that $Z$ is a circuit of $M_1$ and $M_2$. Let $C=C_1\bigtriangleup C_2$ and $|C_1\cap C_2|\geq 2$. Consider $C_1'=C_1\bigtriangleup Z$ and $C_2'=C_2\bigtriangleup Z$. We have that $C_1'$ and $C_2'$ are cycles of $M_1$ and $M_2$ respectively by Observation~\ref{obs:symm} and $|C_1'\cap C_2'|\leq 1$. It remains to notice that $C=C_1'\bigtriangleup C_2'$. The second claim immediately follows from the fact that a circuit is an inclusion-minimal nonempty cycle.
\end{proof}

We use Observation~\ref{obs:through-sum} to analyze the structure of a solution of \wss\  for matroid sums. If $M=M_1\oplus_3 M_2$ and 
for $t\in T$, a circuit $C$ such that $t\in C\subseteq F\cup\{t\}$ for a solution $F$ has nonempty intersection with $E(M_1)$ and $E(M_2)$, then 
$C=C_1\bigtriangleup C_2$ for cycles $C_1$ and $C_2$ of $M_1$ and $M_2$ respectively and, moreover, it could  be assumed that $C_1$ and $C_2$ are circuits.
By Observation~\ref{obs:through-sum}, we can always assume that $C_1\cap C_2=\{e\}$ for $e\in E(M_1)\cap E(M_2)$. Using this assumption, we say that $C$ \emph{goes through} $e$ in this case.

We also need the following observation about circuits of size $3$.

\begin{observation}\label{obs:triangle}
Let $M$ be a binary matroid, $w\colon E(M)\rightarrow\mathbb{N}_0$. Let also $C=\{e_1,e_2,e_3\}$ be a circuit of $M$. Suppose that $F\subseteq E(M)\setminus C$ is a set of minimum weight such that $M$ has circuits (cycles) $C_1$ and $C_2$ such that $e_1\in C_1\subseteq F\cup\{e_1\}$ and $e_2\in C_2\subseteq F\cup\{e_2\}$. 
Then $F$ is a subset of $E(M)\setminus C$ of minimum weight such that for each $i\in\{1,2,3\}$, $M$ has a circuit (cycle) $C_i$ such that $e_i\in C_i\subseteq F\cup\{e_i\}$. Moreover, for any distinct $i,j\in\{1,2,3\}$, $F$ is a subset of minimum weight of $E(M)\setminus C$ such that $M$ has 
 circuits (cycles) $C_i$ and $C_j$ such that $e_i\in C_i\subseteq F\cup\{e_i\}$ and $e_j\in C_j\subseteq F\cup\{e_j\}$.
\end{observation}

\begin{proof}
Let $C'=C_1\bigtriangleup C_2\bigtriangleup C$. Because $M$ is binary, $C'$ is a cycle by Observation~\ref{obs:symm}. 
Since $\{e_1\}= C\cap C_1$, $\{e_2\}= C\cap C_2$ and $e_3\notin C_1\cup C_2=F$, $C'$ contains a circuit  
$C_3$ such that $e_3\in C_3\subseteq C'\subseteq F\cup\{e_3\}$. Hence, the first claim holds by symmetry.
Also by symmetry,  the second claim is fulfilled.
\end{proof}

If  a child of $M_s$ has terminals,  then we recursively solve the problem as described below in Branching Rule~\ref{brule:3lb} and if any of these returns yes then we return that the given instance is a \yesinstance.

\begin{branchrule}[{\bf $3$-Leaf branching}]
\label{brule:3lb}
If there is a child $M_{\ell}$ of $M_s$ that is a $3$-leaf with $E(M_s)\cap E(M_\ell)=Z$ and $T\cap E(M_\ell)=T_\ell\neq\emptyset$, then let $M'$ the matroid defined by $\mathcal{T}'=\mathcal{T}-M_\ell$ and let $T'=T\setminus T_\ell$. We set $w'(e)=w(e)$ for $e\in E(M')\setminus Z$ and $w_\ell(e)=w(e)$ for $e\in E(M_\ell)\setminus Z$.
Consider the following branches of six types.
\begin{itemize}
\setlength{\itemsep}{-2pt}
\item[(i)] Let $w_\ell(e_h)=k+1$ for $h\in\{1,2,3\}$.
For each $i\in\{1,2,3\}$ do the following. 
Set $w'(e_i)=0$ and $w'(e_h)=k+1$ for $h\in\{1,2,3\}$ such that $h\neq i$. Find the minimum  
$k_i^{(1)}\leq k$ such that $(M_\ell,w_\ell,T_\ell\cup\{e_i\},k_i^{(1)})$ is a yes-instance of \wss\ using  Lemmas~\ref{lem:graphic} or~\ref{lem:cographicnew}, respectively, depending on the type of $M_\ell$. 
If  $(M_\ell,w_\ell,T_\ell\cup\{e_i\},k_i^{(1)})$ is a \noinstance\ for every $k_i^{(1)}\leq k$,  then  
 we return \no\  and stop.
Otherwise, solve the problem on the instance $(M',w',T',k-k_i^{(1)})$.
\item[(ii)] Let $w_\ell(e_h)=k+1$ for $h\in\{1,2,3\}$.
Set $w'(e_1)=w'(e_2)=0$ and $w'(e_3)=k+1$. Find the minimum  
$k^{(2)}\leq k$ such that $(M_\ell,w_\ell,T_\ell\cup\{e_1,e_2\},k^{(2)})$ is a yes-instance of \wss\ using  Lemmas~\ref{lem:graphic} or~\ref{lem:cographicnew}, respectively, depending on the type of $M_\ell$. 
 If   $(M_\ell,w_\ell,T_\ell\cup\{e_1,e_2\},k^{(2)})$ is a \noinstance\ for every $k^{(2)}\leq k$,  
then  we return \no\  and stop.
Otherwise,  solve the problem on the instance $(M',w',T',k-k^{(2)})$.
\item[(iii)] For any two distinct $i,j\in\{1,2,3\}$, do the following. Let $h\in\{1,2,3\}$ such that $h\neq i,j$. 
Set $w_\ell(e_i)=0$ and $w_\ell(e_j)=w_\ell(e_h)=k+1$. Let $w'(e_j)=0$  and $w'(e_i)=w'(e_h)=k+1$.
Find the minimum  
$k_{ij}^{(3)}\leq k$ such that $(M_\ell,w_\ell,T_\ell\cup\{e_j\},k_{ij}^{(3)},e_i,e_j )$ is a yes-instance of 
\rwss\  using  Lemmas~\ref{lem:graphic-restr} or~\ref{lem:cographic-restr-new}, respectively, depending on the type of $M_\ell$.  If   $(M_\ell,w_\ell,T_\ell\cup\{e_j\},k_{ij}^{(3)},e_i,e_j )$ is a \noinstance\ for every $k_{ij}^{(3)}\leq k$, then we return \no\  and stop. 
Otherwise, solve the problem on the instance $(M',w',T'\cup\{e_i\},k-k_{ij}^{(3)})$.
\item[(iv)] Let $w'(e_h)=k+1$ for $h\in\{1,2,3\}$. For each $i\in\{1,2,3\}$ do the following. 
Set $w_\ell(e_i)=0$ and $w_\ell(e_h)=k+1$ for $h\in\{1,2,3\}$ such that $h\neq i$. Find the minimum  
$k_i^{(4)}\leq k$ such that $(M_\ell,w_\ell,T_\ell,k_i^{(4)})$ is a yes-instance of \wss\ using  Lemmas~\ref{lem:graphic} or~\ref{lem:cographicnew}, respectively, depending on the type of $M_\ell$. 
  If   $(M_\ell,w_\ell,T_\ell,k_i^{(4)})$  is a \noinstance\ for every $k_i^{(4)}\leq k$, 
then we return \no\  and stop. 
Otherwise, solve the problem on the instance $(M',w',T'\cup\{e_i\},k-k_i^{(4)})$.
\item[(v)] Let $w_\ell(e_1)=w_\ell(e_2)=0$ and $w_\ell(e_3)=k+1$.
Set $w'(e_1)=w'(e_2)=w'(e_3)=k+1$. Find the minimum  
$k^{(5)}\leq k$ such that $(M_\ell,w_\ell,T_\ell,k^{(5)})$ is a yes-instance of \wss\ using  Lemmas~\ref{lem:graphic} or~\ref{lem:cographicnew}, respectively, depending on the type of $M_\ell$.
  If   $(M_\ell,w_\ell,T_\ell,k^{(5)})$  is a \noinstance\ for every $k^{(5)}\leq k$ 
then  we return \no\  and stop.
Otherwise, solve the problem on the instance $(M',w',T'\cup\{e_1,e_2\},k-k^{(5)})$.
\item[(vi)] Set $w_\ell(e_1)=w_\ell(e_2)=w_\ell(e_3)=k+1$ and  $w'(e_1)=w'(e_2)=w'(e_3)=k+1$. Find the minimum  
$k^{(6)}\leq k$ such that $(M_\ell,w_\ell,T_\ell,k^{(6)})$ is a yes-instance of \wss\ using  Lemmas~\ref{lem:graphic} or~\ref{lem:cographicnew}, respectively, depending on the type of $M_\ell$. 
  If   $(M_\ell,w_\ell,T_\ell,k^{(6)})$  is a \noinstance\ for every $k^{(6)}\leq k$,
then  we return \no\  and stop.
Otherwise, solve the problem on the instance $(M',w',T',k-k^{(6)})$.
\end{itemize}
\end{branchrule}

Note that the branching of the third type is the only place of our algorithm where we are solving  \rwss.

\begin{lemma}\label{lem:branch-3}
Branching Rule~\ref{brule:3lb} is exhaustive and in each recursive call the parameter strictly reduces. 
Each call of the rule takes $2^{\cO(k)}\cdot ||M||^{\cO(1)}$ time.
\end{lemma}

\begin{proof}
To show correctness, assume first that $(M,w,T,k)$ is a yes-instance of \wss. Let $F\subseteq E(M)\setminus T$ be a set of weight at most $k$ that spans $T$. 
Without loss of generality we assume that $F$ is inclusion minimal and, therefore, $F$ is independent by Observation~\ref{obs:spn}. 
For each $t\in T$, there is a circuit $C_t$ of $M$ such that $t\subseteq C_t\subseteq F\cup\{t\}$. We have the following five cases corresponding to the types of branches.

\medskip
\noindent
{\bf Case~1.} There is $i\in\{1,2,3\}$ such that a)there is $t\in T'$ such that $C_t\cap E(M_\ell)\neq\emptyset$ and $C_t$ goes through $e_i$, and b)for any $t\in T$, there is no circuit $C_t$ that goes through $e_h\in Z$ for $h\neq i$. Let $F_\ell=F\cap E(M_\ell)$ and $F'=(F\cap E(M'))\cup\{e_i\}$. 
We claim that $F_\ell$ spans $T_\ell\cup\{e_i\}$ in $M_\ell$ and $F'$ spans $T'$ in $M'$.

First, we show that  $F_\ell$ spans $T_\ell\cup\{e_i\}$ in $M_\ell$. By a), there is $t\in T'$ such that $C_t\cap E(M_\ell)\neq\emptyset$ and $C_t$ goes through $e_i$. Hence, there are 
cycles $C_t'$ of $M'$ and $C_t''$ of $M_\ell$ respectively such that $C_t=C_t'\bigtriangleup C_t''$ and $C_t'\cap C_t''=\{e_i\}$. 
Because $C_t''\setminus\{e_i\}\subseteq F_\ell$, we
obtain that $F_\ell$ spans $e_i$ in $M_\ell$. Let $t'\in T_\ell$. Since $F$ spans
$t'$ in $M$, there is a circuit $C_{t'}$ of $M$ such that $t'\in C_{t'}\subseteq F\cup\{t'\}$.
 If $C_{t'}\setminus {t'}\subseteq E(M_{\ell})$, then $F_\ell$ spans $t'$, because $C_{t'}\setminus\{t'\}\subseteq F_\ell$. Suppose that $C_{t'}\cap E(M')\neq \emptyset$. 
By b), $C_{t'}$ goes through $e_i$.
Then  there are  cycles $C_{t'}'$ of $M'$ and $C_{t'}''$ of $M_\ell$ such that $\{e_i\}=C_{t'}'\cap C_{t'}''$ and $C_{t'}=C_{t'}'\bigtriangleup C_{t'}''$.  Consider $C=C_{t}''\bigtriangleup C_{t'}''$. By Observation~\ref{obs:symm}, $C$ is a cycle.  As $C\setminus\{e_i\}\subseteq F_\ell$,
$\{e_i\}= C_{t'}''\cap C_t''$
and $t'\notin C_t''$, we obtain that $C$ is a cycle of $M_\ell$ and $t'\in C\subseteq F_\ell\cup\{t'\}$. Therefore, $F_\ell$ spans $t'$.

To prove that $F'$ spans $T'$ in $M'$, consider $t'\in T'$. Since $F$ spans $t'$ in $M$, there is a circuit $C_{t'}$ of $M$ such that $t'\in C_{t'}\subseteq F\cup\{t'\}$. If $C_{t'}\setminus {t'}\subseteq E(M')$, then $F'$ spans $t'$, because $C_{t'}\setminus\{t'\}\subseteq F'$. 
Suppose that $C_{t'}\cap E(M_\ell)\neq \emptyset$. Then by the definition of 3-sum and b), there are  cycles $C_{t'}'$ of $M'$ and $C_{t'}''$ of $M_\ell$ such that $\{e_i\}= C_{t'}'\cap C_{t'}''$ and $C_{t'}=C_{t'}'\bigtriangleup C_{t'}''$. Observe that $C_{t'}'\setminus\{t'\}\subseteq F'$ and, therefore, $F'$ spans $t'$ in $M'$.

Since $F_\ell$ spans $T_\ell\cup \{e_i\}$ in $M_\ell$, $w(F_\ell)\geq k_i^{(1)}$. Because $w(F')+w(F_\ell)=w(F)\leq k$ if the weight of $e_i$ in $M'$ is 0, 
 $w(F')\leq k-k_i^{(1)}$ in this case. Hence, $(M',w',T',k-k_i^{(1)})$ is a yes-instance for a branch of type (i).

\medskip
\noindent
{\bf Case~2.} There are distinct $i,j\in\{1,2,3\}$ such that a)there is $t\in T'$ such that $C_t\cap E(M_\ell)\neq\emptyset$ and $C_t$ goes through $e_i$, b) there is $t\in T'$ such that $C_t\cap E(M_\ell)\neq\emptyset$ and $C_t$ goes through $e_j$.
Let $F_\ell=F\cap E(M_\ell)$ and $F'=(F\cap E(M'))\cup\{e_1,e_2\}$. 
We claim that $F_\ell$ spans $T_\ell\cup\{e_1,e_2\}$ in $M_\ell$ and $F'$ spans $T'$ in $M'$.

We prove first that $F_\ell$ spans $T_\ell\cup\{e_i,e_j\}$ in $M_\ell$. By a), there is $t\in T'$ such that $C_t\cap E(M_\ell)\neq\emptyset$ and $C_t$ goes through $e_i$. Hence, there are 
cycles $C_t'$ of $M'$ and $C_t''$ of $M_\ell$ respectively such that $C_t=C_t'\bigtriangleup C_t''$ and $C_t'\cap C_t''=\{e_i\}$. Because $C_t''\setminus\{e_i\}\subseteq F_\ell$,
obtain that $F_\ell$ spans $e_i$ in $M_\ell$. 
By the same arguments and b), we have that $F_\ell$ spans $e_j$ in $M_\ell$. Let $h\in\{1,2,3\}$ such that $h\neq i,j$. Since $F_\ell$ spans $e_i$ and $e_j$ in $M_\ell$, there are cycles $C^i$ and $C^j$ of $M_\ell$ such that $e_i\in C^i\subseteq F_\ell\cup\{e_i\}$ and $e_j\in C^j\subseteq F_\ell\cup\{e_j\}$. Consider $C=C^i\bigtriangleup C^j\bigtriangleup Z$. By Observation~\ref{obs:symm}, $C$ is a cycle of $M_\ell$. Notice that $e_h\in C\subseteq F_\ell\cup\{e_h\}$. 
Hence, $F_\ell$ spans $e_h$. Because $F_\ell$ spans $Z=\{e_1,e_2,e_3\}$, in particular, $F_\ell$ spans $e_1$ and $e_2$.  
Let $t\in T_\ell$. Since $F$ spans
$t$ in $M$, there is a circuit $C_{t}$ of $M$ such that $t\in C_{t}\subseteq F\cup\{t\}$.
 If $C_{t}\setminus {t}\subseteq E(M_{\ell})$, then $F_\ell$ spans $t$, because $C_{t}\setminus\{t\}\subseteq F_\ell$. Suppose that $C_{t}\cap E(M)\neq \emptyset$. 
We have that $C_t$ goes through $e_h$ for some $h\in\{1,2,3\}$.
Then  there are  cycles $C_{t}'$ of $M'$ and $C_{t}''$ of $M_\ell$ such that $\{e_h\}=C_{t}'\cap C_{t}''$ and $C_{t}=C_{t}'\bigtriangleup C_{t}''$.  Consider $C=C^h\bigtriangleup C_{t'}''$. By Observation~\ref{obs:symm}, $C$ is a cycle of $M_\ell$.  Notice that  $t\in C\subseteq F_\ell\cup\{t\}$ and, therefore, $F_\ell$ spans $t$.

Now we show that $F'$ spans $T'$ in $M'$.  Let $t\in T'$.  Since $F$ spans $t$ in $M$, there is a circuit $C_{t}$ of $M$ such that $t\in C_{t}\subseteq F\cup\{t\}$. If $C_{t}\setminus {t}\subseteq E(M')$, then $F'$ spans $t$, because $C_{t}\setminus\{t\}\subseteq F'$. 
Suppose that $C_{t}\cap E(M_\ell)\neq \emptyset$. 
Then  there are  cycles $C_{t}'$ of $M'$ and $C_{t}''$ of $M_\ell$ such that $\{e_h\}= C_{t}'\cap C_{t}''$ for some $h\in\{1,2,3\}$ and $C_{t}=C_{t'}'\bigtriangleup C_{t}''$. 
If $h=1$ or $h=2$, then $C_{t}'\setminus\{t\}\subseteq F'$ and, therefore, $F'$ spans $t'$ in $M'$. Let $h=3$. Consider $C=C_t'\bigtriangleup Z$. Now $t\in C\subseteq F'\cup \{t\}$. Because $C$ is a cycle of $M'$ by Observation~\ref{obs:symm}, $F'$ spans $t$ in $M'$.

Since $F_\ell$ spans $T_\ell\cup \{e_1,e_2\}$ in $M_\ell$, $w(F_\ell)\geq k^{(2)}$. Because $w(F')+w(F_\ell)=w(F)\leq k$, 
$w(F')\leq k-k^{(2)}$ in this case. Hence, $(M',w',T',k-k^{(2)})$ is a yes-instance for a branch of type (ii). 

\medskip
\noindent
{\bf Case~3.} There are distinct $i,j\in\{1,2,3\}$ such that a)there is $t\in T_\ell$ such that $C_t\cap E(M')\neq\emptyset$ and $C_t$ goes through $e_i$, b) there is $t'\in T'$ such that $C_{t'}\cap E(M_\ell)\neq\emptyset$ and $C_{t'}$ goes through $e_j$, 
and c) for any $t''\in T$, there is no circuit $C_{t''}$ that goes through $e_h\in Z$ for $h\neq i,j$. 
Let $F_\ell=(F\cap E(M_\ell))\cup \{e_i\}$ and $F'=(F\cap E(M'))\cup\{e_j\}$. 
We claim that $F_\ell$ spans $T_\ell\cup\{e_j\}$ and $F_\ell\setminus\{e_i\}$ spans $e_j$ in $M_\ell$ 
and $F'$ spans $T'\cup\{e_i\}$ in $M'$.

We prove that $F_\ell$ spans $T_\ell\cup\{e_j\}$. By b),  there is $t'\in T'$ such that $C_{t'}\cap E(M_\ell)\neq\emptyset$ and $C_{t'}$ goes through $e_j$. Then there are cycles $C_{t'}'$ and $C_{t'}''$ of $M'$ and $M_\ell$ respectively such that $C_{t'}=C_{t'}'\bigtriangleup C_{t'}''$ and $C_{t'}'\cap C_{t'}''=\{e_j\}$. Because $e_j\in C_{t'}''\subseteq F_\ell\cup\{e_j\}$ and $e_i\notin C_{t'}''$, $F_\ell\setminus\{e_i\}$ spans $e_j$ in $M_\ell$. Let $t''\in T_\ell$.
There is a circuit $C_{t''}$ of $M$ such that $t''\in C_{t''}\subseteq F\cup\{t''\}$. If $C_{t''}\setminus \{t''\}\subseteq E(M_\ell)$, then $C_{t''}\setminus \{t''\}\subseteq F_\ell$ and $F_\ell$ spans $t''$ in $M_\ell$. Assume that $C_{t''}\cap E(M')\neq \emptyset$. Then there are 
 cycles $C_{t''}'$ and $C_{t''}''$ of $M'$ and $M_\ell$ respectively such that $C_{t''}=C_{t''}'\bigtriangleup C_{t''}''$ and $C_{t''}'\cap C_{t''}''=\{e_h\}$ for some $h\in\{1,2,3\}$. By c), either $h=i$ of $h=j$. If $h=i$, then $e_h\in F_\ell$ and, therefore, $C_{t''}''\setminus\{t'\}\subseteq F_\ell$.   Hence, $F_\ell$ spans $t''$ in this case. Assume that $h=j$ and consider $C=C_{t''}''\bigtriangleup C_{t'}''$. Notice that $C$ is a cycle of $M_\ell$ by Observation~\ref{obs:symm} and $t''\in C\subseteq F_\ell\cup\{t''\}$. Hence, $F_\ell$ spans $t''$.

The proof of the claim that $F'$ spans $T'\cup\{e_i\}$ in $M'$ is done by the same arguments using symmetry.

Since $F_\ell$ spans $T_\ell\cup \{e_j\}$ in $M_\ell$, $w(F_\ell)\geq k_{ij}^{(3)}$. Because $w(F')+w(F_\ell)=w(F)\leq k$, 
$w(F')\leq k-k_{ij}^{(3)}$ in this case. Hence, $(M',w',T'\cup\{e_i\},k-k_{ij}^{(3)})$ is a yes-instance for a branch of type (iii).

\medskip
\noindent
{\bf Case~4.} There is $i\in\{1,2,3\}$ such that a)there is $t\in T_\ell$ such that $C_t\cap E(M')\neq\emptyset$ and $C_t$ goes through $e_i$, and b)for any $t\in T$, there is no circuit $C_t$ that goes through $e_h\in Z$ for $h\neq i$. Notice that this case is symmetric to Case~1. Using the same arguments, we prove that $(M',w',T'\cup\{e_i\},k-k_i^{(4)})$ is a yes-instance for a branch of type (iv).

\medskip
\noindent
{\bf Case~5.} There are distinct $i,j\in\{1,2,3\}$ such that a)there is $t\in T_\ell$ such that $C_t\cap E(M')\neq\emptyset$ and $C_t$ goes through $e_i$, b) there is $t\in T'$ such that $C_t\cap E(M')\neq\emptyset$ and $C_t$ goes through $e_j$.
This case is symmetric to Case~2. Using the same arguments, we obtain that $(M',w',T'\cup\{e_1,e_2\},k-k^{(5)})$ is a yes-instance for a branch of type (v). 

\medskip
If the conditions of Cases~1--5 are not fulfilled, we get the last case.

\medskip
\noindent
{\bf Case~6.} For any $t\in T$, either $C_t\subseteq E(M_\ell)$ or $C_t\subseteq E(M')$. Let $F_\ell=F\cap E(M_\ell)$ and $ F'=F\cap E(M')$. We have that $F_\ell$ spans $T_\ell$ and $F'$ spans $T'$. Notice that $w(F_\ell)\geq k^{(6)}$. Because $w(F')+w(F_\ell)=w(F)\leq k$, we have that $(M',w',T',k-k^{(6)})$ is a yes-instance for a branch of type (vi).

\medskip
Assume now that for one of the branches, we get a yes-answer. We show that the original instance $(M,w,T,k)$ is a yes-instance. To do it, we consider 6 cases corresponding to the types of branches. We use essentially the same arguments in all the cases: we take a solution $F'$ for the instance obtained in the corresponding branch and combine it with a solution $F_\ell$ of the instance for $M_\ell$ to obtains a solution for the original instance.

\medskip
\noindent
{\bf Case~1.}  $(M',w',T',k-k_i^{(1)})$ is a yes-instance of a branch of type (i). Let $F_\ell\subseteq E(M_\ell)\setminus (T_\ell\cup\{e_i\})$ with $w_\ell(F_\ell)\leq k_i^{(1)}$ be a set that spans $T_\ell\cup\{e_i\}$ in $M_\ell$. Clearly, $k_i^{(1)}\leq k$.  Consider $F'\subseteq E(M')\setminus T'$ with $w'(F')\leq k-k_i^{(1)}$ that spans $T'$ in $M'$. Let $F=(F'\setminus \{e_i\})\cup F_\ell$. 
Notice that $Z\cap F_\ell=\emptyset$, because $w_\ell(e_h)=k+1$ for $h\in\{1,2,3\}$. Similarly, $e_h\notin F'$ for $h\in\{1,2,3\}$ such that $h\neq i$, because $w'(e_h)=k+1$. Hence, $F\subseteq E(M)\setminus T$. It is easy to see that $w(F)\leq k$. We show that $F$ spans $T$ in $M$.

Let $t\in T$. Suppose first that $t\in T_\ell$. There is a circuit $C_t$ of $M_\ell$ such that $t\in C_t\subseteq F_\ell\cup\{t\}$. It is sufficient to notice that $C_t$ is a cycle of $M$ and, therefore, $F$ spans $t$ in $M$. Let $t\in T'$. There is a circuit $C_t$ of $M'$ such that $t\in C_t\subseteq F'\cup\{t\}$. If $C_t\setminus\{t\}\subseteq F$, i.e., 
$e_i\notin C_t$, then $F'$ spans $t$. Suppose that $e_i\in C_t$. Recall that $F_\ell$ spans $e_i$ in $M_\ell$. Hence,  there is a cycle $C^{(i)}$ of $M_\ell$ such that $e_i\in C^{(i)}\subseteq F_\ell\cup\{e_i\}$. Let $C_t'=C_t\bigtriangleup C^{(i)}$. By the definition of 3-sums, $C_t'$ is a cycle of $M$. We have that $t\in C_t'\subseteq F\cup\{t\}$ and, therefore, $F$ spans $t$.

\medskip
\noindent
{\bf Case~2.} $(M',w',T',k-k{(2)})$ is a yes-instance of a branch of type (ii). Let $F_\ell\subseteq E(M_\ell)\setminus (T_\ell\cup\{e_1,e_2\})$ with $w_\ell(F_\ell)\leq k_i^{(1)}$ be a set that spans $T_\ell\cup\{e_1,e_2\}$ in $M_\ell$. Clearly, $k^{(2)}\leq k$.  Consider $F'\subseteq E(M')\setminus T'$ with $w'(F')\leq k=k^{(2)}$ that spans $T'$ in $M'$. Let $F=(F'\setminus \{e_1,e_2\})\cup F_\ell$. 
Notice that $Z\cap F_\ell=\emptyset$, because $w_\ell(e_h)=k+1$ for $h\in\{1,2,3\}$. Similarly, $e_3\notin F'$, because $w'(e_3)=k+1$. Hence, $F\subseteq E(M)\setminus T$. It is easy to see that $w(F)\leq k$. We show that $F$ spans $T$ in $M$.

Let $t\in T$. Suppose first that $t\in T_\ell$. There is a circuit $C_t$ of $M_\ell$ such that $t\in C_t\subseteq F_\ell\cup\{t\}$. It is sufficient to notice that $C_t$ is a cycle of $M$ and, therefore, $F$ spans $t$ in $M$. Let $t\in T'$. There is a circuit $C_t$ of $M'$ such that $t\in C_t\subseteq F'\cup\{t\}$. If $C_t\setminus\{t\}\subseteq F$, i.e., 
$e_1,e_2\notin C_t$, then $F'$ spans $t$. Suppose that $e_1\in C_t$ and $e_2\notin C_t$. 
Recall that $F_\ell$ spans $e_1$ in $M_\ell$. Hence, there is a cycle $C^{(1)}$ of $M_\ell$ such that $e_1\in C^{(1)}\subseteq F_\ell\cup\{e_1\}$. Let $C_t'=C_t\bigtriangleup C^{(1)}$. By the definition of 3-sums, $C_t'$ is a cycle of $M$. We have that $t\in C_t'\subseteq F\cup\{t\}$ and, therefore, $F$ spans $t$.
If  $e_1\notin C_t$ and $e_2\in C_t$, then we observe that $F_\ell$ spans $e_2$ in $M_\ell$ and there is a cycle $C^{(2)}$ of $M_\ell$ such that $e_2\in C^{(2)}\subseteq F_\ell\cup\{e_1\}$. 
Then we conclude that $F$ spans $t$ using the same arguments as before using symmetry. Suppose that $e_1,e_2\in C_t$. Consider  the cycle 
$C_t'=C_t\bigtriangleup C^{(1)}\bigtriangleup C^{(2)}$ of $M$. We have that $t\in C_t'\subseteq F\cup\{t\}$ and, therefore, $F$ spans $t$.

\medskip
\noindent
{\bf Case~3.} $(M',w',T'\cup\{e_i\},k-k_{ij}{(3)})$ is a yes-instance of a branch of type (iii). Let $F_\ell\subseteq E(M_\ell)\setminus (T_\ell\cup\{e_j\})$ with $w_\ell(F_\ell)\leq k_{ij}^{(3)}$ be a set that spans $T_\ell\cup\{e_j\}$ in $M_\ell$ such that $F\setminus \{e_i\}$ spans $e_j$. 
Clearly, $k_{ij}^{(3)}\leq k$.  Consider $F'\subseteq E(M')\setminus (T'\cup\{e_i\})$ with $w'(F')\leq k-k_{ij}^{(3)}$ that spans $T'\cup\{e-i\}$ in $M'$. Let $F=(F'\setminus \{e_j\})\cup (F_\ell\setminus\{e_i\})$. 
Notice that $e_h\notin F_\ell=\emptyset$ for $h\in\{1,2,3\}$ such that $h\neq i$, because $w_\ell(e_h)=k+1$, and  $e_h\notin F'=\emptyset$ for $h\in\{1,2,3\}$ such that $h\neq j$, because $w'(e_h)=k+1$.  Hence, $F\subseteq E(M)\setminus T$. It is straightforward that $w(F)\leq k$. We show that $F$ spans $T$ in $M$.

Let $t\in T$. Suppose first that $t\in T_\ell$. There is a circuit $C_t$ of $M_\ell$ such that $t\in C_t\subseteq F_\ell\cup\{t\}$. If $e_i\notin F_\ell$, then $C_t\setminus\{t\}\subseteq F$ and, therefore, $F$ spans $t$ in $M$. Suppose that $e_i\in C_t$. Because $F'$ spans $e_i$ in $M'$,  there is a cycle $C^{(i)}$ of $M'$ such that
$e_i\in C^{(i)}\subseteq F'\cup\{e_i\}$. Suppose that $e_j\notin C^{(i)}$. 
Let $C_t'=C_t\bigtriangleup C^{(i)}$. We have that $C_t'$ is a cycle of $M$ and $t\in C_t'\subseteq F\cup\{t\}$. Hence, $F$ spans $t$.
Suppose now that $e_j\in C^{(i)}$. Since $F_\ell\setminus\{e_i\}$ spans $e_j$, there is a cycle $C^{(j)}$ of $M_\ell$ such that $e_j\subseteq C^{(j)}\subseteq (F_\ell\setminus\{e_i\})\cup\{e_j\}$.
Let $C_t'=C_t\bigtriangleup C^{(i)}\bigtriangleup C^{(j)}$. We obtain that $C_t'$ is a cycle of $M$ and $t\in C_t'\subseteq F\cup\{t\}$. Hence, $F$ spans $t$.
The proof for the case $t\in T'$ uses the same arguments using symmetry. 

\medskip
\noindent
{\bf Case~4.} $(M',w',T'\cup\{e_i\},k-k_{i}{(4)})$ is a yes-instance of a branch of type (iv). 
This case is symmetric to Case~1 and is analyzed in the same way. We consider a set 
$F_\ell\subseteq E(M_\ell)\setminus T_\ell$ with $w_\ell(F_\ell)\leq k_i^{(4)}$ that spans $T_\ell$ in $M_\ell$ and $F'\subseteq E(M')\setminus T'$ with $w'(F')\leq k-k_i^{(4)}$ that spans $T'\cup\{e_i\}$ in $M'$. Let $F=F'\cup (F_\ell\setminus\{e_i\})$. We have that $F\subseteq E(M)\setminus T$ has weight at most $k$ and spans $T$ in $M$.

\medskip
\noindent
{\bf Case~5.} $(M',w',T'\cup\{e_1,e_2\},k-k^{(5)})$  is a yes-instance of a branch of type (v). 
This case is symmetric to Case~2 and is analyzed in the same way. We consider a set 
$F_\ell\subseteq E(M_\ell)\setminus T_\ell$ with $w_\ell(F_\ell)\leq k^{(5)}$ that spans $T_\ell$ in $M_\ell$ and $F'\subseteq E(M')\setminus T'$ with $w'(F')\leq k-k^{(5)}$ that spans $T'\cup\{e_1,e_2\}$ in $M'$. Let $F=F'\cup (F_\ell\setminus\{e_1,e_2\})$. We have that $F\subseteq E(M)\setminus T$ has weight at most $k$ and spans $T$ in $M$.

\medskip
It remains to consider the last case.

\medskip
\noindent
{\bf Case~6.}  $(M',w',T',k-k^{(6)})$   is a yes-instance of a branch of type (v).  Let $F_\ell\subseteq E(M_\ell)\setminus T_\ell$ with $w_\ell(F_\ell)\leq k_{(6)}$ be a set that spans $T_\ell$ in $M_\ell$ and let $F'\subseteq E(M')\setminus T'$ be a set  with $w'(F')\leq k-k^{(6)}$ that spans $T'$ in $M'$. 
Notice that for $i\in\{1,2,3\}$, $e_i\notin F_\ell$ and $e_i\notin F'$, because
$w_\ell(e_i)=w'(e_i)=k+1$. Consider $F=F_F'\cup F_\ell$. Clearly, $w(F)\leq k$. We show that $F$ spans $T$ in $M$.

Let $t\in T$. If $t\in T_\ell$, then  there is a circuit $C_t$ of $M_\ell$ such that $t\in C_t\subseteq F_\ell\cup\{t\}$. Since $C_t\subseteq E(M_\ell)$, we have that $F_\ell$ spans $t$ in $M$. If $t\in T'$, then by the  same arguments, $F'$ spans $t$ not only in $M'$ but also in $M$.

Since we always have that $k_i^{(1)},k^{(2)},k_{ij}^{(3)},k_i^{(4)},k^{(5)},k^{(6)}\geq 1$, the recursive calls are done for the parameters that are strictly less than $k$. This completes the proof. 

The claim that each call of the rule (without recursive steps) takes $2^{\cO(k)}\cdot ||M||^{\cO(1)}$ time follows from Lemmas \ref{lem:graphic-restr}, \ref{lem:cographic-restr-new}
and \ref{lem:simpleAlgo}.
\end{proof}

\begin{quote}
\noindent 
{
From now  onwards we assume that there is no child of $M_s$ with terminals.
Recall that $M_s$ is either a graphic or cographic matroid. The subsequent steps depend on the type of $M_s$ and  are considered in separate sections.}
\end{quote}

\subsection{The case of a graphic sub-leaf}
Throughout this section we assume that $M_s$ is a graphic matroid. Let $G$ be a graph such that its cycle matroid $M(G)$ is isomorphic to $M_s$. We assume that $M(G)=M_s$. Recall that the circuits of $M(G)$ are exactly the cycles of $G$. We reduce leaves in this case by the following reduction rule. 
In this reduction rule  we first solve a few instances of \wss\ and later use the solutions to these instances to reduce the graph and re-define the weight function.

\begin{reduction}[{\bf Graphic $3$-leaf reduction rule}]
\label{rule:graphic3leafrule}
For a child $M_{\ell}$ of $M_s$ with $T\cap E(M_\ell)=\emptyset$, do the following. 
Let $Z=\{e_1,e_2,e_3\}=E(M_s)\cap E(M_\ell)$.  
Set $w_\ell(e)=w(e)$ for $e\in E(M_\ell)\setminus Z$, $w_\ell(e_1)=w_\ell(e_2)=w_\ell(e_3)=k+1$. 
\begin{itemize}
\item[(i)]  For each $i\in \{1,2,3\}$, find the minimum $k_i\leq k$ such that   $(M_\ell,w_\ell,\{e_i\},k_i)$ is a 
 yes-instance of \wss\ using  Lemmas~\ref{lem:graphic} or~\ref{lem:cographicnew}, respectively, depending on the type of $M_\ell$.   If    $(M_\ell,w_\ell,\{e_i\},k_i)$ is a \noinstance\ for every $k_i\leq k$,  
then  we set $k_i=k+1$.

\item[(ii)] Find the minimum $k'\leq k$ such that $(M_\ell,w_\ell,\{e_1,e_2\},k')$ is a yes-instance of \wss\ using  Lemmas~\ref{lem:graphic} or~\ref{lem:cographicnew}, respectively, depending on the type of $M_\ell$. 
If   $(M_\ell,w_\ell,\{e_1,e_2\},k')$ is a \noinstance\ for every $k'\leq k$,  
then  we set $k'=k+1$. 
If $k'\leq k$, then we find an inclusion minimal set $F_{\ell}\subseteq E(M_\ell)\setminus Z$ of weight $k'$ that spans $e_1$ and $e_2$. Observe that Lemmas~\ref{lem:graphic} or~\ref{lem:cographicnew} are only for decision version. However, we can apply standard self reducibility tricks to make them output a solution also. There are circuits $C_1$ and $C_2$ of $M_\ell$ such that 
$e_1\in C_1\subseteq F_{\ell}\cup\{e_1\}$, $e_2\in C_2\subseteq F_\ell \cup\{e_2\}$ and $F_{\ell}=(C_1\setminus\{e_1\})\cup (C_2\setminus\{e_2\})$. Notice that $C_1$ and $C_2$ can be found by finding inclusion minimal subsets of $F_{\ell}$ that span $e_1$ and $e_2$, respectively. 
\end{itemize}
Recall that $Z$ induces a cycle of $G$. Denote by $v_1,v_2,$ and $v_3$ the vertices of the cycle.
Furthermore, let $v_1,v_2,$ and $v_3$ be incident to $e_3,e_1$, $e_1,e_2$ and $e_2,e_3$, respectively. 
We construct the graph $G'$ by adding a new vertex $u$ and making it adjacent to $v_1$, $v_2$ and $v_3$. Notice that because the circuits of $M(G)$ are cycles of $G$, any circuit of $M(G)$ is also  a circuit of $M(G')$. Let $M'$ the matroid defined by the conflict tree $\mathcal{T}'=\mathcal{T}-M_\ell$ and where $M_s$ is replaced by $M(G')$. The weight function $w'\colon E(M')\rightarrow\mathbb{N}$ is defined by setting 
$w'(e)=w(e)$ for $e\in E(M')\setminus (Z\cup\{v_1u,v_2u,v_3u\})$, 
$w'(e_1)=k_1$, $w'(e_2)=k_2$, and $w'(e_3)=k_3$. If if $k'\leq k$ then we set 
 $w'(v_1u)=w(C_1\setminus (C_2\cup\{e_1\}))$, $w'(v_3u)=w(C_1\setminus (C_2\cup\{e_2\}))$ and $w'(v_1u)=w(C_1\cap C_2)$; else we set $w'(v_1u)=w'(v_2u)=w'(v_3u)=k+1$.  The reduced instance is denoted by  $(M',w',T,k)$.  
\end{reduction}
\medskip
The construction of $G'$ and Observation~\ref{obs:triangle} immediately imply the following observation.

\begin{observation}\label{obs:ineq}
For any distinct $i,j\in \{1,2,3\}$, $$w'(e_i)+w'(e_j)=k_i+k_j\geq k'=w'(v_1u)+w'(v_2u)+w'(v_3u)$$ and  if $k'\leq k$ then $w'(v_iu)+w'(v_ju)\geq w'(v_iv_j)$. Also, if $w'(e_i)+w'(e_j)\leq k$ for some distinct $i,j\in\{1,2,3\}$, then $k'\leq k$. 
\end{observation}

We use Observation~\ref{obs:ineq} to prove that the rule is safe.

\begin{lemma}\label{lem:red-3}
Reduction Rule~\ref{rule:graphic3leafrule} is safe and can be applied 
in $2^{\cO(k)}\cdot ||M||^{\cO(1)}$ time. 
\end{lemma}

\begin{proof}
Denote by $M''$ the matroid defined by $\mathcal{T}'=\mathcal{T}-M_\ell$. To prove that the rule is safe,  first assume that $(M,w,T,k)$ is a yes-instance. Then there is an inclusion minimal 
set $F\subseteq E(M)\setminus T$ of weight at most $k$ that spans $T$. If $F\cap E(M_\ell)=\emptyset$, then $F$ spans $T$ in $M'$ as well and  $(M',w',T,k)$ is a yes-instance. Suppose from now that $F\cap E(M_\ell)\neq\emptyset$.

For each $t\in T$, there is a circuit $C_t$ of $M$ such that $t\in C\subseteq F\cup\{t\}$. If $C_t\cap E(M_\ell)\neq\emptyset$, $C_t=C_t'\bigtriangleup C_t''$, where $C_t'$ is a cycle of $M''$ and  $C_t''$ is a cycle of $M_\ell$. By Observation~\ref{obs:through-sum}, we can assume that $C_t'\cap C_t''$ contains the unique element $e_i$, i.e., $C_t$ goes through $e_i$. 
To simplify notations, it is assumed that $v_4=v_1$.
We consider the following three cases.

\medskip
\noindent
{\bf Case 1.} There is a unique $e_i\in Z$ such that for any $t\in T$, either $C_t\subseteq E(M'')$ or $C_t$ goes through $e_i$. Let $F'=(F\cap E(M''))\cup\{e_i\}$. 

We show that $F'$ spans $T$ in $M'$. Let $t\in T$. If $C_t\subseteq E(M'')$, then $t\in C_t\subseteq (F\cap E(M''))\cup\{t\}$ and, therefore, $F'$ spans $t$ in $M'$. Suppose that $C_t\cap E(M_\ell)\neq \emptyset$. Then $C_t=C_t'\bigtriangleup C_t''$, where  $C_t'$ is a cycle of $M''$, $C_t''$ is a cycle of $M_\ell$ and $C_t'\cap C_t''=\{e_i\}$. We have that $t\in C_t'\cup\{t\}$ and $C_t'\setminus\{t\}\subseteq F'$ spans $t$. 

Because $F\cap E(M_\ell)\neq\emptyset$ and $F$ is inclusion minimal spanning set, there is $t\in T$ such that $C_t$ goes through $e_i$. Let  $C_t=C_t'\bigtriangleup C_t''$, where  $C_t'$ is a cycle of $M''$, $C_t''$ is a cycle of $M_\ell$ and $C_t'\cap C_t''=\{e_i\}$. Notice that $C_t''\setminus\{e_i\}$ spans $e_i$ in $M_\ell$. Hence, $w_\ell(C_t''\setminus \{e_i\})\geq k_i$. 
Because $w'(e_i)=k_i$, we conclude that $w'(F')\leq w(F)$.

Since $F'\subseteq E(M')\setminus T$ spans $T$ and has the weight at most $k$ in $M'$, $(M',w',T,k)$ is a yes-instance.

\medskip
\noindent
{\bf Case 2.} There are two distinct  $e_i,e_j\in Z$ such that for any $t\in T$, either $C_t\subseteq E(M'')$, or $C_t$ goes through $e_i$, or $C_t$ goes through $e_j$, and at least one $C_t$ goes through $e_i$ and at least one $C_t$ goes through $e_j$. 
Let $F'=(F\cap E(M''))\cup\{v_1u,v_2u,v_3u\}$. 

We claim that $F'$ spans $T$ in $M'$. Let $t\in T$. If $C_t\subseteq E(M'')$, then $t\in C_t\subseteq (F\cap E(M''))\cup\{t\}$ and, therefore, $F'$ spans $t$ in $M'$. Suppose that $C_t\cap E(M_\ell)\neq \emptyset$. Then $C_t=C_t'\bigtriangleup C_t''$, where  $C_t'$ is a cycle of $M''$, $C_t''$ is a cycle of $M_\ell$ and either $C_t'\cap C_t''=\{e_i\}$ or  $C_t'\cap C_t''=\{e_j\}$. By symmetry, let $C_t'\cap C_t''=\{e_i\}$. Because $e_i,v_iu,v_{i+1}u$ 
induce a cycle of the graph $G'$, $\{e_i,v_iu,v_{i+1}u\}$ is a circuit of $M'$ and 
 $C_t'''=C_t'\bigtriangleup \{e_i,v_iu,v_{i+1}u\}$ is a cycle of $M'$. 
We have that $t\in C_t'''\cup\{t\}$ and $C_t'''\setminus\{t\}\subseteq F'$ spans $t$. 

Because $F\cap E(M_\ell)\neq\emptyset$, there is $t\in T$ such that $C_t$ goes through $e_i$ and there is $t'\in T$ such that $C_{t'}$ goes through $e_j$. 
Let  $C_t=C_t'\bigtriangleup C_t''$ and $C_{t'}=C_{t'}'\bigtriangleup C_{t''}''$, where  $C_t',C_{t'}'$ are cycles of $M''$, $C_t'',C_{t'}''$ are cycles of  $M_\ell$ and $C_t'\cap C_t''=\{e_i\}$, $C_{t'}'\cap C_{t'}''=\{e_j\}$. Notice that $C_t''\setminus\{e_i\}$ spans $e_i$ in $M_\ell$ and $C_{t'}''\setminus \{e_j\}$ spans $e_j$. Hence, $w_\ell((C_t''\setminus \{e_i\})\cup (C_{t'}''\setminus\{e_j\}))\geq w_\ell(F_\ell)=k'$ by Observation~\ref{obs:triangle}. Because $w'(\{v_1u,v_2u,v_3u\})=k'$,
$w'(F')\leq w(F)$.

Since $F'\subseteq E(M')\setminus T$ spans $T$ and has the weight at most $k$ in $M'$, $(M',w',T,k)$ is a yes-instance.

\medskip
\noindent
{\bf Case 3.} For each $i\in \{1,2,3\}$, there is $t\in T$ such that $C_t$ goes through $e_i$.  As in Case~1, we set  $F'=(F\cap E(M''))\cup\{v_1u,v_2u,v_3u\}$ and use the same arguments to show that $F'\subseteq E(M')\setminus T$ spans $T$ and has the weight at most $k$ in $M'$.

\medskip
Assume now that the reduced instance $(M',w',T,k)$ is a yes-instance. Let $F'\subseteq E(M')\setminus T$ be an  inclusion minimal set of  weight at most $k$ that spans $T$ in $M'$. 
Let $S=\{e_1,e_2,e_3,v_1u,v_2u,v_3u\}$.
If $F'\cap S=\emptyset$, then $F'\subseteq E(M)$ and, therefore, $F'$ spans $T$ in $M$ as well. Assume from now that  
$F'\cap S\neq\emptyset$. By Observation~\ref{obs:spn} and because $\{v_1,v_2,v_3\}$ separates $u$ from $V(G)\setminus \{v_1,v_2,v_3\}$ in $G'$, the edges of $F'\cap S$ induce a tree in $G'$. 
Moreover, $u$ is incident to either 2 or 3 edges of this tree. 
We consider the following cases depending on the structure of the tree.

\medskip
{\bf Case 1.} One one the following holds: i) $v_1u,v_2u,v_3u\in F'$ or ii) $|\{v_1u,v_2u,v_3u\}\cap F'|=2$ and $\{e_1,e_2,e_3\}\cap F'\neq \emptyset$ or iii)$|\{e_1,e_2,e_3\}\cap F'|\geq 2$. We define $F=(F'\setminus S)\cup F_\ell$. Clearly, $F\subseteq E(M)\setminus T$. Notice also that $w'(F\cap S)\geq k'$ by Observation~\ref{obs:ineq} and, therefore, 
$w(F)\leq k$. To show that $(M,w,T,k)$ is a yes-instance, we prove that  $F$ spans $T$ in $M$. 

Let $t\in T$. Since $F'$ spans $t$ in $M'$, there is a circuit $C_t$ of $M'$ such that $t\in C_t\subseteq F'\cup\{t\}$. If $C_t\cap S=\emptyset$, then 
$C_t\setminus\{t\}$ spans $t$ in $M$. Suppose that $C_t\cap S\neq\emptyset$. As $S$ induces a complete graph on 4 vertices in $G'$ and $\{v_1,v_2,v_3\}$ separate $u$ from $V(G)\setminus \{v_1,v_2,v_3\}$, we conclude that there is $i\in\{1,2,3\}$ such that $C_t'=(C_t\setminus S)\cup\{e_i\}$ is a cycle of $M'$. Notice that $C_t'$ is also a cycle of $M''$. By the definition  of $F_\ell$ and Observation~\ref{obs:triangle}, there is a cycle $C_t''$ of $M_\ell$ such that $e_i\in C_t''\subseteq F_\ell\cup\{e_i\}$. Consider the cycle $C_t'''=C_t'\bigtriangleup C_t''$ of $M$. We have that $t\in C_t'''\subseteq F$ and, therefore, $F$ spans $t$. 

\medskip
If the conditions i)--iii) of Case~1 are not fulfilled, then $F'\cap S=\{e_i\}$ for some $i\in\{1,2,3\}$. 

\medskip
{\bf Case 2.} $F'\cap S=\{e_i\}$ for some $i\in\{1,2,3\}$. By the definition of $w'(e_i)=k_i$, there is a circuit $C$ of $M_\ell$ such that $e_i\in C\subseteq (E(M_\ell)\setminus Z)\cup \{e_i\}$
and $w_\ell(C\setminus\{e_i\})=k_i$. Let $F=F'\bigtriangleup C$. Clearly, $w(F)\leq k$. We show that $F$ spans $T$.

Let $t\in T$. Since $F'$ spans $t$ in $M'$, there is a circuit $C_t$ of $M'$ such that $t\in C_t\subseteq F'\cup\{t\}$. If $C_t\cap S=\emptyset$, then 
$C_t$ spans $t$ in $M$. Suppose that $C_t\cap S\neq\emptyset$, i.e., $C_t\cap S=\{e_i\}$.  Notice that $C_t$ is also a cycle of $M''$. 
Consider the cycle $C_t'=C_t\bigtriangleup C$. Since $t\in C_t'\subseteq F\cup\{t\}$, $F$ spans $t$. 

\medskip

From the description of Reduction Rule~\ref{rule:graphic3leafrule} and 
Lemma~\ref{lem:simpleAlgo}, it can be deduced that Reduction Rule~\ref{rule:graphic3leafrule} can be applied in  time $2^{\cO(k)}\cdot ||M||^{\cO(1)}$.
\end{proof}

\subsection{The case of a cographic sub-leaf}
Now we have reached the final step of our algorithm.  Throughout this section we assume that $M_s$ is a cographic matroid. Let $G$ be a graph such that the bond matroid of $G$ is isomorphic to $M_s$. The algorithm that constructs a good $\{1,2,3\}$-decomposition could be also be used to output the graph $G$
Without loss of generality, we can assume that $G$ is connected. Also, recall  that the circuits of the bond matroid $M^*(G)$ are exactly minimal cut-sets of $G$.

The isomorphism between $M_s$ and $M^*(G)$ is not necessarily unique. We could choose any isomorphism between $M_s$ and $M^*(G)$ that is beneficial for our algorithmic purposes. Indeed, in what follows we fix an isomorphism that is useful in designing our algorithm. 
Let $M_\ell^{(1)},\ldots,M_\ell^{(s)}$ denote those leaves of the conflict tree $\mathcal{T}$ that are also the children of $M_s$. Let $Z_i=E(M_s)\cap E(M_\ell^{(i)})$,  $i\in\{1,\ldots,s\}$.
If $M_s$ has a parent $M^*$ in $\mathcal{T}$ and $E(M_s)\cap E(M^*)\neq\emptyset$, then let $Z^*$  denote $Z^*=E(M_s)\cap E(M^*)$; we {\em emphasize}  that $Z^*$ may  not exist. Next we define the notion of {\em clean cut}. 
\begin{definition}
We say that $\alpha(Z_i)\subseteq E(G)$ is a \emph{clean cut} with respect to an isomorphism $\alpha\colon M_s\rightarrow M^*(G)$, if there is a component $H$ of $G-\alpha(Z_i)$ such that 
\begin{itemize}
\setlength{\itemsep}{-2pt}
\item[(i)] $H$ has no bridge, 
\item[(ii)] $E(H)\cap \alpha(Z_j)=\emptyset$ for $j\in\{1,\ldots,s\}$, and 
\item[(iii)] $E(H)\cap\alpha(Z^*)=\emptyset$ if $Z^*$ exists. 
\end{itemize}
We call $H$ a \emph{clean component} of $G-\alpha(Z_i)$.
\end{definition}
Next we show that given any isomorphism between  $M_s$ and $M^*(G)$, we can obtain another  isomorphism between  $M_s$ and $M^*(G)$ with respect to which we have at least one clean component. 

\begin{lemma}\label{lem:clean}
There is an isomorphism $\alpha\colon M_s\rightarrow M^*(G)$ and a 
 child $M_\ell^{(i)}$ of $M_s$ such that $\alpha(Z_i)$ is a clean cut with respect to $\alpha$. 
 Moreover,  given any arbitrary isomorphism from  $M_s$ to $M^*(G)$, one can obtain  
 such an isomorphism and a clean cut together with a clean component  in polynomial time.
\end{lemma}

\begin{proof}
We prove the lemma  first assuming that $Z^*$ exists. Let $\alpha \colon M_s\rightarrow M^*(G)$ be an isomorphism. Clearly $\alpha$ maps $E(M_s)$ to the edges of $G$.    
Suppose that there is $p\in\{1,\ldots,s\}$ such that there is a component $H$ of $G-\alpha(Z_p)$ with $E(H)\cap \alpha(Z^*)=\emptyset$.
Then we set 
$\alpha_0=\alpha$, $H^{(0)}=H$ and $i_0=p$. Otherwise, let $p\in\{1,\ldots,s\}$. Denote by $H_1$ and $H_2$ the components of $G-\alpha(Z_p)$. Because $|Z^*|\leq 3$,
 $E(H_1)\cap\alpha(Z^*)\neq \emptyset$ and $E(H_2)\cap\alpha(Z^*)\neq \emptyset$, there is $H_j$ for $j\in\{1,2\}$ such that $|E(H_j)\cap\alpha(Z^*)|=1$. Let $\{e\}=E(H_j)\cap\alpha(Z^*)$.
Since $\alpha(Z^*)$ is a cut-set, $e$ is a bridge of $H_j$. By the minimality of $\alpha(Z^*)$, every component of $H-e$ contains an end vertex of an edge of $\alpha(Z_p)$. Since $|\alpha(Z_p)|=3$, we obtain that there is $e'\in \alpha(Z_p)$ such that $\{e,e'\}$ is a minimal cut-set of $G$. Let
$\alpha'(x)=\alpha(x)$ for $x\in E(M_s)\setminus\{\alpha^{-1}(e),\alpha^{-1}(e')\}$, 
$\alpha'(\alpha^{-1}(e))=e'$ and $\alpha'(\alpha^{-1}(e'))=e$.
By Observation~\ref{obs:par}, $\alpha'$ is an isomorphism of $M_s$ to $M^*(G)$.
Notice that now we have a component $H$ of $G-\alpha'(Z_p)$ with $E(H)\cap \alpha'(Z^*)=\emptyset$. Respectively, 
we set 
$\alpha_0=\alpha'$, $H^{(0)}=H$ and $i_0=p$.

Assume inductively that we have a sequence $(\alpha_0,i_0,H^{(0)}),\ldots,(\alpha_q,i_q,H^{(q)})$, where $\alpha_0,\ldots,\alpha_q$ are isomorphisms of $M_s$ to $M^*(G)$, $i_0,\ldots,i_q\in\{1,\ldots, s\}$, $H^{(j)}$ is a component of $G-\alpha_j(Z_{i_j})$ for $j\in\{1,\ldots,q\}$, $Z^*\cap E(H^{(j)})=\emptyset$ for $j\in\{1,\ldots,s\}$, and $V(H^{(0)})\supset\ldots\supset V(H^{(q)})$. 

If $\alpha(Z_{i_q})$ is a clean cut with respect to $\alpha_q$, the algorithm returns $(\alpha_q,i_q,H^{(q)})$ and stops. Suppose that $\alpha(Z_{i_q})$ is not clean cut with respect to $\alpha_q$.
We show that we can extend the sequence in this case. To do it, we consider the following three cases.

\medskip
\noindent
{\bf Case~1.} $H^{(q)}$ has a bridge $e$. Because loops of $M$ are deleted by {\bf Loop reduction rule}, $e$ is not a bridge of $G$. Hence, each of the two components of $H^{(q)}$ contains an end vertex of an edge of $\alpha_q(Z_{i_q})$. Since $|Z_{i_q}|=3$, there is a component $H'$ of $H^{(q)}-e$ that contains an end vertex of a unique edge $e'$ of  $\alpha_q(Z_{i_q})$ and the other component $H^{(q+1)}$ contains end vertices of two edges of $\alpha_q(Z_{i_q})$.
We obtain that  $\{e,e'\}$ is a minimal cut-set of $G$.   
Let
$\alpha_{q+1}(x)=\alpha_q(x)$ for $x\in E(M_s)\setminus\{\alpha_q^{-1}(e),\alpha_q^{-1}(e')\}$, 
$\alpha_{q+1}(\alpha_q^{-1}(e))=e'$ and $\alpha_{q+1}(\alpha_q^{-1}(e'))=e$.
By Observation~\ref{obs:par}, $\alpha_{q+1}$ is an isomorphism of $M_s$ to $M^*(G)$.
Clearly, $H^{(q+1)}$ is a component of $G-\alpha_{q+1}(Z_{i_q})$ and $V(H^{(q+1)})\subset V(H^{(q)})$. Hence, we can extend the sequence by
$(\alpha_{q+1},i_{q+1},H^{(q+1)})$ for $i_{q+1}=i_q$.

\medskip
\noindent
{\bf Case~2.} There is $i_{q+1}\in\{1,\ldots,s\}$ such that $\alpha_q(Z_{i_{q+1}})\subseteq E(H^{(q)})$. Because $\alpha_q(Z_{i_{q+1}})$ is a minimal cut-set of $G$, we obtain that there is a component  
$H^{(q+1)}$ of $G-\alpha_q(Z_{i_{q+1}})$ such that $V(H^{(q+1)})\subset V(H^{(q)})$. We extend the sequence by $(\alpha_{q+1},i_{q+1},H^{(q+1)})$ for $\alpha_{q+1}=\alpha_q$.

\medskip
\noindent
{\bf Case~3.} There is  $i_{q+1}\in\{1,\ldots,s\}$ such that $\alpha_q(Z_{i_{q+1}})\cap E(H^{(q)})\neq\emptyset$ but
$|\alpha_q(Z_{i_{q+1}})\cap E(H^{(q)})|\leq 2$. 
If $|\alpha_q(Z_{i_{q+1}})\cap E(H^{(q)})|=1$, then the unique edge $e\in\alpha_q(Z_{i_{q+1}})\cap E(H^{(q)})$ is a bridge of $H^{(q)}$, because $\alpha_q(Z_{i_{q+1}})$ is a minimal cut-set. Hence, we have Case~1. Assume that $|\alpha_q(Z_{i_{q+1}})\cap E(H^{(q)})|=1$. Let $H'$ be the component of $G-\alpha_q(Z_{i_q})$ distinct from $H^{(q)}$. Since $|Z_{i_{q+1}}|=3$, we have that 
$|\alpha_q(Z_{i_{q+1}})\cap E(H')|=1$, then the unique edge $e\in\alpha_q(Z_{i_{q+1}})\cap E(H')$ is a bridge of $H'$. By the same arguments as in Case~1, there is $e'\in \alpha_q(Z_{i_q})$ such that $\{e,e'\}$ is a minimal cut-set of $G$. Using Observation~\ref{obs:par}, we construct the isomorphism $\alpha_{q+1}$ of $M_s$ to $M^*(G)$
by defining $\alpha_{q+1}(x)=\alpha_q(x)$ for $x\in E(M_s)\setminus\{\alpha_q^{-1}(e),\alpha_q^{-1}(e')\}$, 
$\alpha_{q+1}(\alpha_q^{-1}(e))=e'$ and $\alpha_{q+1}(\alpha_q^{-1}(e'))=e$.
It remains to observe that $G-\alpha_{q+1}(Z_{i_{q+1}})$ has a component $H^{(q+1)}$ such that $V(H^{(q+1)})\subset V(H^{(q)})$ and extend the sequence by $(\alpha_{q+1},i_{q+1},H^{(q+1)})$.

For  each $j\geq 1$ we have that $V(H^{(j)}\subset V(H^{(j-1)})$. This implies that the sequence 
$$(\alpha_0,i_0,H^{(0)}),\ldots,(\alpha_q,i_q,H^{(q)})$$ has length at most $n$. Hence, after at most $n$ iteration we obtain an isomorphism $\alpha$ and a clean cut with respect to $\alpha$ together with a clean component. Since every step in the iterative construction of the sequence $(\alpha_0,i_0,H^{(0)}),\ldots,(\alpha_q,i_q,H^{(q)})$ can be done in polynomial time, the algorithm is polynomial.

 Recall that in the beginning we assume that $Z^*$ is present. The case when $Z^*$ is absent is more simpler and could be proved as in the case when $Z^*$ is present and thus it is omitted.  
\end{proof}

Using Lemma~\ref{lem:clean}, we can always assume that we  have an isomorphism of $M_s$ to $M^*(G)$ such that for a child $M_\ell$ of $M_s$ in $\mathcal(T)$, $Z=E(M_s)\cap E(M_\ell)$ is mapped to a clean cut. To simplify notations, we assume that $M_s=M^*(G)$ and $Z$ is a clean cut with respect to this isomorphism. Denote by $H$ the clean component. Let  $Z=\{e_1,e_2,e_3\}$ and let 
$e_i=x_iy_i$ for $i\in\{1,2,3\}$, where $y_1,y_2,y_3\in V(H)$. Notice that some $y_1,y_2,y_3$ can be the same.   We first handle the case when $E(H)\cap T=\emptyset$. 

\subsubsection{Cographic sub-leaf: $E(H)\cap T=\emptyset$.}
\label{sec:cslemptyset}
In this case we give a reduction rule  that  reduces the leaf $M_\ell$. Recall that $E(M_\ell)\cap T=\emptyset$. Now we are ready to give  a reduction rule analogous to the one for graphic matroid.

\begin{reduction}[{\bf Cographic $3$-leaf reduction rule}]
\label{rule:cographic3rr}
If $E(H)\cap T=\emptyset$, then do the following. Set $w_\ell(e)=w(e)$ for $e\in E(M_\ell)\setminus Z$, $w_\ell(e_1)=w_\ell(e_2)=w_\ell(e_3)=k+1$.
\begin{itemize}
\item[(i)] For each $i\in \{1,2,3\}$, find the minimum $k_i^{(1)}\leq k$ such that   $(M_\ell,w_\ell,\{e_i\},k_i^{(1)})$ is a 
 yes-instance of \wss\ using  Lemmas~\ref{lem:graphic} or~\ref{lem:cographicnew}, respectively, depending on the type of $M_\ell$. If  $(M_\ell,w_\ell,\{e_i\},k_i^{(1)})$ is a \noinstance\ for every $k_i^{(1)}\leq k$,  then  we set $k_i^{(1)}=k+1$.

\item[(ii)] Find the minimum $p^{(1)}\leq k$ such that $(M_\ell,w_\ell,\{e_1,e_2\},p^{(1)})$ is a yes-instance of \wss\ using  Lemmas~\ref{lem:graphic} or~\ref{lem:cographicnew}, respectively, depending on the type of $M_\ell$.  If  $(M_\ell,w_\ell,\{e_1,e_2\},p^{(1)})$ is a \noinstance\ for every $p^{(1)}\leq k$),  then  we set $p^{(1)}=k+1$.
 If $p^{(1)}\leq k$, then we find an inclusion minimal set $F_{\ell}\subseteq E(M_\ell)\setminus Z$ of weight $p^{(1)}$ that spans $e_1$ and $e_2$.  Observe that Lemmas~\ref{lem:graphic} or~\ref{lem:cographicnew} are only for decision version. However, we can apply standard self reducibility tricks to make them output a solution also. There are circuits $C_1$ and $C_2$ of $M_\ell$ such that 
$e_1\in C_1\subseteq F_{\ell}\cup\{e_1\}$, $e_2\in C_2\subseteq F_\ell \cup\{e_2\}$ and $F_{\ell}=(C_1\setminus\{e_1\})\cup (C_2\setminus\{e_2\})$. Notice that $C_1$ and $C_2$ can be found by finding inclusion minimal subsets of $F_{\ell}$ that span $e_1$ and $e_2$ respectively. Let $p_1^{(1)}=w_\ell(C_1\setminus(C_2\cup\{e_1\}))$,
$p_2^{(1)}=w_\ell(C_2\setminus(C_1\cup\{e_2\}))$ and $p_3^{(1)}=w_\ell(C_1\cap C_2)$. If $p^{(1)}=k+1$, we set $p_1^{(1)}=p_2^{(1)}=p_3^{(1)}=k+1$. 

\end{itemize}
Construct an auxiliary graph $H'$ from $H$ by adding a vertex $u$ and edges $e_1',e_2',e_3'$, where $e_i'=uy_i$ for $i\in\{1,2,3\}$; notice that this could result in multiple edges.
Set $w_h(e)=w(e)$ for $e\in E(H)$ and set $w_h(e_1')=w_h(e_2')=w_h(e_3')=k+1$. 

\begin{itemize}
\item[(iii)] For each $i\in \{1,2,3\}$, find the minimum $k_i^{(2)}\leq k$ such that   $(M^*(H'),w_h,\{e_i'\},k_i^{(2)})$ is a 
 yes-instance of \wss\ using  Lemma~\ref{lem:cographicnew}. If  $(M^*(H'),w_h,\{e_i'\},k_i^{(2)})$ is a \noinstance\ for every $k_i^{(1)}\leq k$,  then  we set $k_i^{(2)}=k+1$.

\item[(iv)] Find the minimum $p^{(2)}\leq k$ such that $(M^*(H'),w_h,\{e_1',e_2'\},p^{(2)})$ is a yes-instance of \wss\ using  Lemma~\ref{lem:cographicnew}. If  $(M^*(H'),w_h,\{e_1',e_2'\},p^{(2)})$  is a \noinstance\ for every $p^{(2)}\leq k$,  then  we set $p^{(2)}=k+1$ . 
 If $p^{(2)}\leq k$, then we find an inclusion minimal set $F_h\subseteq E(H')\setminus Z$ of weight $p^{(2)}$ that spans $e_1'$ and $e_2'$. Observe that Lemma~\ref{lem:cographicnew} is only for decision version. However, we can apply standard self reducibility tricks to make it output a solution also. 
  There are circuits $C_1$ and $C_2$ of $M^*(H')$ such that 
$e_1'\in C_1\subseteq F_{h}\cup\{e_1'\}$, $e_2\in C_2\subseteq F_h\cup\{e_2'\}$ and $F_{h}=(C_1\setminus\{e_1'\})\cup (C_2\setminus\{e_2'\})$.  Notice that $C_1$ and $C_2$ can be found by finding inclusion minimal subsets of $F_h$ that span $e_1'$ and $e_2'$ respectively. Let $p_1^{(2)}=w_h(C_1\setminus(C_2\cup\{e_1'\}))$,
$p_2^{(2)}=w_h(C_2\setminus(C_1\cup\{e_2'\}))$ and $p_3^{(3)}=w_h(C_1\cap C_2)$. If $p^{(2)}=k+1$, we set $p_1^{(2)}=p_2^{(2)}=p_3^{(2)}=k+1$. 
\end{itemize}
Construct the graph $G'$ from $G-V(H)$ by adding three pairwise adjacent vertices $z_1,z_2,z_3$ and edges $x_1z_1,x_2z_2,x_3z_3$. 
Let $M'$ the matroid defined by $\mathcal{T}'=\mathcal{T}-M_\ell$, where $M_s$ is replaced by $M^*(G')$.
The weight function $w'\colon E(M')\rightarrow\mathbb{N}$ is defined by setting 
$w'(e)=w(e)$ for $e\in E(M')\setminus \{x_1z_1,x_2z_2,x_2z_3,z_1z_2,z_2z_3,z_1z_3\}$, 
$w'(x_iz_i)=\min\{k_i^{1},k_i^{2}\}$ for $i\in\{1,2,3\}$. If $p^{(1)}\leq p^{(2)}$, then 
$w'(z_1z_3)=p_1^{(1)}$, $w'(z_2z_3)=p_2^{(1)}$ and $w'(z_1z_2)=p_3^{(1)}$, and 
$w'(z_1z_3)=p_1^{(2)}$, $w'(z_2z_3)=p_2^{(2)}$ and $w'(v_1v_2)=p_3^{(2)}$ otherwise.
The reduced instance is $(M',w',T,k)$.  
\end{reduction}

Similarl to Observation~\ref{obs:ineq}, we observe the following using Observation~\ref{obs:triangle}.

\begin{observation}\label{obs:ineq-2}
For 
each $i\in\{1,2,3\}$, and $j,q\in\{1,2,3\}\setminus \{i\}$ we have that 
$w'(z_iz_j)+w'(z_iz_q)\geq w'(x_iz_i)$.
Also, for any distinct $i,j\in \{1,2,3\}$ and $q\in\{1,2\}$, if $k_i^{(q)}+k_j^{(q)}\leq k$, then $p^{(q)}\leq k_i^{(q)}+k_j^{(q)}$. 
\end{observation}

The next lemma proves the safeness of the Reduction Rule~\ref{rule:cographic3rr}. 

\begin{lemma}\label{lem:red-3-cograph}
Reduction Rule~\ref{rule:cographic3rr} is safe and can be applied 
in $2^{\cO(k)}\cdot ||M||^{\cO(1)}$ time. 
\end{lemma}

\begin{proof}
Denote by $M''$ the matroid defined by $\mathcal{T}'=\mathcal{T}-M_\ell$. 
To prove that the rule is safe, assume first that $(M,w,T,k)$ is a yes-instance. Then there is an inclusion minimal 
set $F\subseteq E(M)\setminus T$ of weight at most $k$ that spans $T$. 

Suppose that $F\cap E(M_\ell)=\emptyset$ and $F\cap E(H)=\emptyset$. By the definition of $G'$, any minimal cut-set $C$ of $G$ such that $C\cap Z$ and 
$C\cap E(H)=\emptyset$ is a minimal cut-set of $G'$, because $H$ is a connected graph. We obtain that $F$ spans $T$ in $M'$ and 
$(M',w',T,k)$ is a yes-instance.

Assume that $F\cap E(M_\ell)\neq\emptyset$ and $F\cap E(H)=\emptyset$. The proof for this case is, in fact, almost identical to the proof for {\bf Graphic $3$-leaf Reduction Rule}. 

For each $t\in T$, there is a circuit $C_t$ of $M$ such that $t\in C\subseteq F\cup\{t\}$. If $C_t\cap E(M_\ell)\neq\emptyset$, $C_t=C_t'\bigtriangleup C_t''$, where $C_t'$ is a cycle of $M''$ and  $C_t''$ is a cycle of $M_\ell$. By Observation~\ref{obs:through-sum}, we can assume that $C_t'$ and $C_t''$ are circuits of $M''$ and $M_\ell$ respectively 
and 
$C_t'\cap C_t''$ contains the unique element $e_i$, i.e., $C_t$ goes through $e_i$. 
Notice that every $(C_t'\setminus\{e_i\})\cup\{x_iz_i\}$ is a minimal cut-set of $G'$ and, therefore, a circuit of $M^*(G')$. 
We consider the following three cases.

\medskip
\noindent
{\bf Case 1.} There is a unique $e_i\in Z$ such that for any $t\in T$, either $C_t\subseteq E(M'')$ or $C_t$ goes through $e_i$. Let $F'=(F\cap E(M''))\cup\{x_iz_i\}$. 

We show that $F'$ spans $T$ in $M'$. Let $t\in T$. If $C_t\subseteq E(M'')$, then $t\in C_t\subseteq (F\cap E(M''))\cup\{t\}$ and, therefore, $F'$ spans $t$ in $M'$. Suppose that $C_t\cap E(M_\ell)\neq \emptyset$. Then $C_t=C_t'\bigtriangleup C_t''$, where  $C_t'$ is a circuit of $M''$, $C_t''$ is a circuit of $M_\ell$ and $C_t'\cap C_t''=\{e_i\}$. We have that $t\in C_t'\cup\{t\}$ and $((C_t'\setminus\{e_i\})\cup\{x_iz_i\})\setminus\{t\} \subseteq F'$ spans $t$. 

Because $F\cap E(M_\ell)\neq\emptyset$ and $F$ is inclusion minimal spanning set, there is $t\in T$ such that $C_t$ goes through $e_i$. 
Let  $C_t=C_t'\bigtriangleup C_t''$, where  $C_t'$ is a circuit of $M''$, $C_t''$ is a circuit of $M_\ell$ and $C_t'\cap C_t''=\{e_i\}$. Notice that $C_t''\setminus\{e_i\}$ spans $e_i$ in $M_\ell$. Hence, $w_\ell(C_t''\setminus \{e_i\})\leq k_i^{(1)}$. 
Because $w'(x_iz_i)\leq k_i^{(1)}$, we conclude that $w'(F')\leq w(F)$.

Since $F'\subseteq E(M')\setminus T$ spans $T$ and has the weight at most $k$ in $M'$, $(M',w',T,k)$ is a yes-instance.

\medskip
\noindent
{\bf Case 2.} There are two distinct  $e_i,e_j\in Z$ such that for any $t\in T$, either $C_t\subseteq E(M'')$, or $C_t$ goes through $e_i$, or $C_t$ goes through $e_j$, and at least one $C_t$ goes through $e_i$ and at least one $C_t$ goes through $e_j$. 
Let $F'=(F\cap E(M''))\cup\{z_1z_2,z_2z_3,z_1z_3\}$. 

We claim that $F'$ spans $T$ in $M'$. Let $t\in T$. If $C_t\subseteq E(M'')$, then $t\in C_t\subseteq (F\cap E(M''))\cup\{t\}$ and, therefore, $F'$ spans $t$ in $M'$. Suppose that $C_t\cap E(M_\ell)\neq \emptyset$. Then $C_t=C_t'\bigtriangleup C_t''$, where  $C_t'$ is a circuit of $M''$, $C_t''$ is a circuit
 of $M_\ell$ and either $C_t'\cap C_t''=\{e_i\}$ or  $C_t'\cap C_t''=\{e_j\}$. By symmetry, let $C_t'\cap C_t''=\{e_i\}$. Because $\{x_iz_i,z_iz_{i-1},z_iz_{i+1}\}$ (here and further it is assumed that $z_0=z_3$ and $z_4=z_1$)
is a minimal cut-set of $G$,  $\{x_iz_i,z_iz_{i-1},z_iz_{i+1}\}$ is a circuit of $M'$ and 
$C_t'''=((C_t'\setminus \{e_i\})\cup\{x_iz_i\})\bigtriangleup \{x_iz_i,z_iz_{i-1},z_iz_{i+1}\}$ is a cycle of $M'$.
We have that $t\in C_t'''\cup\{t\}$ and $C_t'''\setminus\{t\}\subseteq F'$ spans $t$. 

Because $F\cap E(M_\ell)\neq\emptyset$, there is $t\in T$ such that $C_t$ goes through $e_i$ and there is $t'\in T$ such that $C_{t'}$ goes through $e_j$. 
Let  $C_t=C_t'\bigtriangleup C_t''$ and $C_{t'}=C_{t'}'\bigtriangleup C_{t''}''$, where  $C_t',C_{t'}'$ are cycles of $M''$, $C_t'',C_{t'}''$ are cycles of  $M_\ell$ and $C_t'\cap C_t''=\{e_i\}$, $C_{t'}'\cap C_{t'}''=\{e_j\}$. Notice that $C_t''\setminus\{e_i\}$ spans $e_i$ in $M_\ell$ and $C_{t'}''\setminus \{e_j\}$ spans $e_j$. Hence, $w_\ell((C_t''\setminus \{e_i\})\cup (C_{t'}''\setminus\{e_j\}))\geq w_\ell(F_\ell)=p^{(1)}$ by Observation~\ref{obs:triangle}. Because $w'(\{z_1z_2,z_2z_3,z_1z_3\})\geq p^{(1)}$,
$w'(F')\leq w(F)$.

Since $F'\subseteq E(M')\setminus T$ spans $T$ and has the weight at most $k$ in $M'$, $(M',w',T,k)$ is a yes-instance.

\medskip
\noindent
{\bf Case 3.} For each $i\in \{1,2,3\}$, there is $t\in T$ such that $C_t$ goes through $e_i$.  As in Case~2, we set  $F'=(F\cap E(M''))\cup\{z_1z_2,z_2z_3,z_1z_3\}$ and use the same arguments to show that $F'\subseteq E(M')\setminus T$ spans $T$ and has the weight at most $k$ in $M'$.

\medskip
Suppose that $F\cap E(M_\ell)=\emptyset$ and $F\cap E(H)\neq\emptyset$. 

For each $t\in T$, there is a circuit $C_t$ of $M$ such that $t\in C\subseteq F\cup\{t\}$. By the definition of 1, 2 and 3-sums and Observation~\ref{obs:through-sum}, we have that $C_t=C_t'\bigtriangleup C^{(1)}\bigtriangleup\ldots\bigtriangleup C^{(q)}$, where $C_t'$ is a circuit of $M_s$ and each 
 $C^{(1)},\ldots,C^{(q)}$  is a circuit of child of $M_s$ in $\mathcal{T}$ or a circuit in the matroid defined by the conflict tree $\mathcal{T}''$ obtained from $\mathcal{T}$ by the deletion of $M_s$ and its children.  Notice that if $C_t\cap E(H)\neq\emptyset$, then 
$C_t\cap E(H)$ is a minimal cut-set of $H$. Moreover, each component of $H-C_t\cap E(H)$ contains a vertex from the set $\{y_1,y_2,y_3\}$.

We consider the following three cases.

\medskip
\noindent
{\bf Case 1.} There is a unique $i\in \{1,2,3\}$ such that for any $t\in T$, either $C_t\cap E(H)=\emptyset$ or $y_i$ is in one component of $H-C_t\cap E(H)$ and $y_{i-1},y_{i+1}$ are in the other.
Let $F'=(F\setminus E(H))\cup\{x_iz_i\}$. 

We show that $F'$ spans $T$ in $M'$. Let $t\in T$. If $C_t\cap E(H)=\emptyset$, then $F'$ spans $t$ in $M'$, because 
$C_t$ is a circuit of $M^*(G')$ as $H$ is connected.  Suppose that $C_t\cap E(H)\neq \emptyset$. Consider $C_t''=(C_t\setminus (C_t\cap E(H)))\cup\{x_iz_i\}$.
Since $(C_t'\setminus (C_t\cap E(H)))\cup\{x_iz_i\}$ is a minimal cut-set of $G$, we obtain that $C_t''\setminus \{t\}\subseteq F'$ spans $t$ in $M'$.

Because $F\cap E(H)\neq\emptyset$,  there is $t\in T$ such that $C_t\cap E(H)\neq\emptyset$. Observe that $w(C_t\cap E(H))\geq k^{(2)}_i\geq w'(x_iz_i)$. Hence,  
 $w'(F')\leq w(F)$.

Since $F'\subseteq E(M')\setminus T$ spans $T$ and has the weight at most $k$ in $M'$, $(M',w',T,k)$ is a yes-instance.

\medskip
\noindent
{\bf Case 2.} There are two distinct  $i,j\in \{1,2,3\}$ such that for any $t\in T$, either i) $C_t\cap E(H)=\emptyset$ or ii) $y_i$ is in one component of $H-C_t\cap E(H)$ and $y_{i-1},y_{i+1}$ are in the other or iii) $y_j$ is in one component of $H-C_t\cap E(H)$ and $y_{j-1},y_{j+1}$ are in the other, and for at least one $t$, ii) holds and for at least one $t$ iii) is fulfilled.
Let $F'=(F\setminus E(H))\cup\{z_1z_2,z_2z_3,z_1z_3\}$. 

We claim that $F'$ spans $T$ in $M'$. Let $t\in T$. If $C_t\cap E(H)=\emptyset$, then $F'$ spans $t$ in $M'$, because 
$C_t'$ is a circuit of $M^*(G')$ as $H$ is connected.  Suppose that $C_t\cap E(H)\neq \emptyset$. By symmetry, assume without loss of generality that  ii) is fulfilled for $C_t$.
Consider $C_t''=(C_t\setminus (C_t\cap E(H)))\cup\{z_iz_{i-1},z_{i}z_{i+1}\}$.
Since $(C_t'\setminus (C_t\cap E(H)))\cup\{x_iz_i\}$ is a minimal cut-set of $G$, we obtain that $C_t''\setminus \{t\}\subseteq F'$ spans $t$ in $M'$.

Because there are distinct $i,j\in\{1,2,3\}$ such that ii) holds for some $t\in T$ iii) for some $t'\in T$, we have that $w(C_t\cap E(H))+w(C_{t'}\cap E(H))\geq k^{2}\geq w'(\{z_1z_2,z_2z_3,z_1z_3\})$.
 Hence,  
$w'(F')\leq w(F)$.
As $F'\subseteq E(M')\setminus T$ spans $T$ and has the weight at most $k$ in $M'$, $(M',w',T,k)$ is a yes-instance.

\medskip
\noindent
{\bf Case 3.} For each $i\in \{1,2,3\}$, there is $t\in T$ such that $y_i$ is in one component of $H-C_t\cap E(H)$ and $y_{i-1},y_{i+1}$ are in the other.  As in Case~2, we set  $F'=(F\setminus E(H))\cup\{z_1z_2,z_2z_3,z_1z_3\}$ and use the same arguments to show that $F'\subseteq E(M')\setminus T$ spans $T$ and has the weight at most $k$ in $M'$.

\medskip
Finally, assume that $F\cap E(M_\ell)\neq\emptyset$ and $F\cap E(H)\neq\emptyset$. 
For each $t\in T$, there is a circuit $C_t$ of $M$ such that $t\in C\subseteq F\cup\{t\}$. 
Then there is $i\in\{1,2,3\}$ such that $C_t=C_t'\bigtriangleup C_t''$, where $C_t'$ and $C_t''$ are circuits of $M''$ and $M_\ell$, and $C_t$ goes through $e_i$,
i.e, $C_t'\cap C_t''=\{e_i\}$. Also there is $j\in\{1,2,3\}$ such that $y_j$ is in one component of $H-C_t\cap E(H)$ and $y_{j-1},y_{j+1}$ are in the other. 
Notice that $i\neq j$, as otherwise $F$ contains a dependent set $(C_t\cap E(H))\cup \{e_i\}$, where $y_i$ is in one component of $H-C_t\cap E(H)$ and $y_{i-1},y_{i+1}$ are in the other,
contradicting minimality of $F$. 
Let 
$F'=((F\cap E(M''))\setminus E(H))\cup\{x_iz_i,x_jz_j\}$. Denote by $q\in\{1,2,3\}$ the element of the set distinct from $i$ and $j$.
 
We claim that $F'$ spans $T$ in $M'$. Let $t\in T$.

If $C_t\cap E(H)=\emptyset$ and $C_t\subseteq E(M'')$, then it is straightforward to verify that $C_t\setminus\{t\}$ spans $t$ in $M'$ and, therefore, $F'$ spans $t$. 

Suppose that $C_t\cap E(H)\neq \emptyset$ and $C_t\subseteq E(M'')$. Then $C_t\cap E(H)$ is a minimal cut-set of $H$ such that a vertex $y_f$ is in one component of $H-C_t\cap E(H)$
and $y_{f-1},y_{f+1}$ are in the other. If $f=i$ or $f=j$, then in the same way as in the case, where $F\cap E(M_\ell)=\emptyset$ and $F\cap E(H)\neq\emptyset$, we have that 
$((C_t\setminus E(H))\cup\{x_fz_f\})\setminus\{t\}$ spans $t$. Suppose that $f=q$. Then we observe that $((C_t\setminus E(H))\cup\{x_iz_i,x_jy_j\})\setminus\{t\}$ spans $t$.
Hence, $F'$ spans $t$.

Suppose that $C_t\cap E(H)=\emptyset$ and $C_t\cap E(M_\ell)\neq\emptyset$. Then $C_t=C_t'\bigtriangleup C_t''$, where $C_t'$ and $C_t''$ are cycles of $M''$ and $M_\ell$ respectively, and $C_t$ goes through some $e_f$ for $f\in\{1,2,3\}$. 
 If $f=i$ or $f=j$, then in the same way as in the case, where $F\cap E(M_\ell)\neq\emptyset$ and $F\cap E(H)=\emptyset$, we have that 
$((C_t'\setminus\{e_f\})\cup\{x_fz_f\})\setminus\{t\} \subseteq F'$ spans $t$.
Suppose that $f=q$. Then we 
observe that 
$((C_t'\setminus\{e_f\})\cup\{x_iz_i,x_jy_j\})\setminus\{t\} \subseteq F'$ spans $t$, because
$\{x_1z_1,x_2z_2,x_3z_3\}$ is a circuit of $M'$. 

Suppose now that $C_t\cap E(H)\neq\emptyset$ and $C_t\cap E(M_\ell)\neq\emptyset$. Then $C_t\cap E(H)$ is a minimal cut-set of $H$ such that a vertex $y_f$ is in one component of $H-C_t\cap E(H)$
and $y_{f-1},y_{f+1}$ are in the other. Also $C_t=C_t'\bigtriangleup C_t''$, where $C_t'$ and $C_t''$ are circuits of $M''$ and $M_\ell$ respectively, and $C_t$ goes through some $e_g$ for $g\in\{1,2,3\}$. 
Notice that $f\neq g$, as otherwise $C_t'$ contains a dependent set $(C_t\cap E(H))\cup \{e_f\}$ contradicting minimality of circuits. 
If $\{f,g\}=\{i,j\}$, we obtain that $(((C_t'\setminus E(H))\setminus\{e_f\})\cup\{x_fz_f,x_gz_g\})\setminus\{t\} \subseteq F'$ spans $t$ by the same arguments as in previous cases.
If $\{f,g\}\neq\{i,j\}$, then let $q'\in\{1,2,3\}$ be distinct form $f,g$. Clearly, $q'\in\{i,j\}$. 
Then 
$(((C_t'\setminus E(H))\setminus\{e_f\})\cup\{x_{q'}z_{q'}\})\setminus\{t\} \subseteq F'$ spans $t$ spans $t$, because
$\{x_1z_1,x_2z_2,x_3z_3\}$ is a circuit of $M'$. 

Now we show that $w'(F)\leq k$. Recall that there is  $C_t=C_t'\bigtriangleup C_t''$, where $C_t'$ and $C_t''$ are circuits of $M''$ and $M_\ell$, and $C_t$ goes through $e_i$. Observe that $w'(e_i)\leq k_i^{(1)}\leq w(C_t''\setminus \{e_i\})$. 
Recall also that there is $C_t$ such that $C_t\cap E(H)\neq\emptyset$ and $y_j$ is in one component of $H-C_t\cap E(H)$ and $y_{j-1},y_{j+1}$ are in the other. We have that 
$w'(x_jz_j)\leq k_j^{(2)}\leq w(C_t\cap E(H))$. It implies that $w'(F)\leq k$.

\medskip
We considered all possible cases and obtained that if the original instance $(M,w,T,k)$ is a yes-instance, then the reduced instance $(M',w',T,k)$ is also a yes-instance.
Assume now that the reduced instance $(M',w',T,k)$ is a yes-instance. Let $F'\subseteq E(M')\setminus T$ be an  inclusion minimal set of  weight at most $k$ that spans $T$ in $M'$. 

Let $S=\{x_1z_1,x_2z_2,x_3z_3,z_1z_2,z_2z_3,z_1z_3\}$.
If $F'\cap S=\emptyset$, then we have that $F'$ spans $T$ in $M$ as well. Assume from now that  
$F'\cap S\neq\emptyset$. 

Notice that $|F'\cap\{z_1z_2,z_2z_3,z_1z_3\}|\neq 1$, because $z_1z_2,z_2z_3,z_1z_3$ induce a cycle in $C'$. Observe also  that if $F'\cap\{z_1z_2,z_2z_3,z_1z_3\}=\{z_{i-1}z_i,z_iz_{i+1}\}$ for some $i\in\{1,2,3\}$, then by Observation~\ref{obs:ineq-2} we can replace $z_{i-1}z_i,z_iz_{i+1}$ by $x_iz_i$ in $F$ using the fact that $z_{i-1}z_i,z_iz_{i+1},x_iz_i$ is a cut-set of $G'$. Hence, without loss of generality we assume that  either $F'\cap\{z_1z_2,z_2z_3,z_1z_3\}=\emptyset$ or $z_1z_2,z_2z_3,z_1z_3\in F'$. We have that  $|F'\cap\{x_1z_1,x_2z_2,x_3z_3\}|\leq 2$, because $\{x_1z_1,x_2z_2,x_3z_3\}$ is a minimal cut-set of $G'$, and if $z_1z_2,z_2z_3,z_1z_3\in F'$, then 
$F'\cap\{x_1z_1,x_2z_2,x_3z_3\}=\emptyset$ by the minimality of $F'$. We consider the cases according to these possibilities.

\medskip
\noindent
{\bf Case~1.} $z_1z_2,z_2z_3,z_1z_3\in F'$. 

If $p^{(1)}\leq p^{(2)}$, then recall that $(M_\ell,w_\ell,\{e_1,e_2\},p^{(1)})$ of is a yes-instance of \wss. Let $F_\ell$ be a set of weight at most $p^{(1)}$ in that spans $e_1$ and $e_2$ in $M_\ell$. Notice that $F_\ell$ spans $e_3$ by Observation~\ref{obs:triangle}. Notice also that $e_1,e_2,e_3\notin F_\ell$.
 We define $F=(F'\setminus\{z_1z_2,z_2z_3,z_1z_3\})\cup F_\ell$. Clearly, $F\subseteq E(M)\setminus T$ and $w(F)\leq k$ as $w'(\{z_1z_2,z_2z_3,z_1z_3\})=p^{(1)}$. We claim that $F$ spans $T$ in $M$.
Consider $t\in T$. There is a circuit $C_t'$ of $M'$ such that $t\in C_t'\subseteq F'\cup\{t\}$. 
If $C_t'\cap\{z_1z_2,z_2z_3,z_1z_3\}=\emptyset$, then $C_t'\setminus\{t\}$ spans $t$ in $M$. Suppose that $C_t'\cap\{z_1z_2,z_2z_3,z_1z_3\}\neq\emptyset$. Notice that because $z_1z_2,z_2z_3,z_1z_3$ form a triangle in $G'$, $C_t'$ contains exactly two elements of $\{z_1z_2,z_2z_3,z_1z_3\}$. By symmetry, assume without loss of generality that $z_1z_2,z_2z_3\in C_t'$. There is a circuit $C$ of $M_\ell$ such that $e_1\in C\subseteq F_\ell\cup\{e_1\}$. Observe that for any $X\subseteq E(G')$ such that $X\cap S=\{z_1z_2,z_1z_3\}$, $X$ is a minimal cut-set of $G'$ if and only if $(X\setminus\{z_1z_2,z_1z_3\})\cup\{e_1\}$ is a minimal cut-set of $G$. It implies that
$C_t=(C_t'\setminus\{z_1z_2,z_1z_3\})\cup (C\setminus\{e_1\})\subseteq F$ is a cycle of $M$. Hence, $F$ spans $t$.

Suppose that $p^{(2)}<p^{(1)}$. Recall that  $(M^*(H'),w_h,\{e_1',e_2'\},p^{(2)})$ is a yes-instance of 
\wss. Let $F_h$ be a set of weight at most 
$p^{(2)}$ in that spans $e_1'$ and $e_2'$ in $M^*(H')$. Notice that $F_h$ spans $e_3'$ by Observation~\ref{obs:triangle}. Notice also that $e_1',e_2',e_3'\notin F_h$.
 We define $F=(F'\setminus\{z_1z_2,z_2z_3,z_1z_3\})\cup F_h$. Clearly, $F\subseteq E(M)\setminus T$ and $w(F)\leq k$ as $w'(\{z_1z_2,z_2z_3,z_1z_3\})=p^{(2)}$. We claim that $F$ spans $T$ in $M$.
Consider $t\in T$. There is a circuit $C_t'$ of $M'$ such that $t\in C_t'\subseteq F'\cup\{t\}$. 
If $C_t'\cap\{z_1z_2,z_2z_3,z_1z_3\}=\emptyset$, then $C_t'\setminus\{t\}$ spans $t$ in $M$. Suppose that $C_t'\cap\{z_1z_2,z_2z_3,z_1z_3\}\neq\emptyset$. Notice that because $z_1z_2,z_2z_3,z_1z_3$ form a triangle in $G'$, $C_t'$ contains exactly two elements of $\{z_1z_2,z_2z_3,z_1z_3\}$. By symmetry, assume without loss of generality that $z_1z_2,z_2z_3\in C_t'$. There is a circuit $C$ of $M_h$ such that $e_1'\in C\subseteq F_h\cup\{e_1'\}$. 
Notice that for any $X\subseteq E(G')$ such that $X\cap S=\{z_1z_2,z_1z_3\}$, $X$ is a minimal cut-set of $G'$ if and only if $(X\setminus\{z_1z_2,z_1z_3\})\cup Y$ is a minimal cut-set of $G$ for a minimal cut-set $Y$ of $H$ such that $y_1$ is in one component of $H-Y$ and $y_2,y_3$ are in the other.
It implies that
$C_t=(C_t'\setminus\{z_1z_2,z_1z_3\})\cup (C\setminus\{e_1'\})\subseteq F$ is a cycle of $M$. Hence, $F$ spans $t$.

\medskip
\noindent
{\bf Case~2.} $F'\cap S=\{x_iz_i\}$ for $i\in\{1,2,3\}$. 

 Suppose first that $k_i^{(1)}\leq k_i^{(2)}$. Then $(M_\ell,w_\ell,\{e_i\},k_i^{(1)})$ is a 
 yes-instance of \wss.  Let $F_\ell$ be a set of weight at most $k_i^{(1)}$ in that spans $e_i$ in $M_\ell$. Notice $e_1,e_2,e_3\notin F_\ell$.
 We define $F=(F'\setminus\{x_iz_i\})\cup F_\ell$. Clearly, $F\subseteq E(M)\setminus T$ and $w(F)\leq k$ as $w'(x_iz_i)=k_i^{(1)}$. We claim that $F$ spans $T$ in $M$.
Consider $t\in T$. There is a circuit $C_t'$ of $M'$ such that $t\in C_t'\subseteq F'\cup\{t\}$. 
If $x_iz_i\notin C_t'$, then $C_t'\setminus\{t\}$ spans $t$ in $M$. Suppose that $x_iz_i\in C_t'$. There is a circuit $C$ of $M_\ell$ such that $e_i\in C\subseteq F_\ell\cup\{e_i\}$. Observe that for any $X\subseteq E(G')$ such that $X\cap S=\{x_iz_i\}$, $X$ is a minimal cut-set of $G'$ if and only if $(X\setminus\{x_iz_i\})\cup\{e_i\}$ is a minimal cut-set of $G$. It implies that
$C_t=(C_t'\setminus\{x_iz_i\})\cup (C\setminus\{e_i\})\subseteq F$ is a cycle of $M$. Hence, $F$ spans $t$.

Assume that $k_i^{(2)}<k_i^{(1)}$.  Recall that  $(M^*(H'),w_h,\{e_i'\},k_i^{(2)})$ is a yes-instance of \wss. Let $F_h$ be a set of weight at most 
$k_i^{(2)}$ in that spans $e_i'$ in $M^*(H')$.  Notice that $e_1',e_2',e_3'\notin F_h$.
 We define $F=(F'\setminus\{x_iz_i\})\cup F_h$. Clearly, $F\subseteq E(M)\setminus T$ and $w(F)\leq k$ as $w'(\{x_iz_i\})=k_i^{(2)}$. We claim that $F$ spans $T$ in $M$.
Consider $t\in T$. There is a circuit $C_t'$ of $M'$ such that $t\in C_t'\subseteq F'\cup\{t\}$. 
If $x_iz_i\notin C_t'$, then $C_t'\setminus\{t\}$ spans $t$ in $M$. Suppose that $x_iz_i\in C_t'$. There is a circuit $C$ of $M_h$ such that $e_i'\in C\subseteq F_h\cup\{e_i'\}$. 
Observe that any $X\subseteq E(G')$ such that $X\cap S=\{x_iz_i\}$, $X$ is a minimal cut-set of $G'$ if and only if $(X\setminus\{x_iz_i\})\cup Y$ is a minimal cut-set of $G$  for a minimal cut-set $Y$ of $H$ such that $y_i$ is in one component of $H-Y$ and $y_{i-1},y_{i+1}$ are in the other.
It implies that
$C_t=(C_t'\setminus\{xIz_i\})\cup (C\setminus\{e_i'\})\subseteq F$ is a cycle of $M$. Hence, $F$ spans $t$.

\medskip
\noindent
{\bf Case~3.} $F'\cap S=\{x_iz_i,x_jz_j\}$ for two distinct $i,j\in\{1,2,3\}$. 

Suppose that $w'(x_iz_i)=k_i^{(1)}$ and $w'(x_jz_j)=k_j^{(1)}$. By Observation~\ref{obs:ineq-2}, $p^{(1)}\leq k_i^{(1)}+k_j^{(1)}$. We have that  $(M_\ell,w_\ell,\{e_1,e_2\},p^{(1)})$ is a yes-instance of \wss. Let $F_\ell$ be a set of weight at most $p^{(1)}$ in that spans $e_1$ and $e_2$ in $M_\ell$. Notice that $F_\ell$ spans $e_3$ by Observation~\ref{obs:triangle}. Notice also that $e_1,e_2,e_3\notin F_\ell$.
 We define $F=(F'\setminus\{x_iz_i,x_jz_j\})\cup F_\ell$. Clearly, $F\subseteq E(M)\setminus T$ and $w(F)\leq k$ as $w'(\{x_iz_i,x_jz_j\})\geq p^{(1)}$. 
In the same way as in Case~1, we obtain that $F$ spans $T$ in $M$.

Assume that $w'(x_iz_i)=k_i^{(2)}$ and $w'(x_jz_j)=k_j^{(2)}$. By Observation~\ref{obs:ineq-2}, $p^{(2)}\leq k_i^{(2)}+k_j^{(2)}$. Recall that $(M^*(H'),w_h,\{e_1',e_2'\},p^{(2)})$ is a yes-instance of \wss. Let $F_h$ be a set of weight at most 
$p^{(2)}$ in that spans $e_1'$ and $e_2'$ in $M^*(H')$. Notice that $F_h$ spans $e_3'$ by Observation~\ref{obs:triangle}. Notice also that $e_1',e_2',e_3'\notin F_h$.
 We define $F=(F'\setminus\{x_iz_i,x_jz_j\})\cup F_h$. Clearly, $F\subseteq E(M)\setminus T$ and $w(F)\leq k$ as $w'(\{x_iz_i,x_jz_j\})\geq p^{(2)}$. By the same arguments as in Case~1, we have that $F$ spans $T$ in $M$. 
 
Suppose now that $w'(x_iz_i)=k_i^{(1)}$ and $w'(x_jz_j)=k_j^{(2)}$ or, symmetrically,  $w'(x_iz_i)=k_i^{(2)}$ and $w'(x_jz_j)=k_j^{(1)}$. Assume that $w'(x_iz_i)=k_i^{(1)}$ and $w'(x_jz_j)=k_j^{(2)}$, as the second possibility is analysed by the same arguments.  
We have that $(M_\ell,w_\ell,\{e_i\},k_i^{(1)})$ is a 
 yes-instance of \wss.  Let $F_\ell$ be a set of weight at most $k_i^{(1)}$ in that spans $e_i$ in $M_\ell$. Notice $e_1,e_2,e_3\notin F_\ell$. We have also that  $(M^*(H'),w_h,\{e_i'\},k_j^{(2)})$ is a yes-instance of \wss. Let $F_h$ be a set of weight at most 
$k_j^{(2)}$ in that spans $e_j'$ in $M^*(H')$.  Notice that $e_1',e_2',e_3'\notin F_h$.
We define $F=(F'\setminus\{x_iz_i,x_jz_j\})\cup F_\ell \cup F_h$. Clearly, $F\subseteq E(M)\setminus T$ and 
$w(F)\leq k$ as $w'(\{x_iz_i\})\leq k_i^{(1)}$ and $w'(\{x_iz_i\})\leq k_i^{(1)}$. 
We show that $F$ spans $T$.  
Consider $t\in T$. There is a circuit $C_t'$ of $M'$ such that $t\in C_t'\subseteq F'\cup\{t\}$. 
There is a circuit $C$ of $M_\ell$ such that $e_i\in C\subseteq F_\ell\cup\{e_i\}$, and there is a circuit
$C'$ of $M_h$ such that $e_j'\in C\subseteq F_h\cup\{e_j'\}$. 
If $x_iz_i,x_jz_j\notin C_t'$, then $C_t'\setminus\{t\}$ spans $t$ in $M$. Suppose that $x_iz_i\in C_t'$
but $x_jz_j\notin C_t'$. Then by the same arguments as were used to analyse the first possibility of Case~2, we show
that $C_t=(C_t'\setminus \{x_iz_i\})\cup(C\setminus \{e_i\})$
 is a cycle of $M$ such that $t\in C_t\subseteq F\cup\{t\}$. 
If  $x_iz_i\notin C_t'$ and $x_jz_j\in C_t'$. Then by the same arguments as were used to analyse the second possibility of Case~2, we obtain
that $C_t=(C_t'\setminus \{x_jz_j\})\cup(C'\setminus \{e_j'\})$
 is a cycle of $M$ such that $t\in C_t\subseteq F\cup\{t\}$. Finally,
 if $x_iz_i,x_jz_j\in C_t'$, we consider $C_t=(C_t'\setminus \{x_jz_j\})\cup (C\setminus \{e_i\})\cup(C'\setminus \{e_j'\})$ and essentially by the same arguments as in Case~2, obtain that $C_t$ is a cycle of $M$ and 
  $t\in C_t\subseteq F\cup\{t\}$. Hence, in all possible cases $F$ spans $t$.

This completes the correctness proof. 
From the description of Reduction Rule~\ref{lem:red-3-cograph} and 
Lemma~\ref{lem:simpleAlgo}, it follows that Reduction Rule~\ref{rule:graphic3leafrule} can be applied in time $2^{\cO(k)}\cdot ||M||^{\cO(1)}$. 
\end{proof}

\subsubsection{Cographic sub-leaf: $E(H)\cap T\neq \emptyset$.}
From now onwards  we assume that $E(H)\cap T\neq\emptyset$. We either reduce $H$ or recursively solve the problem on smaller  $H$. Rather than describing these steps, 
 we observe that we can decompose $M_s$ further and apply the already described 
 Reduction Rule~\ref{rule:one-leaf} ({\bf $1$-Leaf reduction rule})  or Branching Rules~\ref{brule:2lb} 
 ({\bf $2$-Leaf branching}) and \ref{brule:3lb}  ({\bf $3$-Leaf branching}).

 We use the following fact about matroid decompositions  (see~\cite{Truemper92}).Since we apply the decomposition theorem  for the specific case of bond matroids, for convenience we state it in terms of graphs.  Let $G$ be a graph.  A pair $(X,Y)$ of nonempty subsets  $X,Y\subset V(G)$ is a \emph{separation} of $G$
 if $X\cup Y=V(G)$ and no vertex of $X\setminus Y$ is adjacent to a vertex of $Y\setminus Y$. For our convenience we assume that $(X,Y)$ is an ordered pair. The next lemma can be derived from either the general results of~\cite[Chapter 8]{Truemper92}, or it can be proved directly using  definitions of 
 $1$-, $2$- and 3-sums and the fact that the circuits of the bond matroid of $G$ are exactly the 
minimal cut-sets of $G$.

\begin{lemma}\label{lem:decomp-sub}
Let $(X,Y)$ be a separation of a graph $G$, $H_1=G[X]$ and $H_2=G[Y]-E(G_1)$. Then the following holds.
\begin{itemize}
\setlength{\itemsep}{-2pt}
\item[(i)] If $|X\cap Y|=1$, then $M^*(G)=M^*(H_1)\oplus_1M^*(H_2)$.
\item[(ii)] If $|X\cap Y|=2$, then $M^*(G)=M^*(H_1')\oplus_2M^*(H_2')$, where $H_i'$ is the graph obtained from $H_i$ by adding a new edge $e$ with its end vertices in the two vertices of $X\cap Y$ for $i=1,2$;  $E(H'_1)\cap E(H_2')=\{e\}$.
\item[(iii)] If $|X\cap Y|=3$ and $X\cap Y=\{v_1,v_2,v_3\}$, then $M^*(G)=M^*(H_1'')\oplus_2M^*(H_2'')$, where for $i=1,2$, $H_i''$ is the graph obtained from $H_i$ by adding a new vertex $v$ and edges $e_j=vv_j$ for $j\in\{1,2,3\}$;
 $E(H'_1)\cap E(H_2')=\{e_1,e_2,e_3\}$.
\end{itemize}
\end{lemma}

We use this lemma to decompose $M_s=M^*(G)$. Let $Y$ be the set of end vertices of $e_1,e_2,e_3$ in $V(H)$. The set $Y$ contains $y_1,y_2,y_3$, but some of these vertices could be the same. Let $X=(V(G)\setminus V(H))\cup Y$. We have that $(V(H),X)$ is a separation of $G$. We apply Lemma~\ref{lem:decomp-sub} to this separation. Recall that $Z$ is a clean cut of $G$. That means that no edge of $H$ is an element of a matroid that is a node of $\mathcal{T}$ distinct from $M_s$. Therefore, in this way we obtain a good $\{1,2,3\}$-decomposition with the conflict tree $\mathcal{T}'$ that is obtained form $\mathcal{T}$ by adding a leaf adjacent to $M_s$. Then we either reduce the new leaf if it is  a $1$-leaf or branch on it is $2$- or $3$-leaf. More formally, we do the following.
\begin{itemize}
\item If $|Y|=1$, then let $G'=G[X]$,  decompose $M^*(G)=M^*(G')\oplus_1M^*(H)$  
and construct a new conflict tree $\mathcal{T}'$ for the obtained decomposition of $M$:
we replace the node $M_s$ in $\mathcal{T}$ by $M^*(G')$ that remains adjacent to the same nodes as $M_s$ in $\mathcal{T}$ and then add  a new child  $M^*(H)$ of $M^*(G')$ that is a leaf of $\mathcal{T}'$.
Thus we can  apply 
Reduction Rule~\ref{rule:one-leaf} ({\bf $1$-Leaf reduction rule}) on  the new leaf.

\item If $|Y|=2$, then let $G'$ be the graph obtained from $G[X]$ by adding a new edge $e$ with its end vertices being the two vertices of $Y$. Furthermore, let $H'$  be the graph obtained from $H$ by adding a new edge $e$ with its end vertices being the two vertices of $Y$. Now 
decompose $M^*(G)=M^*(G')\oplus_2M^*(H')$ and consider a new conflict tree $\mathcal{T}'$ for the obtained decomposition:
 $M_s$ is replaced by $M^*(G')$ and a new leaf  $M^*(H')$ that is a child of $M^*(G')$ is added.
Notice that because $H$ has no bridges, no terminal $t\in T\cap E(H)$ is parallel to $e$ in $M^*(H')$. 
Thus we can  apply  Branching Rule~\ref{brule:2lb}  ({\bf $2$-Leaf branching})  on  the new leaf.

\item If $|Y|=3$, then $Y=\{y_1,y_2,y_3\}$.
Let $G'$ be the graph obtained from $G[X]$ by adding a new vertex $v$ and the edges $e_1'=y_1v$, $e_2'=y_2v$, $e_3'=y_3v$. 
Let 
$H'$ be the graph obtained from $H$ by adding a new vertex $v$ and the edges $e_1'=y_1v$, $e_2'=y_2v$, $e_3'=y_3v$. 
Then 
decompose $M^*(G)=M^*(G')\oplus_3M^*(H')$ and consider a new conflict tree $\mathcal{T}'$ for the obtained decomposition:
$M_s$ is replaced by $M^*(G')$ and a new leaf
 $M^*(H')$ that is a child of $M^*(G')$ is added.
Notice that because $H$ has no bridges, no terminal $t\in T\cap E(H)$ is parallel to $e_1',e_2',e_3'$ in $M^*(H')$. Thus we can  apply  Branching Rule~\ref{brule:3lb}  ({\bf $3$-Leaf branching})  on  the new leaf.
\end{itemize}

Lemma~\ref{lem:decomp-sub} together with Lemmas~\ref{lem:branch-2} and \ref{lem:branch-3} imply the correctness of the above procedure. Furthermore, all the reduction and branching rules can be performed in  
$2^{\cO(k)}\cdot ||M||^{\cO(1)}$ 
time.

\subsection{Proof of Theorem~\ref{thm:main}.}
Given an instance $(M,w,T,k)$ of \wss\, we either apply a reduction rule or a branching rule and if any of these applications (reduction rule or branching rule) returns \no, we return the same. Correctness of the answer follows from the correctness of the corresponding rules.  

Let $(M,w,T,k)$ be the given instance of \wss.  First, we exhaustively apply  elementary 
Reduction Rules~\ref{rule:zero-rule}-\ref{rule:stoprule}. Thus, by Lemma~\ref{lem:prepr}, in polynomial  time we either solve the problem or obtain an equivalent instance, where $M$ has no loops and the weights of nonterminal elements are positive.  To simplify notations, we also denote the reduced 
instance by  $(M,w,T,k)$. If $M$ is a \simple\ matroid  (obtained from $R_{10}$ by adding parallel elements or $M$ is graphic or cographic)  then we can
solve  \wss\ using Lemma~\ref{lem:simpleAlgo} in time 
$2^{\cO(k)}\cdot ||M||^{\cO(1)}$.

From now  onwards we assume that the matroid  $M$ in the instance $(M,w,T,k)$ is not \simple. 
Now using Corollary~\ref{thm:decomp-good},  we find a conflict tree $\mathcal{T}$. 
Recall that the set of nodes of $\mathcal{T}$ is the collection of basic matroids $\mathcal{M}$ and the edges correspond to $1$-, $2-$ and 3-sums.  The key observation is that $M$ can be constructed from $\mathcal{M}$ by performing the sums corresponding to the edges of $\mathcal{T}$ in an arbitrary order. Our algorithm is based on performing {\em bottom-up} traversal of the tree $\mathcal{T}$.  We select an arbitrarily {\em node $r$ as the root} of $\mathcal{T}$. Selection of $r$,  as the root,  defines the natural parent-child, descendant and ancestor relationship on the nodes of $\mathcal{T}$. We say that $u$ is a \emph{sub-leaf} if its children are leaves of $\mathcal{T}$. Observe that there always exists a sub-leaf in a tree on at least two nodes. Just take a node which is not a leaf and is  farthest from the root.  Clearly, this node can be found in polynomial time.  Rest of our argument is based on selection a sub-leaf $M_s$. 
We say that a child of $M_s$ is a $1$-, $2$- or $3$\emph{-leaf}, respectively, if the edge between $M_s$ and the leaf corresponds to $1$-, $2$- or 3-sum, respectively.   If there is a child $M_\ell$ of $M_s$ such that there is $e\in E(M_s)\cap E(M_\ell)$ that is parallel to a terminal $t\in E(M_\ell)\cap T$ in $M_\ell$, then we apply Reduction Rule~\ref{rule:term-flip-rule} ({\bf Terminal flipping rule}).  We apply Reduction Rule~\ref{rule:term-flip-rule} exhaustively. Correctness of this step follows from Lemma~\ref{lem:tfrulesafe}. 

From now we assume that there is no child $M_\ell$ of $M_s$ such that there exists an element 
$e\in E(M_s)\cap E(M_\ell)$ that is parallel to a terminal $t\in E(M_\ell)\cap T$ in $M_\ell$.  Now given 
a sub-leaf $M_s$ and a child $M_\ell$ of $M_s$, we apply the first rule (reduction or branching) among 
\begin{itemize}
\setlength{\itemsep}{-3pt}
\item Reduction Rule~\ref{rule:one-leaf} ({\bf $1$-Leaf reduction rule})
\item Reduction Rule~\ref{rule:two-leaf} ({\bf $2$-Leaf reduction rule})
\item Branching Rule~\ref{brule:2lb} ({\bf $2$-Leaf branching})
\item Branching Rule~\ref{brule:3lb} ({\bf $3$-Leaf branching})
\item Reduction Rule~\ref{rule:graphic3leafrule} ({\bf Graphic $3$-leaf reduction rule})
\item Reduction Rule~\ref{rule:cographic3rr} ({\bf Cographic $3$-leaf reduction rule})
\end{itemize}
which is applicable. If none of the above is applicable then we are in a specific subcase of $M_s$ being cographic matroid. That is, the case which is being handled in Section~\ref{sec:cslemptyset}. However, even in this case we modify our instance to fall into one of the cases above. 
Note that we we do not recompute the decompositions of the matroids obtained by the application of the rules but use the original decomposition modified by the rules. 
Observe additionally that the elementary Reduction Rules~\ref{rule:zero-rule}-\ref{rule:stoprule} also could be used to modify the decomposition.
Clearly, graphic and cographic remain graphic and cographic respectively and we just modify the corresponding graphs but we can delete or contract an element of a copy $R_{10}$. For this case, observe that Lemma~\ref{obs:r10} still could be applied and these matroids are not participating in 3-sums.
Each of the above rules reduces the $\cal T$ by one and thus these rules are only applied $\cO(|E(M)|))$ times. The correctness of algorithm follows from Lemmas~\ref{lem:red-1}, \ref{lem:red-2}, \ref{lem:branch-2}, \ref{lem:branch-3}, \ref{lem:red-3}  and \ref{lem:red-3-cograph}. The only thing that is remaining is the running time analysis. 

Either we apply reduction rules in polynomial time or in 
$2^{\cO(k)}\cdot ||M||^{\cO(1)}$
time. So all the reduction rules can be carried out in
$\cO(|E(M)|))\cdot 2^{\cO(k)}=2^{\cO(k)}\cdot ||M||^{\cO(1)}$ time. 
By 
Lemmas~\ref{lem:branch-2} and \ref{lem:branch-3} we know that when we apply Branching Rules~\ref{brule:2lb} and  \ref{brule:3lb} then the parameter reduces in each branch and thus the number of leaves in the search-tree is upper bounded by the recurrence, $T(k)\leq 15 T(k-1)$, corresponding to the 
Branching Rule~\ref{brule:3lb}. Thus, the number of nodes in the search tree is upper bounded by $15^k$ and since at each node we take 
$2^{\cO(k)}\cdot ||M||^{\cO(1)}$  
time, we have that the overall running time of the algorithm is upper bounded by 
$2^{\cO(k)}\cdot ||M||^{\cO(1)}$. 
This completes the proof.

\section{Reducing rank}\label{sec:rank-red}
In the well-known  \textsc{$h$-Way Cut} problem, we are given   a connected graph $G$ and positive integers $h$ and $k$, the task is to find at most $k$ edges whose removal increases the number of connected components by at least $h$. The problem has a  simple formulation in terms of matroids: Given a graph $G$ and an integers $k$, $h$, find $k$ elements of the graphical matroid of $G$ whose removal reduces its rank by at least $h$.  This motivated Joret and  Vetta
\cite{JoretV15} to introduce the \rr{}  problem on matroids. Here we define \rr{} on binary matroids. 
\defparproblem{\rr}%
{A binary matroid \mat{} given together with its matrix representation over GF(2) and two positive integers $h$ and $k$.}%
{$k$}
{Is there a set $X\subseteq E$ with $|X|\leq k$ such that $r(M)-r(M-X)\geq h$?}

\noindent
As a corollary of Theorem~\ref{thm:main}, we show that on regular matroids \rr{} is \classFPT{} for any fixed $h$.

We use the following lemma.

\begin{lemma}\label{lem:eq-rr}
 Let $M$ be a binary matroid and let $k\geq h$  be positive integers. Then $M$ has a set $X\subseteq E$ with $|X|\leq k$ such that $r(M)-r(M-X)\geq h$ if and only if there are disjoint sets $F,T\subseteq E$ such that $|T|=h$, $|F|\leq k-h$,  and $T\subseteq \spn(F)$ in $M^*$. 
 \end{lemma}
 
 \begin{proof}
 Notice that deletion of one element cannot decrease the  rank by more than one. Moreover,  deletion of $e\in E$ decreases the rank if and only if $e$ belongs to every basis of $M$. 
 Recall that  $e$ belongs to every basis of $M$ if and only if $e$ is a coloop (see \cite{Oxley11}). 
 It follows that $M$ has a set $X\subseteq E$ with $|X|\leq k$ such that $r(M)-r(M-X)\geq h$ if and only if 
 there are disjoint sets $F,T\subseteq E$ such that $|T|=h$, $|F|\leq k-h$  and every $e\in T$ is a coloop of $M-F$.
 Switching to the dual matroid, we rewrite this as follows: $M$ has a set $X\subseteq E$ with $|X|\leq k$ such that $r(M)-r(M-X)\geq h$ if and only if 
 there are disjoint sets $F,T\subseteq E$ such that $|T|=h$, $|F|\leq k-h$  and every $e\in T$ is a loop of $M^*/F$.   It remains to observe that every $e\in T$ is a loop of $M^*/F$ if and only if $T\subseteq\spn(F)$ in $M^*$.
  \end{proof}
  
 For graphic matroids, when \rr{} is equivalent to  \textsc{$h$-Way Cut}, the problem is  \classFPT{} parameterized by $k$ even if $h$ is a part of the input ~\cite{KawarabayashiT11}. Unfortunately, similar result does not hold for cographic matroids.

 \begin{proposition}\label{prop:w-hard-rr}
 \rr{} is \classW{1}-hard for cographic matroids   parameterized by $h+k$.
 \end{proposition}
 
 \begin{proof}
 Consider the bond matroid  $M^*(G)$ of a simple graph $G$. By Lemma~\ref{lem:eq-rr}, $(M^*(G),h,k)$ is a yes-instance of \rr{} if and only  if  there are disjoint sets of edges $F,T\subseteq E(G)$ such that $|T|=h$ and $|F|\leq k-h$  and $T\subseteq \spn(F)$ in $M(G)$.  Recall that $T\subseteq \spn(F)$ in $M(G)$ if and only if for every $uv\in T$, $G[F]$ has a $(u,v)$-path.   Let $p\geq 3$ be an integer, $k=(p-1)p/2$ and  
  $h=(p-1)(p-2)/2$.  It is easy to see that for this choice of $h$ and $k$, $G$ has  disjoint sets of edges $F,T\subseteq E(G)$ such that $|T|=h$, $|F|\leq k-h$  and for every $uv\in T$, $G[F]$ has a $(u,v)$-path if and only if $G$ has a clique with $p$ vertices. Since it is well-know that it is \classW{1}-complete with the parameter $p$ to decide whether a graph $G$ has a clique of size $p$ (see~\cite{DowneyF13}), we conclude that \rr{} is \classW{1}-hard when parameterized by $h+k$.
\end{proof}

However, by Theorem~\ref{thm:main},  for fixed  $h$,  \rr{}  is \classFPT{} parameterized by $k$ on regular matroids.
 
\begin{theorem}\label{prop:rr-fpt}
\rr{} can be solved in time $2^{\Oh(k)}\cdot ||M||^{\Oh(h)}$ on regular matroids.
\end{theorem}

\begin{proof}
Let $(M,h,k)$ be an instance of \rr. By Lemma~\ref{lem:eq-rr}, $(M,h,k)$ is a yes-instance if and only if here are disjoint sets $F,T\subseteq E$ such that $|T|=h$, $|F|\leq k-h$  and $T\subseteq \spn(F)$ in $M^*$. There are at most $||M||^h$ possibilities to choose $T$. For each choice, we check 
whether there is $F\subseteq E\setminus T$ such that $|F|\leq k-h$ and $T\subseteq \spn(F)$ in $M^*$. By Theorem~\ref{thm:main}, it can be done in time $2^{\Oh(k)}\cdot ||M||^{\Oh(1)}$. Then the total running time is $2^{\Oh(k)}\cdot ||M||^{\Oh(h)}$.
\end{proof}

\section{Conclusion}\label{sec:conclusion}
In this paper, we used the structural theorem  of Seymour for designing parameterized algorithm for \SSp. 
   While  structural graph theory and  graph decompositions serve as the most usable tools in the design of parameterized algorithms, the applications of structural matroid theory  in parameterized algorithms are limited. There is a series of papers about width-measures and decompositions of matroids (see, in particular,~\cite{Hlineny06,HlinenyO08,Jeong0O16,Kral12,OumS06,OumS07} and the bibliography therein) but, apart of this specific area, we are not aware of other applications  except the works Gavenciak et al.  \cite{GavenciakKO12} and our recent work  \cite{FominGLS17}.  
   In spite of the tremendous progress in understanding the structure of matroids representable over  finite fields \cite{GeelenGW02,GeelenGW07,GeelenB14,GeelenGW15}, 
   this rich research area still remains to be explored from the perspective of parameterized complexity. 
    
    As a concrete open problem, what about the parameterized complexity of \SSp{} on 
    any proper minor-closed class   of binary matroids?

\end{document}